\documentclass[reqno]{amsart}
\usepackage[margin=1.2in]{geometry}
\usepackage{enumerate}
\usepackage{enumitem}
\usepackage{amsthm,amsmath,amsfonts,amssymb,amsxtra,appendix,bookmark,euscript,delarray,dsfont,mathrsfs}
\usepackage{enumerate}
\usepackage{amssymb}
\usepackage{xcolor}
\colorlet{RED}{red}
\usepackage{graphicx}
\usepackage{upgreek}
\usepackage[normalem]{ulem}
\usepackage{ stmaryrd }
\usepackage{array}
\usepackage{longtable}
\newcolumntype{L}[1]{>{\raggedright\arraybackslash}p{#1}}
\usepackage{url}
\usepackage{needspace}
\usepackage{verbatim}

\allowdisplaybreaks[1]

\usepackage{tikz}
\usetikzlibrary{decorations.pathreplacing, calc}

\theoremstyle{plain}
\newtheorem{theorem}{Theorem}[section]
\newtheorem{lemma}[theorem]{Lemma}
\newtheorem{remark}[theorem]{Remark}
\newtheorem{definition}[theorem]{Definition}
\newtheorem{proposition}[theorem]{Proposition}

\newtheorem{corollary}[theorem]{Corollary}
\newtheorem{assumption}{Assumption}[section]

\numberwithin{equation}{section}

\renewcommand{\d}{\mathrm{d}}
\newcommand{\dd}{\,\mathrm{d}}
\newcommand{\ii}{\mathrm{i}}

\DeclareMathOperator{\tr}{Tr}
\DeclareMathOperator{\Law}{Law}
\DeclareMathOperator{\Ran}{Ran}

\renewcommand{\leq}{\leqslant}
\renewcommand{\geq}{\geqslant}

\def\eps{\varepsilon}
\def\nn{\nonumber}
\renewcommand{\to}{\rightarrow}
\renewcommand{\phi}{\varphi}

\newcommand{\norm}[1]{\left\lVert #1 \right\rVert}
\newcommand{\abs}[1]{\left\lvert#1\right\rvert}
\renewcommand{\phi}{\varphi}
\newcommand{\ket}[1]{|#1 \rangle}
\newcommand{\bra}[1]{\langle #1 |}
\newcommand		{\lt}			{\left}
\newcommand		{\rt}			{\right}
\newcommand		{\bangle}[1]	{\lt\langle #1\rt\rangle}
\newcommand		{\inprod}[2]	{\bangle{#1, #2}}

\newcommand\R{{\ensuremath {\mathbb R} }}
\newcommand{\CC}{{\ensuremath {\mathbb C} }}
\newcommand\Z{{\ensuremath {\mathbb Z} }}
\newcommand\1{{\ensuremath {\mathds 1} }}
\newcommand{\gH}{\mathfrak{H}}
\newcommand{\dG}{\d\Gamma}
\newcommand{\gX}{\mathfrak{X}}
\newcommand{\gF}{\mathfrak{F}}
\newcommand{\bH}{\mathbb{H}}
\newcommand{\bW}{\mathbb{W}}
\newcommand{\gS}{\mathfrak{S}}
\newcommand{\cN}{\mathcal{N}}

\newcommand{\cH}{\mathcal{H}}
\newcommand{\cM}{\mathcal{M}}

\newcommand{\cW}{\mathcal{W}}
\newcommand{\bT}{\mathbb{T}}

\newcommand{\cl}{\mathrm{cl}}
\newcommand{\ren}{\mathrm{ren}}

\colorlet{red}{black}

\title{An upper-symbol approach to $\Phi^4_d$ limits of Bose gases}
\author[H. Liang]{Hao Liang}
\address{School of Mathematical Sciences, Peking University, Beijing, 100871, China}
\email{leunghao@stu.pku.edu.cn}
\date{}

\begin{document}
\begin{abstract}
We study the high-temperature derivation of $\Phi^4_d$ measures from grand-canonical Bose gases on two- and three-dimensional tori. Building on the upper-symbol variational approach of the author's earlier work \cite{Cub26}, we use relative entropy to give a simple, unified proof of relative free-energy convergence and Hilbert--Schmidt convergence of all fixed-order rescaled reduced density matrices to their Hartree counterparts. For interaction ranges $\varepsilon=\lambda^\eta$, where $\lambda$ is the inverse temperature, these results hold for every fixed $0<\eta<1/2$ in two dimensions and $0<\eta<1/26$ in three. In two dimensions, we improve the result of Jougla and Rougerie \cite{JR26} and the local $\Phi^4_2$ derivation of Fr\"ohlich, Knowles, Schlein and Sohinger \cite{FKSS23}. Combined with classical approximation, the two-dimensional result gives the first local $\Phi^4_2$ derivation valid for every fixed exponent $0<\eta<1/2$, allowing polynomial exponents arbitrarily close to the diluteness threshold from below. \textcolor{red}{In three dimensions, we use the identification of the local measure and the stationary SPDE estimates of Nam, R.~Zhu and X.~Zhu \cite{NZZ25}. From these estimates we strengthen their classical correlation convergence to the Hilbert--Schmidt norm. This yields the local limit at every fixed order in the same polynomial range.}
\end{abstract}
\maketitle

\tableofcontents
\enlargethispage{4pt}

\section{Introduction}

We study the $\Phi^4_d$ theory of a complex scalar field on a domain $\Lambda\subset\R^d$. Formally, the theory is described by the probability measure
\begin{equation}\label{eq:formal-Gibbs}
 \dd\nu(u)=\frac1c e^{-S(u)}\,\mathrm D u
\end{equation}
on fields $u:\Lambda\to\CC$, where $\mathrm D u=\prod_{x\in\Lambda}\dd u(x)$ denotes the formal Lebesgue measure and
\begin{equation}\label{eq:formal-action}
 S(u)
 =
 \int_\Lambda
 \left(|\nabla u(x)|^2+m_0|u(x)|^2\right)\dd x
 +
 \frac12\int_\Lambda |u(x)|^4\dd x .
\end{equation}
Here $m_0\in\R$ is the mass parameter, and $|\cdot|$ denotes the Euclidean norm. The expression \eqref{eq:formal-Gibbs} is formal, and a rigorous construction requires renormalization. Throughout this paper, $d\in\{2,3\}$ and $\Lambda=\bT^d=(\R/2\pi\Z)^d$.

Nonlinear Gibbs measures arise naturally in constructive quantum field theory, see, for example \cite{nelson1973construction,nelson2007probability,guerra1975p,simon2015p}. In suitable settings, nonlinear Gibbs measures are invariant under the flows of nonlinear Schrödinger equations \cite{bourgain1994periodic,bourgain1996invariant,bourgain1997invariant,bourgain2000invariant}.  They have also been used to construct global solutions with rough random initial data; see, for example, \cite{burq2008randomA,burq2008randomB,bourgain2014almostA,bourgain2014almostB,bringmann2024invariant,deng2012two,deng2021invariant,deng2024invariant}.

Our aim is to derive these field measures from grand-canonical quantum Gibbs states of an interacting Bose gas on $\bT^d$. This program began with \cite{LNR15} and has since developed across different dimensions, interactions and confining geometries \cite{frohlich2017gibbs,lewin2018gibbs,LNR21,frohlich2022mean,sohinger2022microscopic}. Subsequent advances include local and singular field limits \cite{FKSS23,NZZ25,JR26,caraci2026euclidean,nam2026derivation,NYZ26}, focusing nonlinearities \cite{lu2026derivation,farhatPerturbativeMicroscopicDerivation2026}, and nonlocal three-body interactions \cite{Cub26}.

Related developments include the derivation of fixed-mass Gibbs measures from one-dimensional canonical ensembles \cite{dinh2024bosonic} and of time-dependent correlation functions for the one-dimensional NLS \cite{frohlich2019microscopic} and the renormalized Hartree NLS in two and three dimensions \cite{nam2026dynamical}. Variational, perturbative, functional-integral and stochastic methods provide complementary approaches to these limits. Let $ \gF=\bigoplus_{n\geq0}L^2_s((\bT^d)^n) $ be the bosonic Fock space, where $L^2_s$ denotes the subspace of symmetric functions. We consider Gibbs states on the bosonic Fock space proportional to $e^{-\bH_\lambda}$, with the Hamiltonian
\begin{align}
 \bH_\lambda
 &=
 0\oplus\lambda(-\Delta-\vartheta)
 \oplus
 \bigoplus_{n\geq2}
 \Big(
 \lambda\sum_{j=1}^n(-\Delta_{x_j}-\vartheta)
 +
 \lambda^2\sum_{1\leq i<j\leq n}
 w_\varepsilon(x_i-x_j)
 \Big)\nn\\
 &=
 \lambda\int_{\bT^d}
 a_x^*(-\Delta_x-\vartheta)a_x\dd x
 +
 \frac{\lambda^2}{2}
 \iint_{(\bT^d)^2}
 w_\varepsilon(x-y)a_x^*a_y^*a_xa_y\dd x\,\dd y .
 \label{eq:quantum-Hamiltonian-chemical-potential}
\end{align}
Here the parameter $\vartheta=\vartheta_{\lambda,\varepsilon}\in\R$ is the chemical potential. The real, even potential $w_\varepsilon$ is a regularized approximation of the Dirac delta function, with
$$
 \int_{\bT^d}w_\varepsilon(x)\dd x=1,
 \qquad
 w_\varepsilon(0)<\infty.
$$
We let $\lambda\to0^+$ and choose $\varepsilon=\varepsilon(\lambda)\to0^+$ so that $w_\varepsilon\to\delta_0$ in the sense of distributions.

The scaled operators $\sqrt\lambda\,a_x$ and $\sqrt\lambda\,a_x^*$ are the quantum counterparts of $u(x)$ and $\overline{u(x)}$. Their canonical commutation relations (CCR) are
\begin{equation}
 [\sqrt\lambda\,a_x,\sqrt\lambda\,a_y]
 =
 [\sqrt\lambda\,a_x^*,\sqrt\lambda\,a_y^*]=0,\quad
 [\sqrt\lambda\,a_x,\sqrt\lambda\,a_y^*]
 =
 \lambda\delta(x-y)\longrightarrow0.
 \label{eq:scaled-CCR}
\end{equation}
The relations \eqref{eq:scaled-CCR} suggest a classical field limit in which the Hamiltonian is represented by the action $S$ in \eqref{eq:formal-action}, and the quantum Gibbs states converge to the corresponding field measure, with the chemical potential and scalar counterterms chosen consistently with the required renormalization.

Before formulating the derivation problem, we explain the construction of the classical measure. To illustrate the renormalization, take $m_0=1$. The kinetic and mass terms are absorbed into the centered complex Gaussian measure $\mu_0$ with covariance $(-\Delta+\1)^{-1}$, formally written as
$$
 \dd\mu_0(u)
 =
 \frac1{\widetilde c}
 \exp\left[
 -\int_{\bT^d}
 \left(|\nabla u(x)|^2+|u(x)|^2\right)\dd x
 \right]\mathrm D u .
$$
The remaining task is to give a meaning to the nonlinear part. In two dimensions, Wick renormalization yields a Gibbs density with respect to $\mu_0$. In three dimensions, the local $\Phi^4_3$ measure is singular with respect to $\mu_0$, and Wick renormalization alone is insufficient.

At each positive interaction range $\varepsilon$, however, the corresponding Hartree Gibbs measure can be constructed as a density with respect to $\mu_0$ in both dimensions. This provides a natural classical counterpart to the quantum gas with interaction $w_\varepsilon$. We consider the renormalized Hartree interaction
$$
 \cW_{0,\varepsilon}[u]
 =
 \frac12\iint_{(\bT^d)^2}
 \big(:|u(x)|^2:\,w_\varepsilon(x-y)\,:|u(y)|^2:\big)\d x\d y
 -
 \alpha_{0,\varepsilon}
 \int_{\bT^d}:|u(x)|^2:\d x
 -
 \beta_{0,\varepsilon},
$$
where the mass and scalar counterterms $\alpha_{0,\varepsilon},\beta_{0,\varepsilon}$ are chosen according to the dimension. The colons denote Wick ordering with respect to $\mu_0$ and apply separately to the two density factors. Formally,
$$
 :|u(x)|^2:
 =
 |u(x)|^2
 -
 \left\langle|u(x)|^2\right\rangle_{\mu_0}.
$$
Under suitable assumptions specified later,
$$
 z_{0,\varepsilon}
 =
 \int e^{-\cW_{0,\varepsilon}[u]}\dd\mu_0(u),
 \qquad
 \dd\nu_\varepsilon(u)
 =
 \frac{1}{z_{0,\varepsilon}}
 e^{-\cW_{0,\varepsilon}[u]}\dd\mu_0(u)
$$
defines a probability measure, and $z_{0,\eps}$ is called the classical partition function. With the appropriate counterterms, these measures approximate the local $\Phi^4_d$ measure as $\varepsilon\to0$.

\textcolor{red}{ In three dimensions, Nam, R.~Zhu and X.~Zhu \cite{NZZ25} study the nonlocal stochastic-quantization equation associated with the Hartree measures. Their paracontrolled analysis determines the mass renormalization beyond Wick ordering, including the logarithmically divergent counterterm and the finite correction to the limiting mass. They prove weak convergence of the Hartree measures and distributional convergence of their correlations to the local $\Phi^4_3$ measure \cite[Theorem~2.6]{NZZ25}. We use their construction and identification of this measure. Their stationary estimates also allow us to obtain correlation bounds, from which we prove Hilbert--Schmidt convergence of the Hartree correlations in Section~\ref{sec:derivation-Phi4d}.}

Guided by the quantum--classical correspondence, we center the quantum density with respect to the free gas and define
$$
 \bW_{\lambda,\varepsilon}^{\ren}
 =
 \frac{\lambda^2}{2}
 \iint_{(\bT^d)^2}
 (a_x^*a_x-\rho_0)
 w_\varepsilon(x-y)
 (a_y^*a_y-\rho_0)
 \dd x\d y
 -
 \alpha_{\lambda,\varepsilon}\lambda
 \int_{\bT^d}(a_x^*a_x-\rho_0)\dd x
 -
 \beta_{\lambda,\varepsilon},
$$
where $\rho_0=\rho_0(\lambda)$ is the free quantum particle density,
\begin{equation}\label{eq:free-particle-density}
 \rho_0
 =
 \frac1{(2\pi)^d}
 \sum_{n\in\Z^d}
 \frac1{e^{\lambda(1+|n|^2)}-1}.
\end{equation}
The quantum counterterms $\alpha_{\lambda,\varepsilon},\beta_{\lambda,\varepsilon}$ are chosen in accordance with their classical counterparts.

Using the CCR \eqref{eq:scaled-CCR}, we can write the Hamiltonian \eqref{eq:quantum-Hamiltonian-chemical-potential} as $ \bH_\lambda=\lambda\dG(h)+\bW^{{\rm ren}}_\lambda $, up to a constant shift, by choosing the chemical potential
$$
 \vartheta_{\lambda,\varepsilon}
 =
 \alpha_{\lambda,\varepsilon}-1
 +\lambda\rho_0
 -\frac{\lambda}{2}w_\varepsilon(0)
$$
Here $h=-\Delta+\1$ and $\dG(h)=\int_{\bT^d}\d x\, a_x^* h_x a_x$. We define the Gibbs state and its partition function by
$$
 Z_{\lambda,\varepsilon}
 :=
 \tr_{\gF}\big(
 e^{-\lambda\dG(h)-\bW_{\lambda,\varepsilon}^{\ren}}\big),\qquad
 \Gamma_{\lambda,\varepsilon}
 :=
 \frac{1}{Z_{\lambda,\varepsilon}}
 e^{-\lambda\dG(h)-\bW_{\lambda,\varepsilon}^{\ren}}.
$$

We compare $\Gamma_{\lambda,\varepsilon}$ with the Hartree measure $\nu_\varepsilon$ at the same interaction range. The comparison concerns both the relative free energy and the reduced density matrices. We write
$$
 Z_{\rm free}=\tr_{\gF}\big(e^{-\lambda\dG(h)}\big),
$$
for the free partition function and use the convention
$$
 \Gamma_{\lambda,\varepsilon}^{(k)}
 =
 \sum_{n\geq k}\binom nk
 \tr_{k+1,\ldots,n}(\Gamma_{\lambda,\varepsilon,n}),\qquad
 \gamma_{\nu_\varepsilon}^{(k)}
 :=
 \int
 |u^{\otimes k}\rangle\langle u^{\otimes k}|
 \dd\nu_\varepsilon(u),
$$
where $\Gamma_{\lambda,\varepsilon,n}$ is the restriction to the $n$-particle sector. 

\textcolor{red}{We first compare the quantum Gibbs states with the Hartree measures. Under the dimension-dependent assumptions and counterterms specified below, we prove, for $\varepsilon=\varepsilon(\lambda)\to0$, convergence of the relative free energy}
$$
 \boxed{
 \Big|
 \log\frac{Z_{\lambda,\varepsilon}}{Z_{\rm free}}
 -
 \log z_{0,\varepsilon}
 \Big|
 \longrightarrow0,}
$$
and of the reduced density matrices in the Hilbert--Schmidt norm
$$
 \boxed{
 \Big\|k!\lambda^k\Gamma_{\lambda,\varepsilon}^{(k)}
 -\gamma_{\nu_\varepsilon}^{(k)}\Big\|_{\gS^2}
 \longrightarrow0
 \qquad\text{for every fixed }k\geq1.}
$$
\textcolor{red}{The local limits are then obtained in Corollaries~\ref{cor:Phi42-local} and~\ref{cor:Phi43-local-HS}: the rescaled reduced density matrices converge in the Hilbert--Schmidt norm to the correlation operators of the local $\Phi^4_d$ measures. In dimension two, we also obtain the corresponding local relative free-energy limit.}

A natural question is how fast $\varepsilon$ can shrink as $\lambda\to0$ while preserving these limits. A natural goal is to approach the dilute regime, where the interaction range is much smaller than the typical interparticle distance \cite[(1.11)]{JR26}. In the present high-temperature scaling, the free density $\rho_0$ in \eqref{eq:free-particle-density} satisfies
$$
 \rho_0^{-1/d}\asymp
 \begin{cases}
 \lambda^{1/2}(\log(1/\lambda))^{-1/2},&d=2,\\
 \lambda^{1/2},&d=3,
 \end{cases}
$$
by \cite[Lemma~B.1]{LNR21}. For fixed power laws $\varepsilon=\lambda^\eta$, diluteness at this density scale, $\varepsilon\ll\rho_0^{-1/d}$, corresponds to $\eta>1/2$. The threshold exponent is thus the same in both dimensions, with a logarithmic correction to the reference length in two dimensions.

The first derivation of the local $\Phi^4_2$ measure from an interacting Bose gas \cite{FKSS23} established convergence of the relative partition function and correlation functions for
$$
 \varepsilon\geq
 \exp\big[-(\log(1/\lambda))^{1/2-c}\big],
 \qquad 0<c<1/2.
$$
This requires $\varepsilon$ to shrink more slowly than every fixed positive power of $\lambda$.

\textcolor{red}{The derivation of $\Phi^4_3$ in \cite{NZZ25} allows polynomial ranges $\varepsilon\geq\lambda^\eta$ for sufficiently small fixed $\eta>0$. In this regime, their quantum-to-Hartree comparison gives relative free-energy convergence and convergence of the first six reduced density matrices against finite Fourier tests \cite[Theorem~2.8]{NZZ25}. Under a stronger logarithmic restriction on the interaction range, their correlation convergence holds at every fixed order. Their two-dimensional argument also covers every fixed order. For the two-dimensional relative free energy, \cite{JR26} establishes the explicit range $\varepsilon=\lambda^\eta$, $0<\eta<1/24$, by refining the variational method of \cite{LNR21}.}

The trapped two-dimensional problem was treated in \cite{caraci2026euclidean} under a stretched-logarithmic range condition. More recently, \cite{NYZ26} derived general defocusing $\mathcal P(\Phi)_2$ measures, including the quartic case, for $\varepsilon\geq\lambda^\eta$ with sufficiently small $\eta>0$, proving relative free-energy convergence and Hilbert--Schmidt convergence of all fixed-order reduced density matrices.

\textcolor{red}{We extend the upper-symbol variational approach of the author's earlier work \cite{Cub26} to shrinking interaction ranges. Together with relative entropy and kernel estimates, this gives a simple, unified proof of the quantum-to-Hartree limits in Theorems~\ref{thm:Phi42} and~\ref{thm:Phi43}, for every fixed $0<\eta<1/2$ in two dimensions and $0<\eta<1/26$ in three. The same ranges apply to the local correlation limits in Corollaries~\ref{cor:Phi42-local} and~\ref{cor:Phi43-local-HS}. The two-dimensional result allows exponents arbitrarily close to the diluteness threshold from below. In three dimensions, we obtain Hilbert--Schmidt convergence at every fixed order in both the quantum-to-Hartree comparison and the subsequent Hartree-to-$\Phi^4_3$ limit. The latter uses the classical results and stationary estimates of \cite{NZZ25}. We retain the exponent $1/26$ for a unified treatment of both dimensions, without optimizing the three-dimensional estimates.}

\textcolor{red}{ For reduced density matrices, our kernel comparison and weighted correlation estimates reduce the quantum a priori bounds to uniform bounds on $\|\gamma_{\nu_\varepsilon}^{(k)}\|_{\gS^2}$ for each fixed $k$. In two dimensions, the classical approximation gives uniform $L^2(\mu_0)$ bounds on the Radon--Nikodym derivatives $\dd\nu_\varepsilon/\dd\mu_0$, and the Gaussian correlation estimate gives the required operator bounds. In three dimensions, the local measure is singular with respect to $\mu_0$, so $\|\dd\nu_\varepsilon/\dd\mu_0\|_{L^p(\mu_0)}\to\infty$ for every fixed $p>1$. Using only our elementary density bounds, with $p$ approaching $1$, gives losses in $\varepsilon$ that increase with the correlation order; for a fixed polynomial scaling, this estimate therefore yields convergence only up to a finite order. Instead, we use the stationary coupling and remainder estimates of \cite{NZZ25} to prove the correlation bounds in Proposition~\ref{prop:classical-Hartree-correlations}. These estimates serve two purposes: the case $\delta=0$ provides the bounds needed for Theorem~\ref{thm:Phi43}, while $\delta>0$ gives uniform Fourier tails and strengthens the classical correlation convergence to the Hilbert--Schmidt norm.}

\medskip

\subsection{Organization of the Paper.}

The remainder of this paper is organized as follows. In Section~\ref{sec:setup-results}, we introduce the classical and quantum models and state our main results, Theorems~\ref{thm:Phi42} and~\ref{thm:Phi43}, together with the local limits in Corollaries~\ref{cor:Phi42-local} and~\ref{cor:Phi43-local-HS}. The two theorems are derived from the quantitative comparison in Theorem~\ref{thm:abstract-comparison}, which is also presented in that section.

Section~\ref{sec:proof-strategy} introduces the common probabilistic framework and outlines the main steps in the proof of Theorem~\ref{thm:abstract-comparison}. These steps are carried out in the subsequent sections. We first establish the entropy identities in Section~\ref{sec:entropy-identities}, and then estimate the fluctuations and Gibbs weights in Section~\ref{sec:free-energy-estimates} to obtain relative free-energy convergence. For the convergence of reduced density matrices, Section~\ref{sec:correlation-preliminaries} develops the weighted correlation estimates and a priori bounds used in the argument. The proof is completed in Section~\ref{sec:kernel-domination}, where we establish the quantum kernel bounds and combine them with the preceding estimates. Finally, Section~\ref{sec:derivation-Phi4d} combines the quantum-to-Hartree limits with classical approximation to prove Corollaries~\ref{cor:Phi42-local} and~\ref{cor:Phi43-local-HS}.

\subsection*{Acknowledgments} \textcolor{red}{The author thanks Zhenfu Wang for insightful discussions on relative entropy, Zhilin Yang for helpful discussions, and Rongchan Zhu and Xiangchan Zhu for their valuable comments on SPDEs.} This work was partially supported by the National Key R\&D Program of China (Project No.~2024YFA1015500), and the author was partially supported by the NSFC (Grant Nos.~12595282 and 12171009). 

\medskip

\paragraph{AI disclosure.}
The author used ChatGPT-6 Astra to help develop the proof of the uniform-in-$\varepsilon$ Hilbert--Schmidt bounds for the Hartree correlation operators. All other mathematical ideas originated with the author or arose from discussions with colleagues. AI tools were also used for language editing, checking calculations, literature searches and bibliographic checks. All mathematical statements, proofs and calculations were independently checked and verified by the author, who takes full responsibility for its content.

\bigskip

\section{Setup and results}\label{sec:setup-results}

\subsection{Classical field theory}\label{subsec:classical-field-theory}

We first introduce Gaussian fields and their correlation operators. Let $\gH=L^2(\bT^d;\CC)$, $h=\1-\Delta$, and $e_n(x)=(2\pi)^{-d/2}e^{\ii n\cdot x}$. Inner products are conjugate-linear in the first argument. For $s\in\R$, we define the Sobolev space
$$
 \gH^s=\Big\{
 u=\sum_{k\in\mathbb Z^d}u_ke_k:
 \sum_{k\in\mathbb Z^d}
 (1+|k|^2)^s|u_k|^2<\infty
 \Big\}.
$$
Pairings between $\gH^s$ and $\gH^{-s}$ are understood by duality.

Let $\mathsf C$ be a bounded nonnegative self-adjoint operator on $\gH$. Choose $s>0$ such that
$$
 \tr_{\gH}(h^{-s/2}\mathsf C h^{-s/2})<\infty.
$$
Such a choice is always possible, since, for $s>d/2$,
$$
 \tr_{\gH}(h^{-s/2}\mathsf C h^{-s/2})
 \leq
 \norm{\mathsf C}\sum_{n\in\Z^d}(1+|n|^2)^{-s}
 <\infty.
$$
Let $(\gX,\mathcal A,\mathbb P)$ be a probability space carrying independent standard circular complex Gaussian variables $g=(g_n)_{n\in\Z^d}$, each with density $\pi^{-1}e^{-|z|^2}\dd\Re z\,\dd\Im z$. In this subsection, $\mathbb E$ denotes expectation with respect to $\mathbb P$. We define the Gaussian field with covariance $\mathsf C$ by
$$
 u_{\mathsf C}
 =
 \sum_{n\in\Z^d}g_n\mathsf C^{1/2}e_n.
$$
The series converges in $L^2(\mathbb P;\gH^{-s})$. Its law, $\mu_{\mathsf C}:=\Law(u_{\mathsf C})$, is the Gaussian measure with covariance $\mathsf C$. For $\phi,\psi\in\gH^s$,
\begin{equation}\label{eq:classical-Gaussian-covariances}
 \mathbb E\left[
 \inprod{\phi}{u_{\mathsf C}}
 \inprod{u_{\mathsf C}}{\psi}
 \right]
 =
 \inprod{\phi}{\mathsf C\psi},\qquad
 \mathbb E\left[
 \inprod{\phi}{u_{\mathsf C}}
 \inprod{\psi}{u_{\mathsf C}}
 \right]
 =0.
\end{equation}
The identities \eqref{eq:classical-Gaussian-covariances} uniquely determine $\mu_{\mathsf C}$ among centered Gaussian probability measures on $\gH^{-s}$. Moreover,
\begin{equation}\label{eq:Gaussian-field-correlations}
 \int_{\gH^{-s}}\norm u_{\gH^{-s}}^2\dd\mu_{\mathsf C}(u)
 =
 \tr_{\gH}(h^{-s/2}\mathsf C h^{-s/2}),\quad
 \int_{\mathfrak H^{-s}}
 |u^{\otimes k}\rangle\langle u^{\otimes k}|
 \,\dd\mu_{\mathsf C}(u)
 =
 k!\,\mathsf C^{\otimes k},
 \qquad k\geq1.
\end{equation}
The second identity is understood weakly on $\bigotimes_{\mathrm{sym}}^k\gH$.

Let $\mu_0$ be the Gaussian measure with covariance $h^{-1}$. From now on, we take $d\in\{2,3\}$. Since $h^{-1}$ is Hilbert--Schmidt, $\mu_0(\gH^{-1})=1$.

To define the interacting theory, we first introduce mass renormalization. For a bounded normal operator $A$ on $\gH$, set
\begin{equation}\label{eq:classical-mass-cutoff}
 \cM_{0,K}^A[u]
 =
 \inprod{P_Ku}{AP_Ku}
 -
 \int\inprod{P_Ku}{AP_Ku}\dd\mu_0(u),
 \qquad
 P_K=\mathds1_{\{h\leq K\}}.
\end{equation}
The sequence in \eqref{eq:classical-mass-cutoff} converges in $L^p(\mu_0)$ for every $1\leq p<\infty$ (see Lemma~\ref{lem:Gaussian-centered-quadratic-moments}). We denote its limit by $\mathcal M_0^A$. Throughout the paper, the two-body interaction $w=w(x-y)$ satisfies the following conditions:
$$
 w(x)=\frac1{(2\pi)^d}
 \sum_{k\in\Z^d}\widehat w(k)e^{\ii k\cdot x},
 \qquad
 \widehat w(k)=\widehat w(-k)\geq0,\qquad
 w(0)<\infty,
 \qquad
 \widehat w(0)>0.
$$

We define the renormalized interaction by
\begin{equation}
 \cW_0[u]
 =
 \frac12\sum_{k\in\mathbb Z^d}
 \widehat w(k)|\cM_0^{e_k}[u]|^2
 -
 \alpha_0\cM_0^\1[u]
 -
 \beta_0,
 \label{eq:abstract-classical-W0}
\end{equation}
where $\alpha_0,\beta_0\in\R$, and $e_k$ is viewed as a multiplication operator. By Minkowski's inequality and Lemma~\ref{lem:Gaussian-centered-quadratic-moments}, the series in \eqref{eq:abstract-classical-W0} converges in $L^1(\mu_0)$. Thus $\mathcal W_0\in L^1(\mu_0)$ is well defined. The momentum-space expression \eqref{eq:abstract-classical-W0} has the equivalent position-space representation
$$
 \cW_0[u]
 =
 \frac12\iint_{(\bT^d)^2}
 \big(:|u(x)|^2:\,w(x-y)\,:|u(y)|^2:\big)
 \d x\,\d y-
 \alpha_0\int_{\bT^d}:|u(x)|^2:\d x
 -
 \beta_0.
$$

The interacting field theory is described by the Gibbs measure
$$
 \dd\nu_0(u)
 :=
 \frac1{z_0}e^{-\cW_0[u]}\dd\mu_0(u),
 \qquad
 z_0
 :=
 \int e^{-\cW_0[u]}\dd\mu_0(u).
$$
We have $e^{-\cW_0}\in L^1(\mu_0)$ and $z_0\in(0,\infty)$; see the proof of Lemma~\ref{lem:Lp-estimate-three-weights}. Hence $\nu_0$ is a well-defined probability measure.

The correlation operators describe the interacting field theory. For $k\in\mathbb N$, define
\begin{equation}\label{eq:abstract-nu-0-corroration-operator}
 \gamma_{\nu_0}^{(k)}
 :=
 \int
 \ket{u^{\otimes k}}\bra{u^{\otimes k}}
 \dd\nu_0(u),
\end{equation}
where the integral is understood in terms of matrix elements, as in the Gaussian correlation identity\kern0pt{} \eqref{eq:Gaussian-field-correlations}. This form has a Hilbert--Schmidt realization on the physical $k$-particle space: $ \gamma_{\nu_0}^{(k)} \in\gS^2\!\big(\bigotimes_{\mathrm{sym}}^k\gH\big). $ See Section~\ref{sec:correlation-preliminaries}.

\medskip

\subsection{Many-body framework}

We now introduce quantum states and their reduced density matrices. For the grand-canonical ensemble, we use the bosonic Fock space $\gF$, defined by
$$
 \gF:=\bigoplus_{n\geq0}\bigotimes_{\mathrm{sym}}^n\gH,
 \qquad
 \bigotimes_{\mathrm{sym}}^0\gH:=\CC.
$$
For $f\in \gH=L^2(\bT^d)$, the annihilation and creation operators $a(f)$ and $a^*(f)$ are defined on the finite-particle subspace by
\begin{equation}\label{eq:annihilation-creation-operators}
 \begin{aligned}
  (a(f)\psi)(x_1,\ldots,x_{n-1})
  &=
  \sqrt n
  \int_{\bT^d}
  \overline{f(x)}
  \psi(x_1,\ldots,x_{n-1},x)\dd x,
  \\
  (a^*(f)\psi)(x_1,\ldots,x_{n+1})
  &=
  \frac1{\sqrt{n+1}}
  \sum_{j=1}^{n+1}
  f(x_j)
  \psi(x_1,\ldots,x_{j-1},x_{j+1},\ldots,x_{n+1}).
 \end{aligned}
\end{equation}
Equivalently, in distributional notation,
$$
 a(f)=\int_{\bT^d}\overline{f(x)}a_x\dd x,
 \qquad
 a^*(f)=\int_{\bT^d}f(x)a_x^*\dd x .
$$
The operators in \eqref{eq:annihilation-creation-operators} satisfy the canonical commutation relations (CCR)
$$
 [a(f),a(g)]=[a^*(f),a^*(g)]=0,
 \qquad
 [a(f),a^*(g)]=\inprod{f}{g}\1_{\gF}.
$$
For a one-body operator $A$ on $\gH$, its second quantization $\dG(A)$ is defined on each particle sector by
$$
 \dG(A)
 =
 0\oplus\bigoplus_{n\geq1}\sum_{j=1}^n A_j,
$$
where $A_j$ acts on the $j$-th variable. In particular, we write $\cN=\dG(\1)$ for the particle-number operator. With $a_k=a(e_k)$, we have
$$
 \dG(h)=\sum_{k\in\Z^d}h(k)a_k^*a_k,
 \qquad h(k)=1+|k|^2.
$$

A state on $\gF$ is a positive trace-class operator $\Gamma$ with $\tr_{\gF}\Gamma=1$. For $k\geq1$, we define its $k$-particle reduced density matrix $\Gamma^{(k)}$ by
\begin{equation}\label{eq:rdm-convention}
 \Gamma^{(k)}
 =
 \sum_{n\geq k}\binom nk
 \tr_{k+1,\ldots,n}\Gamma_n,
\end{equation}
where $\Gamma_n=\Pi_n\Gamma\Pi_n$, and $\Pi_n$ is the projection onto the $n$-particle sector. Equivalently,
$$
 \inprod{g_1\otimes_s\ldots\otimes_s g_k}{\Gamma^{(k)}f_1\otimes_s\ldots\otimes_s f_k}=\tr\Big[a^*(f_1)\ldots a^*(f_k)a(g_1)\ldots a(g_k)\Gamma\Big],
$$
where $ \phi_1\otimes_s\cdots\otimes_s\phi_k := \frac1{\sqrt{k!}}\sum_{\sigma\in S_k} \phi_{\sigma(1)}\otimes\cdots\otimes\phi_{\sigma(k)}. $

For $\lambda>0$, the free Gibbs state $\Gamma_{{\rm free}}$ and its partition function $Z_{{\rm free}}$ are defined by
\begin{equation}\label{eq:def-free-gibbs-state}
 \Gamma_{\rm free}
 =
 \frac1{Z_{\rm free}}e^{-\lambda\dG(h)},
 \qquad
 Z_{\rm free}
 =
 \tr_{\gF}\big(e^{-\lambda\dG(h)}\big).
\end{equation}
The free Gibbs state $\Gamma_{{\rm free}}$ is the quantum counterpart of the Gaussian measure $\mu_0$. Its $k$-particle reduced density matrix, with the convention \eqref{eq:rdm-convention}, is
$$
 \Gamma_{\rm free}^{(k)}
 =
 \big(\Gamma_{\rm free}^{(1)}\big)^{\otimes k}
 \big|_{\bigotimes_{\mathrm{sym}}^k\gH}
 =
 \Big(\frac1{e^{\lambda h}-1}\Big)^{\otimes k}\big|_{\bigotimes_{\mathrm{sym}}^k\gH}.
$$
The expected particle number of the free Gibbs state is
\begin{equation}\label{eq:free-number}
 N_{\rm free}
 :=
 \tr_{\gF}(\cN\Gamma_{\rm free})
 =
 \tr_{\gH}\Gamma_{\rm free}^{(1)}
 =
 \sum_{k\in\Z^d}\frac1{e^{\lambda h(k)}-1}
 =
 (2\pi)^d\rho_0,
\end{equation}
where $\rho_0$ is defined in \eqref{eq:free-particle-density}.

To describe the interacting Bose gas in parallel with the classical field theory, we first introduce quantum mass renormalization. Following the formal quantum--classical correspondence, we define 
$$
 \mathbb M_\lambda^A=\lambda\dG(A)-\tr_{\gF} \left( \lambda\dG(A)\Gamma_{{\rm free}} \right)=\lambda\dG(A)-\tr(A\lambda\Gamma_{{\rm free}}^{(1)})
$$
for every bounded normal operator $A$ on $\gH$. We define the renormalized quantum interaction by 
\begin{equation}
 \bW_\lambda^{\ren}
 =
 \frac12\sum_{k\in\mathbb Z^d}\widehat w(k)|\mathbb M_\lambda^{e_k}|^2-\alpha_\lambda\mathbb M_\lambda^\1-\beta_\lambda,
 \label{eq:abstract-quantum-W}
\end{equation}
where $\alpha_\lambda,\beta_\lambda\in\mathbb R$ may depend on $\lambda$. Using \eqref{eq:interaction-assumption}, we can write this interaction as
$$
 \bW_{\lambda}^{\ren}
 =
 \frac{\lambda^2}{2}
 \iint_{(\bT^d)^2}
 (a_x^*a_x-\rho_0)
 w(x-y)
 (a_y^*a_y-\rho_0)
 \dd x\d y
 -
 \alpha_{\lambda}\lambda
 \int_{\bT^d}(a_x^*a_x-\rho_0)\dd x
 -
 \beta_{\lambda}.
$$

The interacting Bose gas is described by the Gibbs state $\Gamma_\lambda$ with partition function $Z_\lambda$, defined by
$$
 \Gamma_\lambda:=\frac{1}{Z_\lambda}e^{-\lambda\dG(h)-\mathbb W_\lambda^{{\rm ren}}},\qquad Z_\lambda:=\tr_{\gF}\big(e^{-\lambda\dG(h)-\mathbb W_\lambda^{{\rm ren}}}\big).
$$

\medskip

\subsection{Two-dimensional model}

We first consider the two-dimensional model.

\begin{assumption}\label{ass:Phi42} We consider a real, even, nonnegative profile $v\in C_c^\infty(\R^2)$ and assume that its Fourier transform $\widehat v(\xi)=\int_{\R^2}v(x)e^{-\ii\xi\cdot x}\dd x$ satisfies $\widehat v\geq0$ and $\widehat v(0)=1$.
\end{assumption}

For $0<\varepsilon\leq1$, we define the classical interaction by
\begin{equation}\label{eq:Phi42-classical-interaction}
 \cW_{0,\varepsilon}[u]
 =
 \frac12\sum_{k\in\Z^2}
 \widehat v(\varepsilon k)|\cM_0^{e_k}[u]|^2
 -
 \tau_\varepsilon\cM_0^\1[u]
 -
 \frac{(2\pi)^2}{2}E_\varepsilon.
\end{equation}
Following the quantum--classical correspondence, we use the same counterterms in the quantum interaction and define
\begin{equation}\label{eq:Phi42-quantum-interaction}
 \bW_{\lambda,\varepsilon}^{\ren}
 =
 \frac12\sum_{k\in\Z^2}
 \widehat v(\varepsilon k)|\mathbb M_\lambda^{e_k}|^2
 -
 \tau_\varepsilon\mathbb M_\lambda^\1
 -
 \frac{(2\pi)^2}{2}E_\varepsilon.
\end{equation}
Thus, in \eqref{eq:abstract-classical-W0} and \eqref{eq:abstract-quantum-W}, we take $\widehat w(k)=\widehat v(\varepsilon k)$ and
\begin{equation}\label{eq:Phi42-counterterms}
 \alpha_{0,\varepsilon}
 =
 \alpha_{\lambda,\varepsilon}
 =
 \tau_\varepsilon,
 \qquad
 \beta_{0,\varepsilon}
 =
 \beta_{\lambda,\varepsilon}
 =
 \frac{(2\pi)^2}{2}E_\varepsilon.
\end{equation}
To obtain the two-body potential on $\bT^2$, we periodize $\varepsilon^{-2}v(\cdot/\varepsilon)$ and write
$$
 w_\varepsilon(x)
 =
 \sum_{\ell\in\Z^2}\varepsilon^{-2}
 v\Big(\frac{x+2\pi\ell}{\varepsilon}\Big)
 =
 \frac1{(2\pi)^2}\sum_{k\in\Z^2}
 \widehat v(\varepsilon k)e^{\ii k\cdot x},
 \qquad 0<\varepsilon\leq1.
$$
Thus $w_\varepsilon\geq0$, and Assumption~\ref{ass:interaction} holds.

To specify the counterterms, we use the covariance kernel of the Gaussian measure $\mu_0$, whose covariance is $h^{-1}=(-\Delta+\1)^{-1}$:
$$
 G(x-y)
 =
 \left\langle u(x)\overline{u(y)}\right\rangle_{\mu_0}
 =
 \frac1{(2\pi)^2}\sum_{k\in\Z^2}
 \frac{e^{\ii k\cdot(x-y)}}{h(k)}.
$$
We choose the mass counterterm
\begin{equation}\label{eq:Phi42-mass-counterterm}
 \tau_\varepsilon
 =
 \int_{\bT^2}w_\varepsilon(x)G(x)\dd x
 =
 \frac1{(2\pi)^2}
 \sum_{k\in\Z^2}\frac{\widehat v(\varepsilon k)}{h(k)}.
\end{equation}
For the scalar counterterm in \eqref{eq:Phi42-classical-interaction} and \eqref{eq:Phi42-quantum-interaction}, we use
\begin{equation}\label{eq:Phi42-data}
 E_\varepsilon
 =
 \int_{\bT^2}w_\varepsilon(x)G(x)^2\dd x
 =
 \frac1{(2\pi)^4}
 \sum_{k,j\in\Z^2}
 \frac{\widehat v(\varepsilon k)}
 {h(j)h(j+k)}.
\end{equation}
These choices agree with the renormalization in \cite[(2.17)--(2.18), (3.1)]{FKSS23}.

\Needspace{8\baselineskip}
\begin{theorem}[Quantum-to-Hartree limit in two dimensions] \label{thm:Phi42} Let $d=2$, and let $v$ satisfy Assumption~\ref{ass:Phi42}. Consider the Gibbs state $\Gamma_{\lambda,\varepsilon} =Z_{\lambda,\varepsilon}^{-1} e^{-\lambda\dG(h)-\bW_{\lambda,\varepsilon}^{\ren}}$ on the bosonic Fock space over $\gH=L^2(\bT^2;\CC)$, associated with the interaction $\bW_{\lambda,\varepsilon}^{\ren}$ in \eqref{eq:Phi42-quantum-interaction}, with $h=-\Delta+\1$. We write $Z_{\lambda,\varepsilon} =\tr_{\gF}e^{-\lambda\dG(h)-\bW_{\lambda,\varepsilon}^{\ren}}$.

Let $\Gamma_{\rm free} =Z_{\rm free}^{-1}e^{-\lambda\dG(h)}$ be the free Gibbs state defined in\kern0pt{} \eqref{eq:def-free-gibbs-state}, with $Z_{\rm free}=\tr_{\gF}e^{-\lambda\dG(h)}$. Let $\mu_0$ be the Gaussian measure with covariance $(-\Delta+\1)^{-1}$, and let $\dd\nu_\varepsilon(u) =z_{0,\varepsilon}^{-1}e^{-\cW_{0,\varepsilon}[u]}\dd\mu_0(u)$ be the Hartree Gibbs measure associated with\kern0pt{} \eqref{eq:Phi42-classical-interaction}, with $z_{0,\varepsilon} =\int e^{-\cW_{0,\varepsilon}[u]}\dd\mu_0(u)$.

Then, for $\varepsilon=\lambda^\eta$ with $0<\eta<1/2$,
$$
 \Big|
 \log\frac{Z_{\lambda,\varepsilon}}{Z_{\rm free}}
 -
 \log z_{0,\varepsilon}
 \Big|
 \longrightarrow0
 \qquad(\lambda\to0^+).
$$
Moreover, for every fixed $k\geq1$,
$$
 \big\|k!\lambda^k\Gamma_{\lambda,\varepsilon}^{(k)}
 -\gamma_{\nu_\varepsilon}^{(k)}\big\|_{\gS^2}
 \longrightarrow0
 \qquad(\lambda\to0^+).
$$
Here we write $\gamma_{\nu_\varepsilon}^{(k)} =\int |u^{\otimes k}\rangle\langle u^{\otimes k}| \dd\nu_\varepsilon(u)$.
\end{theorem}
\textcolor{red}{Theorem~\ref{thm:Phi42} follows from Theorem~\ref{thm:abstract-comparison}, as shown at the end of this section.}

\textcolor{red}{The local $\Phi^4_2$ measure $\nu_{\Phi^4_2}$ and its partition function $z_{\Phi^4_2}$ are defined in Section~\ref{sec:derivation-Phi4d}. Combining the classical approximation with Theorem~\ref{thm:Phi42} gives the following corollary.}

\begin{corollary}[Derivation of $\Phi^4_2$]\label{cor:Phi42-local}
Under the assumptions of Theorem~\ref{thm:Phi42}, we have
$$
 -\log\frac{Z_{\lambda,\lambda^\eta}}{Z_{\rm free}}
 \longrightarrow-\log z_{\Phi^4_2},
 \qquad 0<\eta<1/2.
$$
For every fixed $k\geq1$,
$$
 \big\|k!\lambda^k\Gamma_{\lambda,\lambda^\eta}^{(k)}
 -\gamma_{\nu_{\Phi^4_2}}^{(k)}\big\|_{\gS^2}
 \longrightarrow0,
 \qquad 0<\eta<1/2.
$$
\end{corollary}
\textcolor{red}{The proof of Corollary~\ref{cor:Phi42-local} is given in Section~\ref{sec:derivation-Phi4d}.}

\begin{remark}
To the author's knowledge, the preceding corollary provides the first derivation of the local $\Phi^4_2$ measure from an interacting Bose gas covering every fixed exponent $0<\eta<1/2$ in the scaling $\varepsilon=\lambda^\eta$. It gives relative free-energy convergence and Hilbert--Schmidt convergence at every fixed order, with polynomial exponents arbitrarily close to the diluteness threshold from below.
\end{remark}

\medskip

\subsection{Three-dimensional model}

We next consider the three-dimensional Hartree model studied in \cite[Section 2.1, Section 7.1]{NZZ25}.

\begin{assumption}\label{ass:Phi43} We consider a nonnegative profile $v\in L^1(\R^3)$, with Fourier transform $\widehat v(\xi)=\int_{\R^3}v(x)e^{-\ii\xi\cdot x}\dd x$. We assume that, for some $\delta_0>0$,
\begin{equation}\label{eq:Phi43-profile}
 \begin{gathered}
  \widehat v(0)=1,\qquad
  0\leq\widehat v(\xi)
  \leq\frac{C_v}{1+|\xi|^{3+\delta_0}},\qquad
  |D^j\widehat v(\xi)|
  \leq\frac{C_v}{1+|\xi|^j},
  \qquad j=1,2.
 \end{gathered}
\end{equation}
\end{assumption}

For $m\in\R$ and $0<\varepsilon\leq1$, we consider the classical interaction
\begin{equation}\label{eq:Phi43-classical-interaction}
 \cW_{0,\varepsilon}[u]
 =
 \frac12\sum_{k\in\Z^3}
 \widehat v(\varepsilon k)|\cM_0^{e_k}[u]|^2
 -
 (a_\varepsilon-6b_\varepsilon-m+1)\cM_0^\1[u].
\end{equation}
On the quantum side, we include the correction $e_{\lambda,\varepsilon}$ in the mass counterterm and define
\begin{equation}\label{eq:Phi43-quantum-interaction}
 \bW_{\lambda,\varepsilon}^{\ren}
 =
 \frac12\sum_{k\in\Z^3}
 \widehat v(\varepsilon k)|\mathbb M_\lambda^{e_k}|^2
 -
 (a_\varepsilon-6b_\varepsilon-m+1-e_{\lambda,\varepsilon})
 \mathbb M_\lambda^\1.
\end{equation}
Thus, in \eqref{eq:abstract-classical-W0} and \eqref{eq:abstract-quantum-W}, we take $\widehat w(k)=\widehat v(\varepsilon k)$, $\alpha_{0,\varepsilon}=a_\varepsilon-6b_\varepsilon-m+1$ and $\alpha_{\lambda,\varepsilon} =\alpha_{0,\varepsilon}-e_{\lambda,\varepsilon}$, with $\beta_{0,\varepsilon}=\beta_{\lambda,\varepsilon}=0$. We obtain the two-body potential on $\bT^3$ by periodization:
$$
 w_\varepsilon(x)
 =
 \frac1{(2\pi)^3}\sum_{k\in\Z^3}
 \widehat v(\varepsilon k)e^{\ii k\cdot x},
 \qquad 0<\varepsilon\leq1.
$$
This periodization of $\varepsilon^{-3}v(\cdot/\varepsilon)$ yields $w_\varepsilon\geq0$, and Assumption~\ref{ass:interaction} holds.

We now specify the mass counterterms in \eqref{eq:Phi43-classical-interaction}. We first take
$$
 a_\varepsilon
 =\int_{\bT^3}w_\eps(x)G(x)\dd x=
 \frac1{(2\pi)^3}
 \sum_{n\in\Z^3}\frac{\widehat v(\varepsilon n)}{h(n)}.
$$
In three dimensions, we also include $b_\varepsilon$, defined by
$$
 6b_\varepsilon
 =
 \frac1{(2\pi)^6}
 \sum_{n,j\in\Z^3}
 \frac{\widehat v(\varepsilon n)^2+
 \widehat v(\varepsilon n)\widehat v(\varepsilon j)}
 {h(n)h(j)h(n+j)}.
$$
Assumption~\ref{ass:Phi43} and these counterterms agree with \cite[(2.10), (2.12)]{NZZ25}. We write their Hartree interaction \cite[(2.17), Section~2.3]{NZZ25} in the form \eqref{eq:Phi43-classical-interaction}, relative to the Gaussian measure with covariance $(-\Delta+\1)^{-1}$.

To define $e_{\lambda,\varepsilon}$, we use the free particle density from \eqref{eq:free-particle-density},
$$
 \rho_0
 =
 \frac1{(2\pi)^3}
 \sum_{n\in\Z^3}\frac1{e^{\lambda(1+|n|^2)}-1}.
$$
For the constant term in the quantum correction, we take
$$
 C_0
 =
 -\frac1{4\pi}
 +
 \frac1{4\pi}
 \sum_{\ell\in2\pi\Z^3\setminus\{0\}}
 \frac{e^{-|\ell|}}{|\ell|}.
$$
With $\zeta$ denoting the Riemann zeta function, we then choose the quantum correction
\begin{equation}\label{eq:Phi43-quantum-counterterm}
 e_{\lambda,\varepsilon}
 =
 \lambda\rho_0
 -
 \frac{\zeta(3/2)}{(4\pi)^{3/2}}\lambda^{-1/2}
 -
 C_0-\frac{\lambda}{2}w_\varepsilon(0).
\end{equation}
This choice follows \cite[(2.7)--(2.8), (7.3)]{NZZ25}.

\Needspace{8\baselineskip}
\begin{theorem}[Quantum-to-Hartree limit in three dimensions] \label{thm:Phi43} Let $d=3$, and let $v$ satisfy Assumption~\ref{ass:Phi43}. Fix $m\in\R$. Consider the Gibbs state $\Gamma_{\lambda,\varepsilon} =Z_{\lambda,\varepsilon}^{-1} e^{-\lambda\dG(h)-\bW_{\lambda,\varepsilon}^{\ren}}$ on the bosonic Fock space over $\gH=L^2(\bT^3;\CC)$, associated with the interaction $\bW_{\lambda,\varepsilon}^{\ren}$ in \eqref{eq:Phi43-quantum-interaction}, with $h=-\Delta+\1$. We write $Z_{\lambda,\varepsilon} =\tr_{\gF}e^{-\lambda\dG(h)-\bW_{\lambda,\varepsilon}^{\ren}}$.

Let $\Gamma_{\rm free} =Z_{\rm free}^{-1}e^{-\lambda\dG(h)}$ be the free Gibbs state defined in\kern0pt{} \eqref{eq:def-free-gibbs-state}, with $Z_{\rm free}=\tr_{\gF}e^{-\lambda\dG(h)}$. Let $\mu_0$ be the Gaussian measure with covariance $(-\Delta+\1)^{-1}$, and let $\dd\nu_\varepsilon(u) =z_{0,\varepsilon}^{-1}e^{-\cW_{0,\varepsilon}[u]}\dd\mu_0(u)$ be the Hartree Gibbs measure associated with\kern0pt{} \eqref{eq:Phi43-classical-interaction}, with $z_{0,\varepsilon} =\int e^{-\cW_{0,\varepsilon}[u]}\dd\mu_0(u)$.

Then, for $\varepsilon=\lambda^\eta$ with $0<\eta<1/26$,
$$
 \Big|
 \log\frac{Z_{\lambda,\varepsilon}}{Z_{\rm free}}
 -
 \log z_{0,\varepsilon}
 \Big|
 \longrightarrow0
 \qquad(\lambda\to0^+).
$$
Moreover, for every fixed $k\geq1$,
$$
 \big\|k!\lambda^k\Gamma_{\lambda,\varepsilon}^{(k)}
 -\gamma_{\nu_\varepsilon}^{(k)}\big\|_{\gS^2}
 \longrightarrow0
 \qquad(\lambda\to0^+).
$$
Here we write $\gamma_{\nu_\varepsilon}^{(k)} =\int |u^{\otimes k}\rangle\langle u^{\otimes k}| \dd\nu_\varepsilon(u)$.
\end{theorem}
\textcolor{red}{Theorem~\ref{thm:Phi43} follows from Theorem~\ref{thm:abstract-comparison}, as shown at the end of this section. Let $\nu_{\Phi^4_3}$ be the local measure in \cite[Theorem~2.6]{NZZ25}, with mass parameter $m_0=m+2C_1+2C_2$ for the counterterms in \eqref{eq:Phi43-classical-interaction}; the constants $C_1,C_2$ are defined in \cite[(2.12)]{NZZ25}. The bounds in Proposition~\ref{prop:classical-Hartree-correlations} strengthen their classical correlation convergence to the Hilbert--Schmidt norm. Combining this with Theorem~\ref{thm:Phi43} gives the following local limit.}

\begingroup\color{red}
\begin{corollary}[Derivation of $\Phi^4_3$]
\label{cor:Phi43-local-HS}
Under the assumptions of Theorem~\ref{thm:Phi43}, for every fixed $k\geq1$, the form $\gamma_{\nu_{\Phi^4_3}}^{(k)}$ has a unique positive Hilbert--Schmidt realization on $\bigotimes_{\mathrm{sym}}^k\gH$, and
\begin{equation}\label{eq:Phi43-local-classical-HS}
 \big\|\gamma_{\nu_\varepsilon}^{(k)}
 -\gamma_{\nu_{\Phi^4_3}}^{(k)}\big\|_{\gS^2}
 \longrightarrow0
 \qquad(\varepsilon\to0^+).
\end{equation}
Consequently, for every fixed $0<\eta<1/26$,
$$\stepcounter{equation}
 \big\|k!\lambda^k\Gamma_{\lambda,\lambda^\eta}^{(k)}
 -\gamma_{\nu_{\Phi^4_3}}^{(k)}\big\|_{\gS^2}
 \longrightarrow0
 \qquad(\lambda\to0^+).
$$
\end{corollary}
The proof of Corollary~\ref{cor:Phi43-local-HS} is given in Section~\ref{sec:derivation-Phi4d}.
\endgroup

\begin{remark}
In Theorem~\ref{thm:Phi43}, we establish the Hartree comparison of \cite[Theorem~2.8]{NZZ25} for the explicit range $0<\eta<1/26$. For polynomially shrinking interaction ranges, the authors prove \cite[(2.24)]{NZZ25} for $1\leq k\leq6$ when tested against tensors with finite Fourier support of the form $\varphi^{\otimes k}$. Here we obtain Hilbert--Schmidt convergence for every fixed $k$. \textcolor{red}{The stationary coupling constructed in \cite[Lemma~4.17]{NZZ25} gives a joint law $\pi_\varepsilon$ with marginals $\nu_\varepsilon$ and $\mu_0$. The moment estimates for the remainder and stochastic terms used in the proof of \cite[Theorem~4.1]{NZZ25} imply
$$
 \sup_{0<\varepsilon\leq1}
 \int \|u-z\|_{\gH}^{p}\dd\pi_\varepsilon(u,z)
 \leq C_{p,v,m},
 \qquad p\geq2.
$$
Appendix~\ref{app:classical-Hartree-correlations} combines this coupling with their remainder estimates and Gaussian integration to prove \eqref{eq:Phi43-local-weighted-HS}. At $\delta=0$, this controls the Hartree correlation term in \eqref{eq:main-HS} and proves the density-matrix convergence in Theorem~\ref{thm:Phi43}. For a positive $\delta$, the same bound gives uniform Fourier tails and yields the local Hilbert--Schmidt limit in Corollary~\ref{cor:Phi43-local-HS}.}
\end{remark}

\begin{remark}
The exponent $1/26$ in Theorem~\ref{thm:Phi43} and Corollary~\ref{cor:Phi43-local-HS} can be improved by refining the three-dimensional estimates within the same upper-symbol and relative-entropy framework. We retain the present bounds for a concise, unified treatment of dimensions two and three.
\end{remark}

\medskip

\subsection{A quantitative comparison}

We now formulate a comparison that applies to both models. For $d\in\{2,3\}$, we consider
\begin{equation}\label{eq:comparison-classical-interaction}
 \cW_0[u]
 =
 \frac12\sum_{k\in\Z^d}
 \widehat w(k)|\cM_0^{e_k}[u]|^2
 -
 \alpha_0\cM_0^\1[u]
 -
 \beta_0
\end{equation}
on the classical side. To compare it with the quantum gas, we use
\begin{equation}\label{eq:comparison-quantum-interaction}
 \bW_\lambda^{\ren}
 =
 \frac12\sum_{k\in\Z^d}
 \widehat w(k)|\mathbb M_\lambda^{e_k}|^2
 -
 \alpha_\lambda\mathbb M_\lambda^\1
 -
 \beta_\lambda,
\end{equation}
as in \eqref{eq:abstract-classical-W0} and \eqref{eq:abstract-quantum-W}. We allow $\alpha_\lambda,\beta_\lambda\in\R$ to depend on $\lambda$, and use $\alpha_0,\beta_0\in\R$ to specify the target Gibbs measure. We impose the following assumption:
\begin{assumption}\label{ass:interaction} The interaction $w$ is real, even, and positive definite, with
\begin{equation}
 w(x)=\frac1{(2\pi)^d}
 \sum_{k\in\Z^d}\widehat w(k)e^{\ii k\cdot x},
 \qquad
 \widehat w(k)=\widehat w(-k)\geq0.
 \label{eq:interaction-assumption}
\end{equation}
Moreover,
\begin{equation}\label{eq:def-of-c0-wo}
 w(0)=\frac1{(2\pi)^d}\sum_k\widehat w(k)<\infty,
 \qquad
 c_0:=\frac{\widehat w(0)}{2(2\pi)^d}>0.
\end{equation}
\end{assumption}

The comparison below keeps the dependence on the interaction explicit. We measure its contribution through the Gaussian covariance by
\begin{equation}\label{eq:physical-errors}
 E
 =
 \frac1{(2\pi)^{2d}}\sum_{k,j\in\Z^d}
 \frac{\widehat w(k)}{h(j)h(j+k)},
 \qquad h(k)=1+|k|^2.
\end{equation}
To account for the free particle number and the differences between the quantum and classical counterterms, we use
\begin{equation}\label{eq:comparison-e}
 \mathfrak e_\lambda
 =
 \lambda^2(N_{\rm free}+1)w(0)
 +|\alpha_\lambda-\alpha_0|
 +|\beta_\lambda-\beta_0|.
\end{equation}
Combining these contributions with the two mass counterterms, we define
\begin{equation}\label{eq:comparison-r}
 \mathfrak r_\lambda
 =
 E^{1/2}\mathfrak e_\lambda^{1/2}
 +( |\alpha_\lambda|+|\alpha_0|)
 (\lambda^2N_{\rm free})^{1/2}
 +\mathfrak e_\lambda.
\end{equation}
To keep track of the size of the interaction and the mass counterterms, we also use
\begin{equation}\label{eq:comparison-a}
 \mathfrak a_\lambda
 =
 1+E+\alpha_\lambda^2+\alpha_0^2+\mathfrak e_\lambda.
\end{equation}
The product $\mathfrak a_\lambda^2\mathfrak r_\lambda$ will control the relative free-energy error. In the density-matrix estimates, we will also use
\begin{equation}\label{eq:comparison-b}
 b_\lambda
 =
 \lambda\left(
 \widehat w(0)\lambda\rho_0+|\alpha_\lambda|
 \right).
\end{equation}

\Needspace{8\baselineskip}
\begin{theorem}[A quantitative two-body comparison] \label{thm:abstract-comparison} Let $d\in\{2,3\}$, and let $w$ satisfy Assumption~\ref{ass:interaction}. Consider the Gibbs state $\Gamma_\lambda =Z_\lambda^{-1}e^{-\lambda\dG(h)-\bW_\lambda^{\ren}}$ on the bosonic Fock space over $\gH=L^2(\bT^d;\CC)$, associated with the interaction $\bW_\lambda^{\ren}$ in \eqref{eq:comparison-quantum-interaction}, with $h=-\Delta+\1$. We write $Z_\lambda=\tr_{\gF}e^{-\lambda\dG(h)-\bW_\lambda^{\ren}}$.

Let $\Gamma_{\rm free} =Z_{\rm free}^{-1}e^{-\lambda\dG(h)}$ be the free Gibbs state defined in\kern0pt{} \eqref{eq:def-free-gibbs-state}, with $Z_{\rm free}=\tr_{\gF}e^{-\lambda\dG(h)}$. Let $\mu_0$ be the Gaussian measure with covariance $(-\Delta+\1)^{-1}$, and let $\dd\nu_0(u)=z_0^{-1}e^{-\cW_0[u]}\dd\mu_0(u)$ be the target Gibbs measure associated with \eqref{eq:comparison-classical-interaction}, with $z_0=\int e^{-\cW_0[u]}\dd\mu_0(u)$. We write $\gamma_{\nu_0}^{(k)} =\int |u^{\otimes k}\rangle\langle u^{\otimes k}|\dd\nu_0(u)$ and use\kern0pt{} \eqref{eq:rdm-convention} for the quantum reduced density matrices.

Then, for sufficiently small $\lambda>0$,
\begin{equation}\label{eq:main-FE}
 \Big|
 \log\frac{Z_\lambda}{Z_{\rm free}}-\log z_0
 \Big|
 \lesssim_{d,c_0} \mathfrak a_\lambda^2\mathfrak r_\lambda.
\end{equation}
Suppose in addition that $w\geq0$ and $ \mathfrak a_\lambda^2\mathfrak r_\lambda\ll 1. $ Then, for every fixed $k,L\in\mathbb N$,
\begin{equation}\label{eq:main-moment}
 k!\lambda^k\big\|\Gamma_\lambda^{(k)}\big\|_{\gS^2}
 \lesssim_{d,c_0,k,L}
 e^{kb_\lambda}
 \big[
 \big\|\gamma_{\nu_0}^{(k)}\big\|_{\gS^2}
 +\mathfrak a_\lambda^k
 (\mathfrak a_\lambda^2\mathfrak r_\lambda)^L
 \big].
\end{equation}
Moreover,
\begin{equation}
 \big\|k!\lambda^k\Gamma_\lambda^{(k)}
 -\gamma_{\nu_0}^{(k)}\big\|_{\gS^2}^{2}
 \lesssim_{d,c_0,k,L}e^{2kb_\lambda}
 \big(
 \mathfrak a_\lambda\mathfrak r_\lambda^{1/2}
 +\lambda^2 N_{{\rm free}}
 \big)
 \big[
 1+\big\|\gamma_{\nu_0}^{(2k)}\big\|_{\gS^2}
 +\mathfrak a_\lambda^{2k}
 (\mathfrak a_\lambda^2\mathfrak r_\lambda)^L
 \big].
 \label{eq:main-HS}
\end{equation}
Each Hilbert--Schmidt norm is taken on the corresponding symmetric particle space; the quantities in \eqref{eq:physical-errors}--\eqref{eq:comparison-b} record the dependence on the interaction.
\end{theorem}

For fixed $w,\alpha_0,\beta_0$, the convergence $\alpha_\lambda\to\alpha_0$ and $\beta_\lambda\to\beta_0$ keeps $\mathfrak a_\lambda$ bounded and gives $\mathfrak r_\lambda,b_\lambda\to0$. We therefore obtain relative free-energy convergence from \eqref{eq:main-FE}. When $w\geq0$, we also obtain Hilbert--Schmidt convergence to $\gamma_{\nu_0}^{(k)}$ by taking $L=1$ in \eqref{eq:main-HS}.

\begin{proof}[Proof of Theorems~\ref{thm:Phi42} and~\ref{thm:Phi43}] We apply Theorem~\ref{thm:abstract-comparison} with $w=w_\varepsilon$ and the counterterms specified in \eqref{eq:Phi42-counterterms} and \eqref{eq:Phi43-classical-interaction}--\eqref{eq:Phi43-quantum-interaction}, respectively. In both dimensions, $w_\varepsilon\geq0$, $c_0=1/[2(2\pi)^d]$, and $w_\varepsilon(0)\lesssim_v\varepsilon^{-d}$. By \cite[Lemma~B.1]{LNR21},
$$
 \lambda^2N_{\rm free}\asymp_d
 \begin{cases}
 \lambda\log(2/\lambda),&d=2,\\
 \lambda^{1/2},&d=3.
 \end{cases}
$$
Fix $k\geq1$ and take $\varepsilon=\lambda^\eta$ with the exponent in the range specified by the corresponding theorem, so that $0<\lambda\leq\varepsilon^2\leq1$.

For $d=2$, $E=E_\varepsilon$, and \cite[(5.3)]{FKSS23} gives $E+\tau_\varepsilon^2\lesssim_v\log^2(2/\varepsilon)$ in the normalization of \eqref{eq:Phi42-mass-counterterm}--\eqref{eq:Phi42-data}. Since the two sets of counterterms agree,
$$
 \mathfrak e_\lambda
 =
 \lambda^2(N_{\rm free}+1)w_\varepsilon(0)
 \lesssim_v
 \lambda\varepsilon^{-2}\log(2/\lambda).
$$
Hence $\mathfrak a_\lambda\lesssim_{v,\eta}\log^2(2/\lambda)$. With $L=1$, \eqref{eq:comparison-r} yields
$$
 \mathfrak a_\lambda^{2k}
 (\mathfrak a_\lambda^2\mathfrak r_\lambda)
 \lesssim_{k,v,\eta}
 \lambda^{1/2}\varepsilon^{-1}
 \log^{4k+11/2}(2/\lambda)=
 \lambda^{1/2-\eta}\log^{4k+11/2}(2/\lambda)
 \longrightarrow0.
$$

For $d=3$, the bounds following \cite[(2.12)]{NZZ25} give $a_\varepsilon+|b_\varepsilon|\lesssim_v\varepsilon^{-1}$. Using \cite[(5.3)]{FKSS23} and \eqref{eq:Phi43-profile}, we obtain
$$
 E
 \lesssim
 \sum_{n\in\Z^3}
 \frac{\widehat v(\varepsilon n)}{1+|n|}
 \lesssim_v\varepsilon^{-2}.
$$
Substituting \cite[(C.3)]{NZZ25} into \eqref{eq:Phi43-quantum-counterterm}, we obtain
$$
 |e_{\lambda,\varepsilon}|
 \lesssim
 \lambda^{1/2}+\lambda w_\varepsilon(0)
 \lesssim_v
 \lambda^{1/2}+\lambda\varepsilon^{-3}.
$$
Since $\lambda\leq\varepsilon^2$, \eqref{eq:comparison-e}--\eqref{eq:comparison-a} give $\mathfrak e_\lambda\lesssim_v\lambda^{1/2}\varepsilon^{-3}$ and $\mathfrak a_\lambda\lesssim_{v,m}\varepsilon^{-2}$. Choose a fixed integer $L\geq1$ such that $L(1/4-13\eta/2)>4k\eta$. Then
$$
 \mathfrak a_\lambda^{2k}
 (\mathfrak a_\lambda^2\mathfrak r_\lambda)^L
 \lesssim_{k,L,v,m}
 \varepsilon^{-4k}
 \big(
 \lambda^{1/4}\varepsilon^{-13/2}
 +\lambda^{1/2}\varepsilon^{-7}
 \big)^L\lesssim_{k,L,v,m}
 \lambda^{L(1/4-13\eta/2)-4k\eta}
 \longrightarrow0.
$$

Since $\mathfrak a_\lambda\geq1$, these estimates imply $\mathfrak a_\lambda^2\mathfrak r_\lambda\to0$ in both cases. Thus \eqref{eq:main-FE} proves the two free-energy limits. Moreover, \eqref{eq:free-number} and \eqref{eq:comparison-b} give
$$
 b_\lambda
 =
 \frac{\lambda^2N_{\rm free}}{(2\pi)^d}
 +\lambda|\alpha_\lambda|
 \leq
 \frac{\lambda^2N_{\rm free}}{(2\pi)^d}
 +\lambda\mathfrak a_\lambda^{1/2}
 \longrightarrow0.
$$
By Appendix~\ref{app:classical-Hartree-correlations}, $\sup_{0<\varepsilon\leq1} \big\|\gamma_{\nu_\varepsilon}^{(2k)}\big\|_{\gS^2}<\infty$. \textcolor{red}{In dimension three, this bound is deduced from the stationary coupling and remainder estimates of \cite{NZZ25}, as detailed in that appendix.} Using \eqref{eq:main-HS} and the preceding bounds, we conclude that
$$
 \big\|k!\lambda^k\Gamma_{\lambda,\varepsilon}^{(k)}
 -\gamma_{\nu_\varepsilon}^{(k)}\big\|_{\gS^2}^{2}
 \lesssim
 e^{2kb_\lambda}
 \big(
 \mathfrak a_\lambda\mathfrak r_\lambda^{1/2}
 +\lambda^2N_{\rm free}
 \big)
 \longrightarrow0.
$$
\end{proof}
We now turn to the proof of Theorem~\ref{thm:abstract-comparison}. The local corollaries are proved in Section~\ref{sec:derivation-Phi4d}.

\bigskip

\section{Proof strategy}\label{sec:proof-strategy}

To prove Theorem~\ref{thm:abstract-comparison}, we compare the quantum trace with a Gibbs integral against $\mu_0$. We begin with the free upper-symbol representation used in \cite{Cub26}.

The Golden--Thompson inequality separates the interaction from the free state, and the semiclassical correspondence suggests that the bound becomes asymptotically sharp as $\lambda\to0$:
\begin{equation}\label{eq:GT-upper-symbol}
 \frac{Z_\lambda}{Z_{\rm free}}
 \approx
 \tr_{\gF}\big(\Gamma_{\rm free}e^{-\bW_\lambda^{\ren}}\big)
 =:z_{\mathrm q},\qquad \mathrm{for\ }\lambda\ll 1.
\end{equation}
We write $z_{\mathrm c}=z_0$, with the subscripts $\mathrm q$ and $\mathrm c$ referring to the quantum and classical comparisons. Heuristically, the first part of Theorem~\ref{thm:abstract-comparison} thus reduces to comparing $z_{\mathrm q}$ and $z_{\mathrm c}$.

Our first key observation concerns the upper-symbol representation of $\Gamma_{{\rm free}}$ introduced in \cite{Cub26}. For $u\in\gH$, we use the coherent vector $\xi(u)=e^{-\|u\|^2/2}\bigoplus_{n\geq0}u^{\otimes n}/\sqrt{n!}$. Following \cite[Lemma~6.2]{Cub26}, we represent the free Gibbs state by integration against the Gaussian measure $\mu_\lambda$:
\begin{equation}\label{eq:free-coherent-representation}
 \Gamma_{\rm free}
 =
 \int_{\gH}
 |\xi(u/\sqrt\lambda)\rangle\langle\xi(u/\sqrt\lambda)|\dd\mu_\lambda(u),\qquad
 \mathrm{Cov}(\mu_\lambda)=\mathsf C_\lambda:=\lambda(e^{\lambda h}-1)^{-1}
 =\lambda\Gamma_{\rm free}^{(1)}.
\end{equation}
The coherent-state projectors satisfy
\begin{align*}
 \big\|
 |\xi(u)\rangle\langle\xi(u)|
 -
 |\xi(u')\rangle\langle\xi(u')|
 \big\|_{\gS^1(\gF)}
 =
 2\sqrt{1-|\langle\xi(u),\xi(u')\rangle|^2}
 =
 2\sqrt{1-e^{-\|u-u'\|^2}}
 \leq 2\|u-u'\|.
\end{align*}
Thus, the integrals in \eqref{eq:free-coherent-representation} is understood as Bochner integral in the separable Banach space $\gS^1(\gF)$. Substituting this representation into \eqref{eq:GT-upper-symbol} gives
\begin{equation}\label{eq:coherent-mu-lambda-integral}
 z_{\mathrm q}:=\tr_{\gF}\big(\Gamma_{\rm free}e^{-\bW_\lambda^{\ren}}\big)=\int_{\gH}\inprod{\xi(u/\sqrt{\lambda})}{e^{-\mathbb W_\lambda^{{\rm ren}}}\xi(u/\sqrt{\lambda})}\dd\mu_\lambda(u).
\end{equation}
By the spectral theorem for the self-adjoint operator $e^{-\mathbb W_\lambda^{{\rm ren}}}$, we can interpret the right-hand side of \eqref{eq:coherent-mu-lambda-integral} as first sampling according to the spectral measure in the coherent state indexed by $u$, and then averaging $u$ against the Gaussian measure $\mu_\lambda$. Since $e^{-\mathbb W_\lambda^{{\rm ren}}}$ acts by multiplication on $\gF$ and multiplication operators commute, we can use a common spectral measure. We describe this measure on the \textit{grand-canonical configuration space\kern0pt{}} $\mathfrak C:=\{\varnothing\}\sqcup \bigsqcup_{n\geq1}(\bT^d)^n$, equiped with the disjoint-union topology and its Borel $\sigma$-algebra.

Fix $\lambda>0$ and $u\in\gH$. Expanding $\xi(u/\sqrt\lambda)$ over the particle sectors defines a measure on $\mathfrak C$, denoted by $\dd\rho_{\lambda,u}(\mathbf x)$:
\begin{equation}
 \int_{\mathfrak C}\Phi(\mathbf x)\dd\rho_{\lambda,u}(\mathbf x)
 =
 e^{-\|u\|^2/\lambda}
 \sum_{n=0}^{\infty}\frac1{n!\lambda^n}
 \int_{(\bT^d)^n}
 \Phi_n(x_1,\ldots,x_n)\prod_{j=1}^n|u(x_j)|^2
 \dd x_1\cdots\dd x_n.
 \label{eq:abstract-coherent-kernel}
\end{equation}
Here $\Phi$ is a bounded Borel function on $\mathfrak C$. For $u,u'\in\gH$, Cauchy--Schwarz gives
\begin{align*}
 \Big|
 \int_{\mathfrak C}\Phi(\mathbf x)\dd\rho_{\lambda,u}(\mathbf x)
 -
 \int_{\mathfrak C}\Phi(\mathbf x)\dd\rho_{\lambda,u'}(\mathbf x)
 \Big|
 \leq
 2\|\Phi\|_\infty
 \|\xi(u/\sqrt\lambda)-\xi(u'/\sqrt\lambda)\|_{\gF}
 \longrightarrow0
 \qquad\text{as }u'\to u\text{ in }\gH.
\end{align*}
Thus $u\mapsto\rho_{\lambda,u}$ is a Borel probability kernel from $\gH$ to $\mathfrak C$. 

For a bounded multiplication operator $A$ and a configuration $\mathbf x=(X_1,\ldots,X_n)\in\mathfrak C$, we represent the quantum mass as a function on $\mathfrak C$:
$$
 \mathbb M_\lambda^A=
 \mathbb M_\lambda^A(\mathbf x)
 =
 \lambda\sum_{j=1}^{|\mathbf x|}A(X_j)
 -\tr_{\gH}(A\mathsf C_\lambda),\qquad |\mathbf x|:=\begin{cases}
 &n,\ {\rm if\ }\mathbf x\in(\mathbb T^d)^n;\\
 &0,\ {\rm if\ }\mathbf x=\varnothing.
 \end{cases}
$$
Substituting into \eqref{eq:free-coherent-representation}, we obtain
$$
 z_{\mathrm q}
 =
 \int_{\gH}\int_{\mathfrak C}
 e^{-\bW_\lambda^{\ren}(\mathbf x)}
 \dd\rho_{\lambda,u}(\mathbf x)\dd\mu_\lambda(u).
$$

Our second key observation is that $\mathrm{Cov}(\mu_\lambda)=\mathsf C_\lambda$ converges to $\mathrm{Cov}(\mu_0)=h^{-1}=:\mathsf C_0$ in $\gS^2$. In this sense, $\mu_\lambda$ provides a regularized version of $\mu_0$ and suggests the following mass renormalization. For a bounded normal operator $A$ and $\lambda>0$, we define the $\lambda-$regularized mass by
$$
 \cM_\lambda^A[u]
 =
 \inprod{u}{Au}-\int\inprod{u}{Au}\dd\mu_\lambda(u)
 =
 \langle u,Au\rangle-\tr_{\gH}(A\mathsf C_\lambda).
$$
Since $\mathsf C_\lambda\in\gS^1$ for $\lambda>0$, this expression is defined pointwise on $\gH$. We form an intermediate interaction from these masses, retaining the quantum counterterms:
$$
 \cW_\lambda[u]
 =
 \frac12\sum_{k\in\Z^d}
 \widehat w(k)|\cM_\lambda^{e_k}[u]|^2
 -\alpha_\lambda\cM_\lambda^\1[u]-\beta_\lambda,
$$
We normalize its Gibbs weight by
$$
 \dd\nu_\lambda(u)
 =
 z_\lambda^{-1}e^{-\cW_\lambda[u]}\dd\mu_\lambda(u),
 \qquad
 z_\lambda=\int_{\gH}e^{-\cW_\lambda[u]}\dd\mu_\lambda(u).
$$
It follows that
$$
 z_{\mathrm q}-z_\lambda
 =
 \int_{\gH}\int_{\mathfrak C}
 \big(e^{-\bW_\lambda^{\ren}(\mathbf x)}-e^{-\mathcal W_\lambda[u]}\big)
 \dd\rho_{\lambda,u}(\mathbf x)\dd\mu_\lambda(u).
$$
The two masses satisfy
$$
 \int_{\mathfrak C}
 \big(\mathbb M_\lambda^A(\mathbf x)-\cM_\lambda^A[u]\big)
 \dd\rho_{\lambda,u}(\mathbf x)
 =
 \langle u,Au\rangle-\tr_{\gH}(A\mathsf C_\lambda)
 -\cM_\lambda^A[u]
 =0,
$$
which suggests that $z_{\mathrm q}\approx  z_\lambda$. We next compare $z_{\mathrm c}$ with $z_\lambda$.

To recover the reference law $\mu_0$ while retaining the same conditional expectation, we add an independent Gaussian increment. Using two independent copies $(g,g')\in\gX\times\gX$ of the Gaussian coordinates from Subsection~\ref{subsec:classical-field-theory}, we define
\begin{equation}
 u_\lambda(g)
 =
 \sum_{n\in\Z^d}\sqrt{c_\lambda(n)}\,g_ne_n,\qquad
 v_\lambda(g')
 =
 \sum_{n\in\Z^d}
 \sqrt{h(n)^{-1}-c_\lambda(n)}\,g_n'e_n,
 \label{eq:coupled-increment}
\end{equation}
where $c_\lambda(n)=\lambda/(e^{\lambda h(n)}-1)$. Thus $u_\lambda,v_\lambda$ are independent Gaussian fields with covariances $\mathsf C_\lambda$ and $\mathsf C_0-\mathsf C_\lambda\geq 0$, respectively, and
\begin{equation}\label{eq:Gaussian-coupling-law}
 \Law(u_\lambda+v_\lambda)=\mu_0.
\end{equation}
Sampling the configuration conditionally on $u_\lambda$ gives the following joint law on $\gX\times\gX\times\mathfrak C$:
$$
 \dd\widehat{\mathbb P}_\lambda(g,g',\mathbf x)
 =
 \dd\mathbb P(g)\dd\mathbb P(g')
 \dd\rho_{\lambda,u_\lambda(g)}(\mathbf x).
$$
From now on, $\mathbb E$ denotes expectation with respect to $\widehat{\mathbb P}_\lambda$, and $\mathbb E[\cdot\mid u_\lambda]$ denotes conditional expectation given $u_\lambda$. For a bounded multiplication operator $A$, we have (see Lemma~\ref{lem:barycenter-identity})
$$
 \mathbb E[\mathbb M_\lambda^A\mid u_\lambda]
 =
 \cM_\lambda^A[u_\lambda],\qquad
 \mathbb E[\cM_0^A[u_\lambda+v_\lambda]\mid u_\lambda]
 =
 \cM_\lambda^A[u_\lambda].
$$

\medskip

\subsection{Relative free energy}

We first outline the relative free-energy comparison in Theorem~\ref{thm:abstract-comparison}. With $\cW_\lambda[u_\lambda]$ as the reference interaction, we compare $z_{\mathrm q}$ and $z_{\mathrm c}$ with $z_\lambda$. For the quantum comparison, we introduce
\begin{equation}\label{eq:quantum-interaction-difference}
 \Delta_{\lambda,\mathrm q}
 =
 \bW_\lambda^{\ren}(\mathbf x)-\cW_\lambda[u_\lambda],
\end{equation}
and for the classical comparison we use
\begin{equation}\label{eq:classical-interaction-difference}
 \Delta_{\lambda,\mathrm c}
 =
 \cW_0[u_\lambda+v_\lambda]-\cW_\lambda[u_\lambda].
\end{equation}
With respect to $\widehat{\mathbb P}_\lambda$, the three normalized Gibbs weights are:
$$
 \begin{gathered}
 f_\lambda=z_\lambda^{-1}e^{-\cW_\lambda[u_\lambda]},\qquad
 f_{\mathrm q}=z_{\mathrm q}^{-1}e^{-\bW_\lambda^{\ren}(\mathbf x)},
 \qquad
 f_{\mathrm c}=z_{\mathrm c}^{-1}e^{-\cW_0[u_\lambda+v_\lambda]}.
 \end{gathered}
$$
Each density has expectation one. 

We recall the quantum and classical Gibbs variational principles. In the quantum setting,
$$
 -\log\frac{Z_\lambda}{Z_{{\rm free}}}
 =
 \inf_{\substack{\Gamma\geq0\\ \tr_{\gF}\Gamma=1}}
 \left\{
 \cH(\Gamma,\Gamma_{{\rm free}})
 +
 \tr_{\gF}\big(\bW_\lambda^{\ren}\Gamma\big)
 \right\},
$$
where the quantum relative entropy is
$$
 \cH(\Gamma,\Gamma')
 =
 \begin{cases}
 \tr_{\gF}\!\left[
 \Gamma\big(\log\Gamma-\log\Gamma'\big)
 \right],
 & \operatorname{supp}\Gamma\subseteq\operatorname{supp}\Gamma',
 \\[0.3em]
 +\infty,
 & \text{otherwise}.
 \end{cases}
$$
The corresponding classical variational principle is
$$
 -\log z_0
 =
 \inf_{\nu}
 \left\{
 \cH_{\cl}(\nu,\mu_0)
 +
 \int\cW_0[u]\dd\nu(u)
 \right\},
$$
where the infimum is taken over probability measures and
$$
 \cH_{\cl}(\nu,\mu)
 =
 \begin{cases}
 \displaystyle
 \int
 \log\!\left(\frac{\dd\nu}{\dd\mu}\right)\dd\nu,
 & \nu\ll\mu,
 \\[0.6em]
 +\infty,
 & \text{otherwise}.
 \end{cases}
$$
The minimizers in these two variational principles are the corresponding quantum Gibbs state and classical Gibbs measure, respectively.

For probability densities $f,f'$ on $(\gX\times\gX\times\mathfrak C,\widehat{\mathbb P}_\lambda)$, the definition gives
$$
 \cH_{\cl}\big(f\dd\widehat{\mathbb P}_\lambda,
 f'\dd\widehat{\mathbb P}_\lambda\big)
 =
 \mathbb E\big[f\log(f/f')\big].
$$

\smallskip \noindent\textbf{Step 1: Compare the Gibbs weights.} For $\diamond\in\{\mathrm q,\mathrm c\}$, we denote by $\mathcal E_{\lambda,\diamond}$ the sum of the relative entropies between $f_\lambda\dd\widehat{\mathbb P}_\lambda$ and $f_\diamond\dd\widehat{\mathbb P}_\lambda$ in the two directions. Adding the two entropy identities cancels the partition functions and gives
$$
 \mathcal E_{\lambda,\diamond}
 :=
 \cH_{\cl}\big(
 f_\lambda\dd\widehat{\mathbb P}_\lambda,
 f_\diamond\dd\widehat{\mathbb P}_\lambda\big)
 +
 \cH_{\cl}\big(
 f_\diamond\dd\widehat{\mathbb P}_\lambda,
 f_\lambda\dd\widehat{\mathbb P}_\lambda\big)
 =
 \Big|
 \mathbb E\big[
 (f_\lambda-f_\diamond)\Delta_{\lambda,\diamond}
 \big]
 \Big|.
$$
Nonnegativity of each relative entropy also gives
$$
 \Big|\log\frac{z_\diamond}{z_\lambda}\Big|
 \leq
 \max\Big\{
 \big|\mathbb E[f_\lambda\Delta_{\lambda,\diamond}]\big|,
 \big|\mathbb E[f_\diamond\Delta_{\lambda,\diamond}]\big|
 \Big\}.
$$
Proposition~\ref{prop:two-entropies} proves both statements.

\smallskip \noindent\textbf{Step 2: Return to the quantum Gibbs state.} Since $z_{\mathrm c}=z_0$, we have
$$
 \log\frac{Z_\lambda}{Z_{\rm free}}-\log z_0
 =
 \log\frac{z_{\mathrm q}}{z_\lambda}
 -
 \log\frac{z_{\mathrm c}}{z_\lambda}
 -
 \log\frac{Z_{\rm free}z_{\mathrm q}}{Z_\lambda}.
$$
Proposition~\ref{prop:quantum-entropy-interface} gives $ 0\leq \log\frac{Z_{\rm free}z_{\mathrm q}}{Z_\lambda} \leq \mathcal E_{\lambda,\mathrm q}, $ which implies
$$
 \Big|
 \log\frac{Z_\lambda}{Z_{\rm free}}-\log z_0
 \Big|
 \leq
 \mathcal E_{\lambda,\mathrm q}
 +
 \sum_{\diamond\in\{\mathrm q,\mathrm c\}}
 \Big|\log\frac{z_\diamond}{z_\lambda}\Big|.
$$
Step~1 controls every term on the right. H\"older's inequality then reduces the entropy and partition-function comparisons to moments of $\Delta_{\lambda,\diamond}$ and integrability of the normalized weights, all with respect to $\widehat{\mathbb P}_\lambda$.

\smallskip \noindent\textbf{Step 3: Estimate the energy differences.} We use the common conditional expectation to estimate $\Delta_{\lambda,\diamond}$. For a bounded multiplication operator $A$, define the quantum fluctuation
$$
 \mathcal R_{\lambda,\mathrm q}^A
 =
 \mathbb M_\lambda^A(\mathbf x)-\cM_\lambda^A[u_\lambda].
$$
For the classical mass, we similarly define
$$
 \mathcal R_{\lambda,\mathrm c}^A
 =
 \cM_0^A[u_\lambda+v_\lambda]-\cM_\lambda^A[u_\lambda].
$$
Expanding the energy differences $\Delta_{\lambda,\diamond}$ gives
\begin{align*}
 \Delta_{\lambda,\mathrm q}
 &=
 \operatorname{Re}\sum_k\widehat w(k)
 \overline{\cM_\lambda^{e_k}[u_\lambda]}\,
 \mathcal R_{\lambda,\mathrm q}^{e_k}
 +\frac12\sum_k\widehat w(k)
 |\mathcal R_{\lambda,\mathrm q}^{e_k}|^2
 -\alpha_\lambda\mathcal R_{\lambda,\mathrm q}^\1,\\
 \Delta_{\lambda,\mathrm c}
 &=
 \operatorname{Re}\sum_k\widehat w(k)
 \overline{\cM_\lambda^{e_k}[u_\lambda]}\,
 \mathcal R_{\lambda,\mathrm c}^{e_k}
 +\frac12\sum_k\widehat w(k)
 |\mathcal R_{\lambda,\mathrm c}^{e_k}|^2
 -\alpha_0\mathcal R_{\lambda,\mathrm c}^\1\\
 &\quad+
 (\alpha_\lambda-\alpha_0)\cM_\lambda^\1[u_\lambda]
 +\beta_\lambda-\beta_0.
\end{align*}
Thus, estimating $\Delta_{\lambda,\diamond}$ reduces to estimating $\mathcal R^A_{\lambda,\diamond}$; see Lemmas~\ref{lem:Lp-estimate-Delta} and~\ref{lem:abstract-single-lift-moments}.

\medskip

\subsection{Reduced density matrices}

We express the target correlation operators and the reduced density matrices of the trial state on the same probability space. For a field $u$ and a scalar weight $f$ on $(\gX\times\gX\times\mathfrak C,\widehat{\mathbb P}_\lambda)$, write
\begin{equation}\label{eq:weighted-correlation-common-space}
 \gamma_f^{(k)}[u]
 =
 \mathbb E\big[
 f\,|u^{\otimes k}\rangle\langle u^{\otimes k}|
 \big].
\end{equation}
We interpret this expression through matrix elements, as in \eqref{eq:abstract-nu-0-corroration-operator}. Section~\ref{sec:correlation-preliminaries} establishes the Hilbert--Schmidt realizations and estimates for these weighted correlation operators. Note that
$$
 k!\lambda^k\widetilde\Gamma_\lambda^{(k)}
 =
 \gamma_{\nu_\lambda}^{(k)}
 =
 \gamma_{f_\lambda}^{(k)}[u_\lambda],
$$
where
\begin{equation}\label{eq:abstract-trial-state}
 \widetilde\Gamma_\lambda
 :=
 \int_{\gH}
 |\xi(u/\sqrt\lambda)\rangle\langle\xi(u/\sqrt\lambda)|
 \dd\nu_\lambda(u).
\end{equation}
For the target measure, \eqref{eq:Gaussian-coupling-law} gives
$$
 \gamma_{\nu_0}^{(k)}
 =
 \gamma_{f_{\mathrm c}}^{(k)}[u_\lambda+v_\lambda].
$$
We decompose the difference into the quantum comparison, the comparison of classical weights, and the change of field:
\begin{align}
 k!\lambda^k\Gamma_\lambda^{(k)}-\gamma_{\nu_0}^{(k)}
 &=
 k!\lambda^k
 \big(\Gamma_\lambda^{(k)}-\widetilde\Gamma_\lambda^{(k)}\big)\nn\\
 &+
 \big(\gamma_{f_\lambda}^{(k)}[u_\lambda+v_\lambda]
 -
 \gamma_{f_{\mathrm c}}^{(k)}[u_\lambda+v_\lambda]\big)
 +
 \big(\gamma_{f_\lambda}^{(k)}[u_\lambda]
 -
 \gamma_{f_\lambda}^{(k)}[u_\lambda+v_\lambda]\big).
 \label{eq:HS-three-comparisons}
\end{align}

\smallskip \noindent\textbf{Step 1: A priori bounds.} Proposition~\ref{prop:quantum-kernel-bound} and the entropy bound from Step~2 of the free-energy comparison give
$$
 k!\lambda^k\|\Gamma_\lambda^{(k)}\|_{\gS^2}
 \leq
 e^{kb_\lambda}
 \frac{Z_{\rm free}z_{\mathrm q}}{Z_\lambda}
 \|\gamma_{f_{\mathrm q}}^{(k)}[u_\lambda]\|_{\gS^2}
 \leq
 e^{kb_\lambda+\mathcal E_{\lambda,\mathrm q}}
 \|\gamma_{f_{\mathrm q}}^{(k)}[u_\lambda]\|_{\gS^2}.
$$
Thus the weighted correlation operators control the rescaled quantum reduced density matrices. Together with Proposition~\ref{prop:moment-propagation}, this gives
$$
 k!\lambda^k\|\Gamma_\lambda^{(k)}\|_{\gS^2}+\big\|\gamma_{f_\lambda}^{(k)}[u_\lambda+v_\lambda]\big\|_{\gS^2}
 +\big\|\gamma_{f_\lambda}^{(k)}[u_\lambda]\big\|_{\gS^2}
 +\big\|\gamma_{f_{\mathrm q}}^{(k)}[u_\lambda]\big\|_{\gS^2}\lesssim\big\|\gamma_{\nu_0}^{(k)}\big\|_{\gS^2}.
$$

\smallskip \noindent\textbf{Step 2: Compare the quantum reduced density matrices.} Both $\Gamma_\lambda$ and $\widetilde\Gamma_\lambda$ commute with $\cN$, so we apply \cite[Lemma~11.4]{LNR21} and Pinsker's inequality to obtain
\begin{equation}
 \big\|k!\lambda^k
 (\Gamma_\lambda^{(k)}-\widetilde\Gamma_\lambda^{(k)})
 \big\|_{\gS^2}^2
 \lesssim_k
 \cH(\widetilde\Gamma_\lambda,\Gamma_\lambda)^{1/2}
 \sum_{\ell=k}^{2k}\lambda^{2k-\ell}
 \Big[
 \lambda^\ell\|\Gamma_\lambda^{(\ell)}\|_{\gS^2}
 +\frac1{\ell!}
 \|\gamma_{f_\lambda}^{(\ell)}[u_\lambda]\|_{\gS^2}
 \Big].\label{eq:LNR-dentity-operator-lemma}
\end{equation}
Proposition~\ref{prop:quantum-entropy-interface} gives $ 0\leq \cH(\widetilde\Gamma_\lambda,\Gamma_\lambda) \leq \mathcal E_{\lambda,\mathrm q}. $ This controls the first difference in \eqref{eq:HS-three-comparisons}.

\smallskip \noindent\textbf{Step 3: Compare the classical weights.} Applying Lemma~\ref{lem:classical-entropy-HS}, a classical counterpart of\kern0pt{} \eqref{eq:LNR-dentity-operator-lemma}, we obtain
$$
 \big\|
 \gamma_{f_\lambda}^{(k)}[u_\lambda+v_\lambda]
 -
 \gamma_{f_{\mathrm c}}^{(k)}[u_\lambda+v_\lambda]\big\|_{\gS^2}^2
 \leq
 2\cH_{\cl}\big(
 f_\lambda\dd\widehat{\mathbb P}_\lambda,
 f_{\mathrm c}\dd\widehat{\mathbb P}_\lambda\big)
 \big\|
 \gamma_{f_\lambda}^{(2k)}[u_\lambda+v_\lambda]
 +
 \gamma_{f_{\mathrm c}}^{(2k)}[u_\lambda+v_\lambda]
 \big\|_{\gS^2}.
$$
We bound the entropy by $\mathcal E_{\lambda,\mathrm c}$ using the free-energy comparison.

\noindent\textbf{Step 4: Estimate the change of field.} Keeping $f_\lambda$ fixed, we average over the independent Gaussian increment. For the first-order correlation operator,
$$
 \gamma_{f_\lambda}^{(1)}[u_\lambda+v_\lambda]
 -
 \gamma_{f_\lambda}^{(1)}[u_\lambda]
 =
 \mathsf C_0-\mathsf C_\lambda.
$$
For general $k$, Lemma~\ref{lem:weighted-Gaussian-addition} gives
$$
 \big\|
 \gamma_{f_\lambda}^{(k)}[u_\lambda+v_\lambda]
 -
 \gamma_{f_\lambda}^{(k)}[u_\lambda]
 \big\|_{\gS^2}
 \lesssim_k
 \sum_{j=1}^k
 \|\gamma_{f_\lambda}^{(k-j)}[u_\lambda]\|_{\gS^2}
 \|\mathsf C_0-\mathsf C_\lambda\|_{\gS^2}^{j}.
$$
We use $\gamma_{f_\lambda}^{(0)}[u_\lambda]=1$. The covariance estimate $\|\mathsf C_0-\mathsf C_\lambda\|_{\gS^2}^2 \lesssim_d\lambda^2N_{\rm free}\to0$ controls this field difference.

Combining these estimates with \eqref{eq:HS-three-comparisons} completes the density-matrix comparison. \Needspace{17\baselineskip}

\bigskip

\section{Entropy identities}\label{sec:entropy-identities}

We prove the two entropy comparisons used in Section~\ref{sec:proof-strategy}.

\Needspace{13\baselineskip}
\begin{proposition}[Comparison of the Gibbs weights] \label{prop:two-entropies} For $\diamond\in\{\mathrm q,\mathrm c\}$, we have
\begin{equation}
 \mathcal E_{\lambda,\diamond}
 :=
 \cH_{\cl}\big(
 f_\lambda\dd\widehat{\mathbb P}_\lambda,
 f_\diamond\dd\widehat{\mathbb P}_\lambda\big)
 +
 \cH_{\cl}\big(
 f_\diamond\dd\widehat{\mathbb P}_\lambda,
 f_\lambda\dd\widehat{\mathbb P}_\lambda\big)
 =
 \big|
 \mathbb E[(f_\lambda-f_\diamond)\Delta_{\lambda,\diamond}]
 \big|.
 \label{eq:symmetric-entropy-error}
\end{equation}
Moreover,
\begin{equation}\label{eq:partition-from-entropy}
 \Big|\log\frac{z_\diamond}{z_\lambda}\Big|
 \leq
 \max\Big\{
 \big|\mathbb E[f_\lambda\Delta_{\lambda,\diamond}]\big|,
 \big|\mathbb E[f_\diamond\Delta_{\lambda,\diamond}]\big|
 \Big\}.
\end{equation}
\end{proposition}
\begin{proof}
The definitions give $ \log\frac{f_\lambda}{f_\diamond} = \log\frac{z_\diamond}{z_\lambda} +\Delta_{\lambda,\diamond}, $ and hence
$$
 \cH_{\cl}\big(
 f_\lambda\dd\widehat{\mathbb P}_\lambda,
 f_\diamond\dd\widehat{\mathbb P}_\lambda\big)=
 \log\frac{z_\diamond}{z_\lambda}
 +\mathbb E[f_\lambda\Delta_{\lambda,\diamond}].
$$
Interchanging the two densities gives
$$
 \cH_{\cl}\big(
 f_\diamond\dd\widehat{\mathbb P}_\lambda,
 f_\lambda\dd\widehat{\mathbb P}_\lambda\big)
 =
 \log\frac{z_\lambda}{z_\diamond}
 -\mathbb E[f_\diamond\Delta_{\lambda,\diamond}].
$$
Adding the two identities proves \eqref{eq:symmetric-entropy-error}. Nonnegativity of the two relative entropies also gives
$$
 -\mathbb E[f_\lambda\Delta_{\lambda,\diamond}]
 \leq
 \log\frac{z_\diamond}{z_\lambda}
 \leq
 -\mathbb E[f_\diamond\Delta_{\lambda,\diamond}],
$$
which proves \eqref{eq:partition-from-entropy}.
\end{proof}

\begin{proposition}[Quantum relative entropy and relative free energy] \label{prop:quantum-entropy-interface} We have
\begin{equation}
 \cH_{\cl}\big(
 f_\lambda\dd\widehat{\mathbb P}_\lambda,
 f_{\mathrm q}\dd\widehat{\mathbb P}_\lambda\big)
 =
 \cH(\widetilde\Gamma_\lambda,\Gamma_\lambda)
 +
 \Big[
 \cH_{\cl}(\nu_\lambda,\mu_\lambda)
 -
 \cH(\widetilde\Gamma_\lambda,\Gamma_{\rm free})
 \Big]
 +
 \Big[\log z_{\mathrm q}-\log\frac{Z_\lambda}{Z_{{\rm free}}}\Big].
 \label{eq:three-nonnegative-terms}
\end{equation}
The three terms on the right are nonnegative. In particular,
\begin{equation}\label{eq:quantum-entropy-control}
 0\leq
 \cH(\widetilde\Gamma_\lambda,\Gamma_\lambda)
 +
 \Big[\log z_{\mathrm q}-\log\frac{Z_\lambda}{Z_{{\rm free}}}\Big]
 \leq
 \mathcal E_{\lambda,\mathrm q}
 =
 \big|
 \mathbb E[(f_\lambda-f_{\mathrm q})\Delta_{\lambda,\mathrm q}]
 \big|.
\end{equation}
Consequently,
\begin{equation}\label{eq:free-energy-from-partitions}
 \Big|
 \log\frac{Z_\lambda}{Z_{\rm free}}-\log z_0
 \Big|
 \leq
 \big|
 \mathbb E[(f_\lambda-f_{\mathrm q})\Delta_{\lambda,\mathrm q}]
 \big|
 +
 \sum_{\diamond\in\{\mathrm q,\mathrm c\}}
 \max\Big\{
 \big|\mathbb E[f_\lambda\Delta_{\lambda,\diamond}]\big|,
 \big|\mathbb E[f_\diamond\Delta_{\lambda,\diamond}]\big|
 \Big\}.
\end{equation}
\end{proposition}
\begin{proof}
By \eqref{eq:abstract-trial-state},
$$
 \mathbb E[f_\lambda\bW_\lambda^{\ren}(\mathbf x)]
 =
 \tr_{\gF}(\bW_\lambda^{\ren}\widetilde\Gamma_\lambda).
$$
Expanding the logarithms of the Gibbs densities and states, we obtain
\begin{align*}
 \cH_{\cl}\big(
 f_\lambda\dd\widehat{\mathbb P}_\lambda,
 f_{\mathrm q}\dd\widehat{\mathbb P}_\lambda\big)
 &=
 \mathbb E[f_\lambda\log f_\lambda]
 +
 \tr_{\gF}(\bW_\lambda^{\ren}\widetilde\Gamma_\lambda)
 +
 \log z_{\mathrm q}\\
 &=
 \cH(\widetilde\Gamma_\lambda,\Gamma_\lambda)
 +
 \mathbb E[f_\lambda\log f_\lambda]
 -
 \cH(\widetilde\Gamma_\lambda,\Gamma_{\rm free})
 +
 \log\frac{Z_{\rm free}z_{\mathrm q}}{Z_\lambda}.
\end{align*}
Together with $\mathbb E[f_\lambda\log f_\lambda] =\cH_{\cl}(\nu_\lambda,\mu_\lambda)$, this proves \eqref{eq:three-nonnegative-terms}. The second Berezin--Lieb inequality in Lemma~\ref{lem:relative-entropy-Berezin-Lieb} and the Golden--Thompson inequality give
$$
 \cH(\widetilde\Gamma_\lambda,\Gamma_{\rm free})
 \overset{{\rm BL}}{\leq}
 \cH_{\cl}(\nu_\lambda,\mu_\lambda),\qquad
 \log z_{\mathrm q}-\log\frac{Z_\lambda}{Z_{{\rm free}}}\overset{{\rm GT}}{\geq} 0.
$$
Hence \eqref{eq:quantum-entropy-control} follows. Finally,
$$
 \log\frac{Z_\lambda}{Z_{\rm free}}-\log z_0
 =
 \log\frac{z_{\mathrm q}}{z_\lambda}
 -
 \log\frac{z_{\mathrm c}}{z_\lambda}
 -
 \log\frac{Z_{\rm free}z_{\mathrm q}}{Z_\lambda},
$$
and Proposition~\ref{prop:two-entropies} gives \eqref{eq:free-energy-from-partitions}.
\end{proof}

\bigskip

\section{%
\texorpdfstring{Estimates of $\mathcal E_{\lambda,\diamond}$ and relative free-energy convergence}{Estimates of mathcal E and relative free-energy convergence}%
}\label{sec:free-energy-estimates}

We establish the following bound on the error terms introduced in Section~\ref{sec:proof-strategy}.

\begin{proposition}[Estimates of $\mathcal E_{\lambda,\diamond}$ and relative free-energy convergence] \label{prop:estimate-E-lambda-diamond} For $0<\lambda\leq1$ and $\diamond\in\{\mathrm q,\mathrm c\}$,
\begin{equation}\label{eq:free-energy-rate}
 \big|
 \mathbb E[(f_\lambda-f_\diamond)\Delta_{\lambda,\diamond}]
 \big|
 +
 \sum_{\diamond\in\{\mathrm q,\mathrm c\}}
 \max\Big\{
 \big|\mathbb E[f_\lambda\Delta_{\lambda,\diamond}]\big|,
 \big|\mathbb E[f_\diamond\Delta_{\lambda,\diamond}]\big|
 \Big\}\lesssim_{d,c_0}\mathfrak a_\lambda^2\mathfrak r_\lambda.
\end{equation}
In particular,
\begin{equation}\label{eq:entropy-rate}
 \mathcal E_{\lambda,\diamond}
 +\Big|\log\frac{z_\diamond}{z_\lambda}\Big|
 \lesssim_{d,c_0}\mathfrak a_\lambda^2\mathfrak r_\lambda,\qquad
 \Big|
 \log\frac{Z_\lambda}{Z_{\rm free}}-\log z_0
 \Big|
 \lesssim_{d,c_0} \mathfrak a_\lambda^2\mathfrak r_\lambda.
\end{equation}
\end{proposition}

We begin by estimating the two fluctuations. Recall that
$$
 \mathcal R_{\lambda,\mathrm q}^A
 =
 \mathbb M_\lambda^A(\mathbf x)-\cM_\lambda^A[u_\lambda];\qquad
 \mathcal R_{\lambda,\mathrm c}^A
 =
 \cM_0^A[u_\lambda+v_\lambda]-\cM_\lambda^A[u_\lambda],
$$
where $A$ is a bounded multiplication operator.

\begin{lemma}\label{lem:barycenter-identity} Let $\lambda>0$ and $\diamond\in\{\mathrm q,\mathrm c\}$. Then
\begin{equation}\label{eq:q-c-fluctuation-identity}
 \mathbb E[\mathcal R_{\lambda,\diamond}^A\mid u_\lambda]=0,
 \qquad \diamond\in\{\mathrm q,\mathrm c\}
\end{equation}
for every bounded multiplication operator $A$. In particular,
\begin{equation}\label{eq:abstract-barycenter-identities}
 \mathbb E|\mathcal R_{\lambda,\mathrm c}^A|^2
 =
 \mathbb E\big|\cM_0^A[u_\lambda+v_\lambda]\big|^2
 -
 \mathbb E\big|\cM_\lambda^A[u_\lambda]\big|^2.
\end{equation}
\end{lemma}

\begin{proof}
The quantum case of\kern0pt{} \eqref{eq:q-c-fluctuation-identity} follows from the identity
$$
 \inprod{\xi(u/\sqrt{\lambda})}{\lambda\dG(A)\xi(u/\sqrt{\lambda})}=\inprod{u}{Au}
$$
for every bounded operator $A$ and every $u\in\gH$.

For the classical identity, take $\phi,\psi\in\gH^1$. Independence in \eqref{eq:coupled-increment} and the rank-one formula of Lemma~\ref{lem:Gaussian-centered-quadratic-moments} give
\begin{equation}
 \mathbb E\big[
 \cM_0^{|\phi\rangle\langle\psi|}[u_\lambda+v_\lambda]
 \mid u_\lambda\big]
 =
 \langle u_\lambda,\phi\rangle\langle\psi,u_\lambda\rangle
 +\langle\psi,(\mathsf C_0-\mathsf C_\lambda)\phi\rangle
 -\langle\psi,\mathsf C_0\phi\rangle
 =
 \cM_\lambda^{|\phi\rangle\langle\psi|}[u_\lambda].\label{eq:finite-rank-identity-mass}
\end{equation}
Moreover, the vectors $e_n$ are eigenvectors of $\mathsf C_\lambda$, and the matrix units $|e_m\rangle\langle e_n|$ form an orthonormal basis of $\gS^2$. We can choose finite linear combinations $A_j$ of these matrix units such that $ \big\|\mathsf C_0^{1/2}(A-A_j) \mathsf C_0^{1/2}\big\|_{\gS^2}\leq j^{-1}. $ Since $e_n\in\gH^1$, the finite-rank identity \eqref{eq:finite-rank-identity-mass} applies to each $A_j$.

Using that identity, the $L^2$-contractivity of conditional expectation, and \eqref{eq:Gaussian-mass-limit-Lp}, we obtain
\begin{align*}
 &\Big\|
 \mathbb E[\cM_0^A[u_\lambda+v_\lambda]\mid u_\lambda]
 -
 \cM_\lambda^A[u_\lambda]
 \Big\|_{L^2(\widehat{\mathbb P}_\lambda)}
 =
 \Big\|
 \mathbb E[\cM_0^{A-A_j}[u_\lambda+v_\lambda]\mid u_\lambda]
 -
 \cM_\lambda^{A-A_j}[u_\lambda]
 \Big\|_{L^2(\widehat{\mathbb P}_\lambda)}\\
 &\leq
 \big\|\cM_0^{A-A_j}[u_\lambda+v_\lambda]
 \big\|_{L^2(\widehat{\mathbb P}_\lambda)}
 +
 \big\|\cM_\lambda^{A-A_j}[u_\lambda]
 \big\|_{L^2(\widehat{\mathbb P}_\lambda)}
 =
 \big\|\cM_0^{A-A_j}\big\|_{L^2(\mu_0)}
 +
 \big\|\cM_\lambda^{A-A_j}\big\|_{L^2(\mu_\lambda)}\\
 &=
 \big\|\mathsf C_0^{1/2}(A-A_j)
 \mathsf C_0^{1/2}\big\|_{\gS^2}
 +
 \big\|\mathsf C_\lambda^{1/2}(A-A_j)
 \mathsf C_\lambda^{1/2}\big\|_{\gS^2}
 \leq \frac2j\longrightarrow0.
\end{align*}
In the second equality, we used $\Law(u_\lambda+v_\lambda)=\mu_0$ and $\Law(u_\lambda)=\mu_\lambda$. It follows that
$$
 \mathbb E[\cM_0^A[u_\lambda+v_\lambda]\mid u_\lambda]
 =
 \cM_\lambda^A[u_\lambda],
$$
and hence $\mathbb E[\mathcal R_{\lambda,\mathrm c}^A\mid u_\lambda]=0$. Equation~\eqref{eq:abstract-barycenter-identities} follows from the orthogonal-projection property of conditional expectation in $L^2$.
\end{proof}

\begin{lemma}[Fluctuation moment estimates] \label{lem:abstract-single-lift-moments} Let $\lambda>0$ and let $A$ be a bounded multiplication operator on $\gH$. For every $1\leq p<\infty$,
\begin{equation}
 \Big(
 \mathbb E\big[
 |\mathcal R_{\lambda,\mathrm q}^A|^{2p}
 \,\big|\,u_\lambda
 \big]
 \Big)^{1/(2p)}
 \lesssim
 \Big[
 \sqrt p\,
 \big(\lambda\langle u_\lambda,|A|^2u_\lambda\rangle\big)^{1/2}
 +p\lambda\|A\|
 \Big]
 \qquad\text{almost surely}.
 \label{eq:abstract-single-quantum-lift-moment}
\end{equation}
Consequently,
\begin{equation}\label{eq:abstract-quantum-lift-moment}
 \big\|\mathcal R_{\lambda,\mathrm q}^A
 \big\|_{L^{2p}(\widehat{\mathbb P}_\lambda)}
 \lesssim
 p\Big[
 \lambda\tr_{\gH}(|A|^2\mathsf C_\lambda)
 +\lambda^2\|A\|^2
 \Big]^{1/2}.
\end{equation}
For the classical fluctuation,
\begin{equation}
 \big\|\mathcal R_{\lambda,\mathrm c}^A
 \big\|_{L^{2p}(\widehat{\mathbb P}_\lambda)}
 \lesssim
 p\Big[
 \big\|\mathsf C_0^{1/2}A\mathsf C_0^{1/2}\big\|_{\gS^2}^2
 -
 \big\|\mathsf C_\lambda^{1/2}A
 \mathsf C_\lambda^{1/2}\big\|_{\gS^2}^2
 \Big]^{1/2}.
 \label{eq:abstract-classical-lift-moment}
\end{equation}
\end{lemma}

\begin{proof}
Assume first that $A=A^*$. By \eqref{eq:abstract-coherent-kernel},
\begin{equation}
 \mathbb E\big[
 e^{t\mathcal R_{\lambda,\mathrm q}^A}
 \,\big|\,u_\lambda
 \big]
 =
 \exp\Big[
 \lambda^{-1}
 \big\langle u_\lambda,
 (e^{\lambda tA}-1-\lambda tA)u_\lambda
 \big\rangle
 \Big],
 \qquad t\in\R.
 \label{eq:abstract-quantum-lift-mgf}
\end{equation}
Using $0\leq e^x-1-x\leq x^2$ for $|x|\leq1/2$, we obtain
\begin{equation}\label{eq:abstract-quantum-lift-mgf-bound}
 \mathbb E\big[
 e^{t\mathcal R_{\lambda,\mathrm q}^A}
 \,\big|\,u_\lambda
 \big]
 \leq
 e^{\lambda t^2\langle u_\lambda,A^2u_\lambda\rangle},
 \qquad |t|\lambda\|A\|\leq\frac12.
\end{equation}
If $\langle u_\lambda,A^2u_\lambda\rangle=0$, then \eqref{eq:abstract-quantum-lift-mgf} gives $\mathcal R_{\lambda,\mathrm q}^A=0$ conditionally on $u_\lambda$. Otherwise, for $R>0$, choose $ t= \min\Big\{ \frac{R}{4\lambda\langle u_\lambda,A^2u_\lambda\rangle}, \frac1{2\lambda\|A\|} \Big\}. $ Then
\begin{equation}
 -tR
 +
 \lambda t^2
 \inprod{u_\lambda}{A^2u_\lambda}
 \leq
 -\frac18
 \min\Big\{
 \frac{R^2}{
 \lambda\inprod{u_\lambda}{A^2u_\lambda}
 },
 \frac{R}{\lambda\norm{A}}
 \Big\}.
 \label{eq:abstract-quantum-lift-Chernoff-exponent}
\end{equation}
Applying the exponential form of Markov's inequality to $\mathcal R_{\lambda,\mathrm q}^A$ and $-\mathcal R_{\lambda,\mathrm q}^A$, and using \eqref{eq:abstract-quantum-lift-mgf-bound}--\eqref{eq:abstract-quantum-lift-Chernoff-exponent}, we obtain
$$
 \mathbb E
 \big[
 \1_{\{
 |\mathcal R_{\lambda,\mathrm q}^A|\geq R
 \}}
 \,\big|\,
 u_\lambda
 \big]
 \leq
 2\exp\Big[
 -\frac18
 \min\Big\{
 \frac{R^2}{
 \lambda\inprod{u_\lambda}{A^2u_\lambda}
 },
 \frac{R}{\lambda\norm{A}}
 \Big\}
 \Big].
$$
The layer-cake formula and $e^{-\min\{a,b\}}\leq e^{-a}+e^{-b}$ now give
\begin{align*}
 &
 \mathbb E\big[
 |\mathcal R_{\lambda,\mathrm q}^A|^{2p}
 \,\big|\,u_\lambda
 \big]
 =
 2p
 \int_0^\infty
 R^{2p-1}
 \mathbb E
 \big[
 \1_{\{
 |\mathcal R_{\lambda,\mathrm q}^A|\geq R
 \}}
 \,\big|\,
 u_\lambda
 \big]
 \dd R
 \\
 &\leq
 4p
 \int_0^\infty
 R^{2p-1}
 \exp\Big[
 -\frac{R^2}{
 8\lambda\inprod{u_\lambda}{A^2u_\lambda}
 }
 \Big]
 \dd R
 +
 4p
 \int_0^\infty
 R^{2p-1}
 \exp\Big[
 -\frac{R}{
 8\lambda\norm{A}
 }
 \Big]
 \dd R
 \\
 &=
 2p\,8^p\Gamma(p)
 \left(
 \lambda\inprod{u_\lambda}{A^2u_\lambda}
 \right)^p
 +
 4p\,8^{2p}\Gamma(2p)
 \lambda^{2p}\norm{A}^{2p}.
\end{align*}
Taking the $2p$-th root and using $\Gamma(p+1)^{1/(2p)}\lesssim p^{1/2}$ and $\Gamma(2p+1)^{1/(2p)}\lesssim p$, we obtain \eqref{eq:abstract-single-quantum-lift-moment} for self-adjoint $A$.

For a multiplication operator with a complex-valued symbol, use $ \mathcal R_{\lambda,\mathrm q}^A = \mathcal R_{\lambda,\mathrm q}^{\operatorname{Re}A} + \ii\mathcal R_{\lambda,\mathrm q}^{\operatorname{Im}A}. $ The conditional triangle inequality, $(\operatorname{Re}A)^2+(\operatorname{Im}A)^2=|A|^2$, and $\|\operatorname{Re}A\|,\|\operatorname{Im}A\|\leq\|A\|$ give \eqref{eq:abstract-single-quantum-lift-moment}.

Applying Lemma~\ref{lem:Gaussian-centered-quadratic-moments} to $|A|^2$, we obtain
\begin{align*}
 \norm{
 \mathcal R_{\lambda,\mathrm q}^A
 }_{L^{2p}(\widehat{\mathbb P}_\lambda)}^2
 &=
 \big(
 \mathbb E
 \left[
 \mathbb E\big[
 |\mathcal R_{\lambda,\mathrm q}^A|^{2p}
 \,\big|\,u_\lambda
 \big]\right]\big)^{1/p}
 \lesssim
 p\lambda
 \big\|\inprod{u_\lambda}{|A|^2 u_\lambda}\big\|_{L^p(\widehat{\mathbb P}_\lambda)}
 +
 p^2\lambda^2\norm{A}^2
 \\
 &\lesssim
 p\lambda \Big[\big\|\mathcal M_\lambda^{|A|^2}[u_\lambda]\big\|_{L^p(\widehat{\mathbb P}_\lambda)}+\tr_{\gH}(|A|^2\mathsf C_\lambda)\Big]+p^2\lambda^2\norm{A}^2\\
 &\lesssim
 p\lambda \Big[p\big\|\mathsf C_\lambda^{\frac12}|A|^2\mathsf C_\lambda^{\frac12}\big\|_{\gS^2}+\tr_{\gH}(|A|^2\mathsf C_\lambda)\Big]+p^2\lambda^2\norm{A}^2\\
 &\lesssim
 p^2\Big[
 \lambda\tr_{\gH}(|A|^2\mathsf C_\lambda)
 +\lambda^2\|A\|^2
 \Big]
\end{align*}
In the last step, we used positivity to bound $\|\mathsf C_\lambda^{1/2}|A|^2\mathsf C_\lambda^{1/2}\|_{\gS^2} \leq\tr_{\gH}(|A|^2\mathsf C_\lambda)$. This proves \eqref{eq:abstract-quantum-lift-moment}.

For the classical fluctuation, \eqref{eq:Gaussian-mass-limit-Lp} gives
\begin{equation}
 \mathbb E|\mathcal R_{\lambda,\mathrm c}^A|^2
 =
 \mathbb E\big|\cM_0^A[u_\lambda+v_\lambda]\big|^2
 -
 \mathbb E\big|\cM_\lambda^A[u_\lambda]\big|^2
 =
 \big\|\mathsf C_0^{1/2}A\mathsf C_0^{1/2}\big\|_{\gS^2}^2
 -
 \big\|\mathsf C_\lambda^{1/2}A
 \mathsf C_\lambda^{1/2}\big\|_{\gS^2}^2.
 \label{eq:abstract-classical-lift-variance}
\end{equation}
The fields $u_\lambda+v_\lambda$ and $u_\lambda$ are jointly Gaussian under $\widehat{\mathbb P}_\lambda$. Applying\kern0pt{} \eqref{eq:Gaussian-joint-quadratic-moment} to this pair gives
$$
 \big\|\mathcal R_{\lambda,\mathrm c}^A
 \big\|_{L^{2p}(\widehat{\mathbb P}_\lambda)}
 \lesssim
 p\big\|\mathcal R_{\lambda,\mathrm c}^A
 \big\|_{L^2(\widehat{\mathbb P}_\lambda)}.
$$
Together with \eqref{eq:abstract-classical-lift-variance}, this proves \eqref{eq:abstract-classical-lift-moment}.
\end{proof}

\begin{lemma}\label{lem:Lp-estimate-Delta} For $0<\lambda\leq1$, $p\geq1$, and $\diamond\in\{\mathrm q,\mathrm c\}$, every bounded multiplication operator $A$ satisfies
\begin{equation}\label{eq:common-mass-fluctuation-bound}
 \big\|\mathcal R_{\lambda,\diamond}^A\big\|_{L^{2p}}
 \lesssim_d
 p\|A\|_{L^\infty(\bT^d)}
 (\lambda^2N_{\rm free})^{1/2}.
\end{equation}
Moreover,
\begin{equation}\label{eq:abstract-interaction-difference-moment}
 \big\|\Delta_{\lambda,\diamond}\big\|_{L^p(\widehat{\mathbb P}_\lambda)}
 \lesssim_d p^2\mathfrak r_\lambda,
\end{equation}
where $\mathfrak r_\lambda$ is defined in \eqref{eq:comparison-r}.
\end{lemma}

\begin{proof}
We abbreviate $L^p=L^p(\widehat{\mathbb P}_\lambda)$. Using $\lambda\tr_{\gH}(|A|^2\mathsf C_\lambda) \leq\lambda^2N_{\rm free}\|A\|_\infty^2$ and \eqref{eq:covariance-multiplication-bound}, we obtain \eqref{eq:common-mass-fluctuation-bound} from Lemma~\ref{lem:abstract-single-lift-moments}. 

Minkowski's inequality therefore gives
\begin{equation}
 \Big\|\sum_k\widehat w(k)
 |\mathcal R_{\lambda,\diamond}^{e_k}|^2\Big\|_{L^p}
 \leq
 \sum_k\widehat w(k)
 \big\|\mathcal R_{\lambda,\diamond}^{e_k}\big\|_{L^{2p}}^2
 \lesssim_d
 p^2\lambda^2N_{\rm free}
 \frac1{(2\pi)^d}\sum_k\widehat w(k)
 =
 p^2\lambda^2N_{\rm free}w(0).
 \label{eq:aggregate-fluctuation-moment}
\end{equation}
For the intermediate renormalized masses, \eqref{eq:Gaussian-mass-limit-Lp} yields
\begin{align}
 \Big\|\sum_k\widehat w(k)
 |\cM_\lambda^{e_k}[u_\lambda]|^2\Big\|_{L^p}
 &\lesssim
 p^2\sum_k\widehat w(k)
 \big\|\mathsf C_\lambda^{1/2}e_k
 \mathsf C_\lambda^{1/2}\big\|_{\gS^2}^2
 =
 \frac{p^2}{(2\pi)^d}
 \sum_{k,j}\widehat w(k)c_\lambda(j)c_\lambda(j+k)\nn\\
 &\leq\frac{p^2}{(2\pi)^{d}}\sum_{k,j\in\Z^d}
 \frac{\widehat w(k)}{h(j)h(j+k)}
 =(2\pi)^dp^2 E.
 \label{eq:aggregate-intermediate-moment}
\end{align}
Similarly, $\|\cM_\lambda^\1[u_\lambda]\|_{L^p}\lesssim_dp$.

Expanding the interactions in \eqref{eq:quantum-interaction-difference} and \eqref{eq:classical-interaction-difference}, we obtain
\begin{align*}
 \Delta_{\lambda,\mathrm q}
 &=
 \operatorname{Re}\sum_k\widehat w(k)
 \overline{\cM_\lambda^{e_k}[u_\lambda]}\,
 \mathcal R_{\lambda,\mathrm q}^{e_k}
 +\frac12\sum_k\widehat w(k)
 |\mathcal R_{\lambda,\mathrm q}^{e_k}|^2
 -\alpha_\lambda\mathcal R_{\lambda,\mathrm q}^\1,\\
 \Delta_{\lambda,\mathrm c}
 &=
 \operatorname{Re}\sum_k\widehat w(k)
 \overline{\cM_\lambda^{e_k}[u_\lambda]}\,
 \mathcal R_{\lambda,\mathrm c}^{e_k}
 +\frac12\sum_k\widehat w(k)
 |\mathcal R_{\lambda,\mathrm c}^{e_k}|^2
 -\alpha_0\mathcal R_{\lambda,\mathrm c}^\1\\
 &\quad+
 (\alpha_\lambda-\alpha_0)\cM_\lambda^\1[u_\lambda]
 +\beta_\lambda-\beta_0.
\end{align*}
The sums converge in $L^p$ by \eqref{eq:aggregate-fluctuation-moment} and \eqref{eq:aggregate-intermediate-moment}. H\"older's inequality in the sum over $k$ gives
$$
 \big\|
 \operatorname{Re}\sum_k\widehat w(k)
 \overline{\cM_\lambda^{e_k}[u_\lambda]}\,
 \mathcal R_{\lambda,\diamond}^{e_k}
 \big\|_{L^p}\leq
 \big\|\sum_k\widehat w(k)|\cM_\lambda^{e_k}[u_\lambda]|^2
 \big\|_{L^p}^{1/2}
 \big\|\sum_k\widehat w(k)|\mathcal R_{\lambda,\diamond}^{e_k}|^2
 \big\|_{L^p}^{1/2}
 \lesssim_d p^2 E^{1/2}\mathfrak e_\lambda^{1/2}.
$$
Consequently,
\begin{align*}
 \|\Delta_{\lambda,\diamond}\|_{L^p}
 &\lesssim_d
 p^2E^{1/2}\mathfrak e_\lambda^{1/2}
 +p^2\lambda^2N_{\rm free}w(0)
 +p(|\alpha_\lambda|+|\alpha_0|)
 (\lambda^2N_{\rm free})^{1/2}\\
 &\quad+
 p|\alpha_\lambda-\alpha_0|+|\beta_\lambda-\beta_0|
 \lesssim_d p^2\mathfrak r_\lambda.
\end{align*}
\end{proof}

\begin{lemma}\label{lem:Lp-estimate-three-weights} Let $0<\lambda\leq1$, $q>1$, and $1/q+1/r=1$. There exists a constant $C_{d,c_0}>0$, depending only on $d,c_0$, such that
\begin{equation}\label{eq:normalized-weights}
 \max\Big\{
 \|f_\lambda\|_{L^q(\widehat{\mathbb P}_\lambda)},
 \|f_{\mathrm q}\|_{L^q(\widehat{\mathbb P}_\lambda)},
 \|f_{\mathrm c}\|_{L^q(\widehat{\mathbb P}_\lambda)}
 \Big\}
 \leq
 \exp\Big(\frac{C_{d,c_0}\mathfrak a_\lambda}{r}\Big).
\end{equation}
Here $\mathfrak a_\lambda$ and $c_0=\widehat w(0)/[2(2\pi)^d]$ are defined in\kern0pt{} \eqref{eq:comparison-a} and\kern0pt{} \eqref{eq:def-of-c0-wo}, respectively.
\end{lemma}

\begin{proof}
The zero Fourier mode gives
\begin{align}
 \cW_\lambda[u_\lambda]+\beta_\lambda
 &\geq
 c_0\big|\cM_\lambda^\1[u_\lambda]\big|^2
 -\alpha_\lambda\cM_\lambda^\1[u_\lambda]
 \geq-\frac{\alpha_\lambda^2}{4c_0},
 \nn\\
 \bW_\lambda^\ren+\beta_\lambda
 &\geq
 c_0|\mathbb M_\lambda^\1|^2
 -\alpha_\lambda\mathbb M_\lambda^\1
 \geq-\frac{\alpha_\lambda^2}{4c_0},
 \nn\\
 \cW_0[u_\lambda+v_\lambda]+\beta_0
 &\geq
 c_0\big|\cM_0^\1[u_\lambda+v_\lambda]\big|^2
 -\alpha_0\cM_0^\1[u_\lambda+v_\lambda]
 \geq-\frac{\alpha_0^2}{4c_0}.
 \label{eq:interaction-lower-bounds}
\end{align}
By \eqref{eq:Gaussian-mass-limit-Lp} and $c_\lambda(j)\leq h(j)^{-1}$, we have
\begin{align*}
 \mathbb E\cW_\lambda[u_\lambda]+\beta_\lambda
 &=
 \frac1{2(2\pi)^d}
 \sum_{k,j}\widehat w(k)c_\lambda(j)c_\lambda(j+k)\\
 &\leq
 \frac1{2(2\pi)^d}
 \sum_{k,j}\frac{\widehat w(k)}{h(j)h(j+k)}
 =
 \mathbb E\cW_0[u_\lambda+v_\lambda]+\beta_0
 =
 \frac{(2\pi)^d}{2}E.
\end{align*}
Moreover, \eqref{eq:abstract-coherent-kernel} gives
$$
 \mathbb E\big[
 |\mathcal R_{\lambda,\mathrm q}^{e_k}|^2
 \mid u_\lambda\big]
 =
 \lambda\langle u_\lambda,|e_k|^2u_\lambda\rangle
 =
 \frac{\lambda}{(2\pi)^d}\|u_\lambda\|^2.
$$
Using \eqref{eq:q-c-fluctuation-identity} and $\tr_{\gH}\mathsf C_\lambda=\lambda N_{\rm free}$, we obtain
\begin{align*}
 \mathbb E\bW_\lambda^\ren+\beta_\lambda
 &=
 \mathbb E\cW_\lambda[u_\lambda]+\beta_\lambda
 +\frac12\sum_k\widehat w(k)
 \mathbb E|\mathcal R_{\lambda,\mathrm q}^{e_k}|^2\\
 &=
 \mathbb E\cW_\lambda[u_\lambda]+\beta_\lambda
 +\frac12\lambda^2N_{\rm free}w(0)
 \leq
 \frac{(2\pi)^d}{2}E
 +\frac12\lambda^2N_{\rm free}w(0).
\end{align*}

Jensen's inequality and \eqref{eq:interaction-lower-bounds} give
\begin{align*}
 0<e^{-\mathbb E\cW_\lambda[u_\lambda]}
 &\leq z_\lambda
 \leq e^{\beta_\lambda+\alpha_\lambda^2/(4c_0)}<\infty,\\
 0<e^{-\mathbb E\bW_\lambda^\ren}
 &\leq z_{\mathrm q}
 \leq e^{\beta_\lambda+\alpha_\lambda^2/(4c_0)}<\infty,\\
 0<e^{-\mathbb E\cW_0[u_\lambda+v_\lambda]}
 &\leq z_{\mathrm c}
 \leq e^{\beta_0+\alpha_0^2/(4c_0)}<\infty.
\end{align*}
Since $\mathbb Ef_\lambda=\mathbb Ef_{\mathrm q} =\mathbb Ef_{\mathrm c}=1$ and $(q-1)/q=1/r$, we obtain
\begin{align*}
 &\max\big\{
 \|f_\lambda\|_{L^q},\|f_{\mathrm q}\|_{L^q},
 \|f_{\mathrm c}\|_{L^q}\big\}
 \leq
 \max\big\{
 \|f_\lambda\|_{L^\infty},\|f_{\mathrm q}\|_{L^\infty},
 \|f_{\mathrm c}\|_{L^\infty}\big\}^{1/r}\\
 &\quad\leq
 \max\left\{
 \frac{e^{\beta_\lambda+\alpha_\lambda^2/(4c_0)}}{z_\lambda},
 \frac{e^{\beta_\lambda+\alpha_\lambda^2/(4c_0)}}{z_{\mathrm q}},
 \frac{e^{\beta_0+\alpha_0^2/(4c_0)}}{z_{\mathrm c}}
 \right\}^{1/r}\\
 &\quad\leq
 \exp\left[
 \frac1r\left(
 \frac{(2\pi)^d}{2}E+\frac12\lambda^2N_{\rm free}w(0)
 +\frac{\max\{\alpha_\lambda^2,\alpha_0^2\}}{4c_0}
 \right)\right]
 \leq
 \exp\Big(\frac{C_{d,c_0}\mathfrak a_\lambda}{r}\Big),
\end{align*}
where the last inequality uses $\lambda^2N_{\rm free}w(0)\leq\mathfrak e_\lambda$ and \eqref{eq:comparison-a}.
\end{proof}

We now combine these estimates to prove Proposition~\ref{prop:estimate-E-lambda-diamond}.

\begin{proof}[Proof of Proposition~\ref{prop:estimate-E-lambda-diamond}] Choose $r=2+\mathfrak a_\lambda$ and $q=r/(r-1)$. By H\"older's inequality\kern0pt{},
$$
 \mathrm{LHS\ of\ \eqref{eq:free-energy-rate}}
 \lesssim
 \sum_{\diamond\in\{\mathrm q,\mathrm c\}}\big(\|f_\lambda\|_{L^q}+\|f_\diamond\|_{L^q}\big)
 \|\Delta_{\lambda,\diamond}\|_{L^r}
 \lesssim_{d,c_0}
 r^2\mathfrak r_\lambda
 \lesssim_{d,c_0}
 \mathfrak a_\lambda^2\mathfrak r_\lambda,
$$
where the norms are bounded using Lemmas~\ref{lem:Lp-estimate-Delta} and \ref{lem:Lp-estimate-three-weights}. Equation~\eqref{eq:entropy-rate} follows from \eqref{eq:symmetric-entropy-error}, \eqref{eq:partition-from-entropy}, and \eqref{eq:free-energy-from-partitions}.
\end{proof}

\bigskip

\section{Weighted correlation operators and a priori bounds}
\label{sec:correlation-preliminaries}

Recall the three Gibbs weights introduced in Section~\ref{sec:proof-strategy}:
$$
 \begin{gathered}
 f_\lambda=z_\lambda^{-1}e^{-\cW_\lambda[u_\lambda]},\qquad
 f_{\mathrm q}=z_{\mathrm q}^{-1}e^{-\bW_\lambda^{\ren}(\mathbf x)},
 \qquad
 f_{\mathrm c}=z_{\mathrm c}^{-1}e^{-\cW_0[u_\lambda+v_\lambda]}.
 \end{gathered}
$$
In this section, we prove the following a priori bounds.

\begin{proposition}[A priori bounds]\label{prop:moment-propagation} Assume Assumption~\ref{ass:interaction}. For every fixed $k,L\geq1$,
\begin{equation}
 \big\|\gamma_{f_\lambda}^{(k)}[u_\lambda+v_\lambda]\big\|_{\gS^2}
 +\big\|\gamma_{f_\lambda}^{(k)}[u_\lambda]\big\|_{\gS^2}
 +\big\|\gamma_{f_{\mathrm q}}^{(k)}[u_\lambda]\big\|_{\gS^2}
 \lesssim_{d,c_0,k,L}
 \Big[
 \big\|\gamma_{\nu_0}^{(k)}\big\|_{\gS^2}
 +\mathfrak a_\lambda^k
 (\mathfrak a_\lambda^2\mathfrak r_\lambda)^L
 \Big].
 \label{eq:propagated-moments}
\end{equation}
\end{proposition}

Before proving Proposition~\ref{prop:moment-propagation}, we give a precise definition of the weighted correlations introduced in \eqref{eq:weighted-correlation-common-space} and establish their operator realizations and approximation properties.

In the general statements of this section, we fix a probability space $(\gX,\mathcal A,\mathbb P)$ and write $\mathbb E=\int_{\gX}\dd\mathbb P$. All fields and weights in the same expression are defined on this space. We return to $\widehat{\mathbb P}_\lambda$ when applying the results to the Gibbs weights.

Fix $s>0$. We use the spaces $\gH^s,\gH^{-s}$ from Subsection~\ref{subsec:classical-field-theory}, with the duality pairing
$$
 \langle\phi,u\rangle
 =
 \langle h^{s/2}\phi,h^{-s/2}u\rangle_{\gH},
 \qquad \phi\in\gH^s,\quad u\in\gH^{-s}.
$$
We write $\langle u,\phi\rangle=\overline{\langle\phi,u\rangle}$. The physical $k$-particle space is $\bigotimes_{\mathrm{sym}}^k\gH$. We denote by $\Pi_k$ the orthogonal projection onto this space:
$$
 \Pi_k(\phi_1\otimes\cdots\otimes\phi_k)
 =
 \frac1{k!}\sum_{\sigma\in S_k}
 \phi_{\sigma(1)}\otimes\cdots\otimes\phi_{\sigma(k)}.
$$
We test against finite linear combinations of these symmetric tensors with $\phi_j\in\gH^s$. For such tensors,
$$
 \big\langle \Pi_k(\phi_1\otimes\cdots\otimes\phi_k),
 u^{\otimes k}\big\rangle
 =
 \prod_{j=1}^k\langle\phi_j,u\rangle.
$$
All tensor norms with an unspecified underlying space refer to the physical particle space. When forming tensor products of operators on symmetric particle spaces, we extend them by zero on the orthogonal complements.

\begin{definition}[Weighted correlations] \label{def:weighted-correlations} Let $u:\gX\to\gH^{-s}$ be measurable and let $f\in L^1(\mathbb P;\CC)$. For $k\geq1$, assume that
\begin{equation}\label{eq:correlation-integrability}
 \mathbb E\big[|f|\,|\langle\phi,u\rangle|^{2k}\big]<\infty
 \qquad\text{for every }\phi\in\gH^s.
\end{equation}
We define the correlation form by
\begin{equation}\label{eq:intrinsic-correlation-form}
 \langle\Phi,\gamma_f^{(k)}[u]\Psi\rangle
 =
 \mathbb E\big[
 f\,\langle\Phi,u^{\otimes k}\rangle
 \langle u^{\otimes k},\Psi\rangle
 \big]
\end{equation}
on the test tensors specified above. If this form satisfies
\begin{equation}\label{eq:correlation-bounded-form}
 \big|\langle\Phi,\gamma_f^{(k)}[u]\Psi\rangle\big|
 \leq C\|\Phi\|\|\Psi\|
\end{equation}
with a constant independent of $\Phi,\Psi$, we use the same notation for its unique bounded operator realization on $\bigotimes_{\mathrm{sym}}^k\gH$. At order zero, we use $\gamma_f^{(0)}[u]=\mathbb Ef$.
\end{definition}

Condition~\eqref{eq:correlation-integrability} ensures absolute integrability of all the matrix elements. Indeed, H\"older's inequality gives
$$
 \mathbb E\Big[|f|
 \prod_{j=1}^k|\langle\phi_j,u\rangle|^2\Big]
 \leq
 \prod_{j=1}^k
 \Big(\mathbb E[|f|\,|\langle\phi_j,u\rangle|^{2k}]\Big)^{1/k},
$$
and the Cauchy--Schwarz inequality controls the mixed matrix elements. The test tensors are dense in $\bigotimes_{\mathrm{sym}}^k\gH$; hence \eqref{eq:correlation-bounded-form} gives a continuous extension of the form, and the Riesz representation theorem yields the stated operator. The correlation form is complex linear in $f$ and depends only on the joint law of $f,u$. Whenever the bounded realizations exist,
$$
 \gamma_{\overline f}^{(k)}[u]
 =
 \big(\gamma_f^{(k)}[u]\big)^*,
 \qquad
 \gamma_f^{(k)}[u]\geq0\quad\text{if }f\geq0.
$$
For $f\geq0$ and $\mathbb Ef=1$, the weighted field law $\nu=u_\#(f\dd\mathbb P)$ satisfies
$$
 \langle\Phi,\gamma_f^{(k)}[u]\Psi\rangle
 =
 \int_{\gH^{-s}}
 \langle\Phi,u^{\otimes k}\rangle
 \langle u^{\otimes k},\Psi\rangle\dd\nu(u).
$$
We denote this form and its bounded realization, when it exists, by $\gamma_\nu^{(k)}$, with $\gamma_\nu^{(0)}=1$. Thus the notation $\int|u^{\otimes k}\rangle\langle u^{\otimes k}|\dd\nu(u)$ is always interpreted through these matrix elements. For an $\gH$-valued field satisfying $\mathbb E[|f|\|u\|^{2k}]<\infty$, it also agrees with the trace-class integral, since
$$
 \mathbb E\big\|
 f\,|u^{\otimes k}\rangle\langle u^{\otimes k}|
 \big\|_{\gS^1}
 =
 \mathbb E[|f|\|u\|^{2k}]<\infty.
$$

We use the Fourier projections $P_K=\1_{\{h\leq K\}}$, $K\geq1$, in the following approximation result.

\begin{lemma}[Realization and approximation] \label{lem:correlation-realization-approximation} Suppose that $u,f$ satisfy \eqref{eq:correlation-integrability}. For all test tensors $\Phi,\Psi$,
\begin{equation}\label{eq:correlation-form-approximation}
 \mathbb E\big[
 f\,\langle\Phi,(P_Ku)^{\otimes k}\rangle
 \langle(P_Ku)^{\otimes k},\Psi\rangle
 \big]
 \longrightarrow
 \mathbb E\big[
 f\,\langle\Phi,u^{\otimes k}\rangle
 \langle u^{\otimes k},\Psi\rangle
 \big].
\end{equation}
If $\gamma_f^{(k)}[u]$ has a bounded realization, then
\begin{equation}\label{eq:correlation-compression}
 \gamma_f^{(k)}[P_Ku]
 =
 P_K^{\otimes k}\gamma_f^{(k)}[u]P_K^{\otimes k}.
\end{equation}
If this realization belongs to $\gS^p$, $1\leq p<\infty$, then
\begin{equation}\label{eq:correlation-Schatten-approximation}
 \big\|\gamma_f^{(k)}[P_Ku]-\gamma_f^{(k)}[u]\big\|_{\gS^p}
 \longrightarrow0.
\end{equation}

Conversely, if $\sup_K\|\gamma_f^{(k)}[P_Ku]\|_{\gS^2}<\infty$, the form \eqref{eq:intrinsic-correlation-form} has a unique Hilbert--Schmidt realization, and
$$
 \big\|\gamma_f^{(k)}[u]\big\|_{\gS^2}^2
 =
 \sup_K\big\|\gamma_f^{(k)}[P_Ku]\big\|_{\gS^2}^2.
$$
For $L\geq K$, the finite-dimensional operators also satisfy
\begin{equation}\label{eq:HS-compatible-compressions}
 \big\|\gamma_f^{(k)}[P_Lu]
 -\gamma_f^{(k)}[P_Ku]\big\|_{\gS^2}^2
 =
 \big\|\gamma_f^{(k)}[P_Lu]\big\|_{\gS^2}^2
 -
 \big\|\gamma_f^{(k)}[P_Ku]\big\|_{\gS^2}^2.
\end{equation}
\end{lemma}

\begin{proof}
We first establish \eqref{eq:correlation-form-approximation}. By \eqref{eq:correlation-integrability}, the linear map $\phi\mapsto\langle u,\phi\rangle$ is defined on all of $\gH^s$ with values in $L^{2k}(|f|\dd\mathbb P)$. It has a closed graph: convergence of $\phi_j$ to $\phi$ in $\gH^s$ gives pointwise convergence of the pairings, while convergence in $L^{2k}(|f|\dd\mathbb P)$ gives the same limit almost everywhere along a subsequence. The closed graph theorem therefore yields $ \|\langle u,\phi\rangle\|_{L^{2k}(|f|\dd\mathbb P)} \leq C\|\phi\|_{\gH^s}. $ In particular, $\langle u,(P_K-1)\phi\rangle\to0$ in this space. For $\Phi=\Pi_k(\phi_1\otimes\cdots\otimes\phi_k)$, a telescoping expansion of the product and H\"older's inequality give
\begin{align*}
 &\big\|\langle\Phi,(P_Ku)^{\otimes k}\rangle
 -\langle\Phi,u^{\otimes k}\rangle
 \big\|_{L^2(|f|\dd\mathbb P)}\\
 &\leq
 \sum_{j=1}^k
 \|\langle(P_K-1)\phi_j,u\rangle\|_{L^{2k}(|f|\dd\mathbb P)}
 \prod_{i<j}
 \|\langle P_K\phi_i,u\rangle\|_{L^{2k}(|f|\dd\mathbb P)}
 \prod_{i>j}
 \|\langle\phi_i,u\rangle\|_{L^{2k}(|f|\dd\mathbb P)}
 \longrightarrow0.
\end{align*}
Linearity gives the same convergence for every test tensor. The Cauchy--Schwarz inequality with the weight $|f|$ now proves \eqref{eq:correlation-form-approximation}.

For fixed $K$,
$$
 \mathbb E[|f|\|P_Ku\|^{2k}]
 \leq
 \big(\dim P_K\gH\big)^{k-1}\sum_{h(n)\leq K}
 \mathbb E[|f|\,|\langle e_n,u\rangle|^{2k}]
 <\infty.
$$
Hence $\gamma_f^{(k)}[P_Ku]$ is a finite-rank operator defined by its integrable matrix elements. If $\gamma_f^{(k)}[u]$ is bounded, then
$$
 \langle\Phi,\gamma_f^{(k)}[P_Ku]\Psi\rangle
 =
 \mathbb E\big[
 f\,\langle P_K^{\otimes k}\Phi,u^{\otimes k}\rangle
 \langle u^{\otimes k},P_K^{\otimes k}\Psi\rangle
 \big]
 =
 \langle\Phi,
 P_K^{\otimes k}\gamma_f^{(k)}[u]P_K^{\otimes k}\Psi\rangle.
$$
Equality on the dense set of test tensors proves \eqref{eq:correlation-compression}.

Suppose that $\gamma_f^{(k)}[u]\in\gS^p$. For a finite-rank operator $A$ on $\bigotimes_{\mathrm{sym}}^k\gH$,
$$
 \big\|P_K^{\otimes k}\gamma_f^{(k)}[u]P_K^{\otimes k}
 -\gamma_f^{(k)}[u]\big\|_{\gS^p}
 \leq
 2\big\|\gamma_f^{(k)}[u]-A\big\|_{\gS^p}
 +
 \big\|P_K^{\otimes k}AP_K^{\otimes k}-A\big\|_{\gS^p}.
$$
Since $P_K^{\otimes k}\to1$ strongly, each rank-one summand satisfies
$$
 \big\|P_K^{\otimes k}|\Phi\rangle\langle\Psi|P_K^{\otimes k}
 -|\Phi\rangle\langle\Psi|\big\|_{\gS^p}
 \leq
 \|(P_K^{\otimes k}-1)\Phi\|\,\|\Psi\|
 +
 \|\Phi\|\,\|(P_K^{\otimes k}-1)\Psi\|
 \longrightarrow0.
$$
Finite-rank operators are dense in $\gS^p$, which proves \eqref{eq:correlation-Schatten-approximation}. The same argument applies to any finite-rank orthogonal projections with ranges in $\gH^s$ that converge strongly to the identity on $\gH$, and yields the operator associated with the same form \eqref{eq:intrinsic-correlation-form}.

Conversely, if $\sup_K\|\gamma_f^{(k)}[P_Ku]\|_{\gS^2}<\infty$, \eqref{eq:correlation-form-approximation} gives, for all test tensors $\Phi,\Psi$,
$$
 \big|\langle\Phi,\gamma_f^{(k)}[u]\Psi\rangle\big|
 =
 \lim_{K\to\infty}
 \big|\langle\Phi,\gamma_f^{(k)}[P_Ku]\Psi\rangle\big|
 \leq
 \sup_K\|\gamma_f^{(k)}[P_Ku]\|_{\gS^2}\,\|\Phi\|\|\Psi\|.
$$
Thus \eqref{eq:correlation-bounded-form} holds, and the form has a bounded operator realization. Let $(\Phi_j)_j$ be the orthonormal symmetric Fourier basis. By \eqref{eq:correlation-compression},
$$
 \big\|\gamma_f^{(k)}[u]\big\|_{\gS^2}^2
 =
 \sum_{i,j}
 \big|\langle\Phi_i,\gamma_f^{(k)}[u]\Phi_j\rangle\big|^2
 =
 \lim_{K\rightarrow\infty}
 \big\|P_K^{\otimes k}\gamma_f^{(k)}[u]P_K^{\otimes k}\big\|_{\gS^2}^2
 =
 \sup_K\|\gamma_f^{(k)}[P_Ku]\|_{\gS^2}^2<\infty.
$$
Thus the operator already defined is Hilbert--Schmidt, and \eqref{eq:correlation-Schatten-approximation} applies.

Finally, $P_K^{\otimes k}\gamma_f^{(k)}[P_Lu]P_K^{\otimes k} =\gamma_f^{(k)}[P_Ku]$ for $L\geq K$. Compression is an orthogonal projection in $\gS^2$, so
$$
 \tr\Big(
 \gamma_f^{(k)}[P_Lu]^*
 \gamma_f^{(k)}[P_Ku]\Big)
 =
 \tr\Big(
 [P_K^{\otimes k}\gamma_f^{(k)}[P_Lu]P_K^{\otimes k}]^*
 \gamma_f^{(k)}[P_Ku]\Big)
 =
 \|\gamma_f^{(k)}[P_Ku]\|_{\gS^2}^2.
$$
Expanding the square gives \eqref{eq:HS-compatible-compressions}.
\end{proof}

\begin{lemma}[Classical relative entropy and correlations] \label{lem:classical-entropy-HS} Let $f,f'$ be probability densities on $(\gX,\mathcal A,\mathbb P)$, and let $u:\gX\to\gH^{-s}$ be measurable. Suppose that $\gamma_f^{(2k)}[u]$ and $\gamma_{f'}^{(2k)}[u]$ have Hilbert--Schmidt realizations. Then
\begin{equation}
 \|\gamma_f^{(k)}[u]-\gamma_{f'}^{(k)}[u]\|_{\gS^2}^2
 \leq
 2\cH_{\cl}(f\dd\mathbb P,f'\dd\mathbb P)
 \|\gamma_f^{(2k)}[u]+\gamma_{f'}^{(2k)}[u]\|_{\gS^2}.
 \label{eq:classical-entropy-HS}
\end{equation}
\end{lemma}

\begin{proof}
The Cauchy--Schwarz inequality gives
\begin{align*}
 &\|\gamma_f^{(k)}[P_Ku]-\gamma_{f'}^{(k)}[P_Ku]\|_{\gS^2}^2\\
 &=
 \iint
 (f-f')(g)(f-f')(g')
 |\langle P_Ku(g),P_Ku(g')\rangle|^{2k}
 \dd\mathbb P(g)\dd\mathbb P(g')\\
 &\leq
 \iint
 |f-f'|(g)|f-f'|(g')
 |\langle P_Ku(g),P_Ku(g')\rangle|^{2k}
 \dd\mathbb P(g)\dd\mathbb P(g')\\
 &\leq
 \mathbb E\frac{|f-f'|^2}{f+f'}\,
 \Big(\iint(f+f')(g)(f+f')(g')
 |\langle P_Ku(g),P_Ku(g')\rangle|^{4k}
 \dd\mathbb P(g)\dd\mathbb P(g')\Big)^{1/2}\\
 &=
 \mathbb E\frac{|f-f'|^2}{f+f'}\,
 \|\gamma_f^{(2k)}[P_Ku]+\gamma_{f'}^{(2k)}[P_Ku]\|_{\gS^2}\\
 &\leq
 2\cH_{\cl}(f\dd\mathbb P,f'\dd\mathbb P)
 \|\gamma_f^{(2k)}[P_Ku]+\gamma_{f'}^{(2k)}[P_Ku]\|_{\gS^2}.
\end{align*}
In the last line, we used the scalar inequality
$$
 \frac{(a-b)^2}{a+b}\leq2(a\log(a/b)-a+b),\qquad a,b>0.
$$
Lemma~\ref{lem:correlation-realization-approximation} allows us to pass to the limit $K\to\infty$, which proves \eqref{eq:classical-entropy-HS}.
\end{proof}

\medskip

\subsection{Gaussian fields}

For a Gaussian field, we estimate the correlation form directly in terms of its covariance. Throughout this subsection, let $u:\gX\to\gH^{-s}$ be a Gaussian field with bounded covariance $\mathsf C$, in the sense of Subsection~\ref{subsec:classical-field-theory}. We use the identity
\begin{equation}\label{eq:gamma-1-k-identity}
 \gamma_1^{(k)}[u]=\mathbb E\big[\ket{u^{\otimes k}}\bra{u^{\otimes k}}\big]=k!\mathsf C^{\otimes k}.
\end{equation}
See \cite[Section 3.2]{LNR15}.

\begin{lemma}\label{lem:Gaussian-weighted-correlations} For $k\geq1$ and $f\in L^q(\mathbb P;\CC)$, $1<q\leq\infty$, the form \eqref{eq:intrinsic-correlation-form} has a unique bounded realization on $\bigotimes_{\mathrm{sym}}^k\gH$. If $\mathsf C\in\gS^p$, $1\leq p<\infty$, then
\begin{equation}\label{eq:Gaussian-weighted-HS}
 \big\|\gamma_f^{(k)}[u]\big\|_{\gS^p}
 \leq
 (2kq')^k\|f\|_{L^q(\mathbb P)}\|\mathsf C\|_{\gS^p}^k,
 \qquad \frac1q+\frac1{q'}=1.
\end{equation}
\end{lemma}

\begin{proof}
For a test tensor $\Phi$ and an integer $m\geq1$, applying \eqref{eq:gamma-1-k-identity} at order $km$ gives
\begin{align*}
 \mathbb E|\langle\Phi,u^{\otimes k}\rangle|^{2m}
 &=
 \mathbb E\big|
 \langle \Pi_{km}(\Phi^{\otimes m}),u^{\otimes km}\rangle
 \big|^2
 =
 \big\langle \Pi_{km}(\Phi^{\otimes m}),
 \gamma_1^{(km)}[u]\Pi_{km}(\Phi^{\otimes m})\big\rangle\\
 &=
 (km)!\,\big\|\Pi_{km}[((\mathsf C^{1/2})^{\otimes k}\Phi)^{\otimes m}]\big\|^2
 \leq
 (km)!\,\|(\mathsf C^{1/2})^{\otimes k}\Phi\|^{2m}.
\end{align*}
Here we used the fact that $(\mathsf C^{1/2})^{\otimes km}$ commutes with $\Pi_{km}$ and that $\Pi_{km}$ is an orthogonal projection.

Choose $m=\lceil q'\rceil$, so that $q'\leq m\leq2q'$. H\"older's inequality gives
\begin{align*}
 \mathbb E\big[|f|\,|\langle\Phi,u^{\otimes k}\rangle|^2\big]
 &\leq
 \|f\|_{L^q(\mathbb P)}
 \|\langle\Phi,u^{\otimes k}\rangle\|_{L^{2q'}(\mathbb P)}^2
 \leq
 \|f\|_{L^q(\mathbb P)}
 \big(\mathbb E|\langle\Phi,u^{\otimes k}\rangle|^{2m}\big)^{1/m}\\
 &\leq
 ((km)!)^{1/m}\|f\|_{L^q(\mathbb P)}\|(\mathsf C^{1/2})^{\otimes k}\Phi\|^2
 \leq
 (2kq')^k\|f\|_{L^q(\mathbb P)}\|(\mathsf C^{1/2})^{\otimes k}\Phi\|^2.
\end{align*}
Taking $\Phi=\phi^{\otimes k}$ proves \eqref{eq:correlation-integrability}, so the correlation form $\gamma_f^{(k)}[u]$ is well defined on the test tensor space. For two test tensors, the Cauchy--Schwarz inequality gives
$$
 \big|\langle\Phi,\gamma_f^{(k)}[u]\Psi\rangle\big|
 \leq
 (2kq')^k\|f\|_{L^q(\mathbb P)}\|(\mathsf C^{1/2})^{\otimes k}\Phi\|\|(\mathsf C^{1/2})^{\otimes k}\Psi\|
 \leq
 (2kq')^k\|f\|_{L^q(\mathbb P)}
 \|\mathsf C\|^k\|\Phi\|\|\Psi\|.
$$
Thus \eqref{eq:correlation-bounded-form} holds, and Definition~\ref{def:weighted-correlations} gives the unique bounded realization.

The preceding estimate also shows that $(\mathsf C^{1/2})^{\otimes k}\Phi=(\mathsf C^{1/2})^{\otimes k}\Phi'$ implies $\langle\Phi-\Phi',\gamma_f^{(k)}[u]\Psi\rangle=0$, with the same conclusion in the second variable. The Riesz representation theorem gives an operator $T_f$ vanishing on $(\Ran (\mathsf C^{1/2})^{\otimes k})^\perp$, with $\|T_f\|\leq(2kq')^k\|f\|_{L^q(\mathbb P)}$, such that $\gamma_f^{(k)}[u]=(\mathsf C^{1/2})^{\otimes k}T_f(\mathsf C^{1/2})^{\otimes k}$. H\"older's inequality for Schatten norms gives
$$
 \|\gamma_f^{(k)}[u]\|_{\gS^p}
 =
 \|(\mathsf C^{1/2})^{\otimes k}T_f(\mathsf C^{1/2})^{\otimes k}\|_{\gS^p}
 \leq
 \|T_f\|\,\|(\mathsf C^{1/2})^{\otimes k}\|_{\gS^{2p}}^2
 \leq
 (2kq')^k\|f\|_{L^q(\mathbb P)}
 \|\mathsf C\|_{\gS^p}^k.
$$
\end{proof}

\begin{lemma}[Adding a Gaussian field]\label{lem:weighted-Gaussian-addition} Let $u,v:\gX\to\gH^{-s}$ be independent Gaussian fields with covariances $\mathsf C,\mathsf D\in\gS^2$. If $f\in L^q(\mathbb P;\CC)$, $1<q\leq\infty$, is measurable with respect to $u$, then
\begin{equation}\label{eq:weighted-Gaussian-addition}
 \gamma_f^{(k)}[u+v]
 =
 \sum_{j=0}^k\binom kj^2j!\,
 \Pi_k\big[\gamma_f^{(k-j)}[u]\otimes\mathsf D^{\otimes j}\big]\Pi_k.
\end{equation}
The identity holds in $\gS^2$. For $f\geq0$, every summand is positive, and $0\leq\gamma_f^{(k)}[u]\leq\gamma_f^{(k)}[u+v]$.
\end{lemma}

\begin{proof}
Lemma~\ref{lem:Gaussian-weighted-correlations} gives Hilbert--Schmidt realizations for $\gamma_f^{(k)}[u+v]$ and $\gamma_f^{(k-j)}[u]$, since $u+v$ has covariance $\mathsf C+\mathsf D$. For $\phi,\psi\in\gH^s$, independence and $\Law(e^{\ii\theta}v)=\Law(v)$ yield
$$
 \mathbb E\big[
 \langle\phi,v\rangle^j\langle v,\psi\rangle^\ell
 \mid u\big]
 =
 \mathbb E\big[
 \langle\phi,v\rangle^j\langle v,\psi\rangle^\ell\big]
 =
 \delta_{j\ell}\,j!\langle\phi,\mathsf D\psi\rangle^j,
 \qquad 0\leq j,\ell\leq k,
$$
where the last equality uses \eqref{eq:gamma-1-k-identity} at order $j$.

Expanding the powers and using the measurability of $f$ with respect to $u$, we obtain
\begin{align*}
 \langle\phi^{\otimes k},
 \gamma_f^{(k)}[u+v]\psi^{\otimes k}\rangle
 &=
 \mathbb E\Big[
 f\big(\langle\phi,u\rangle+\langle\phi,v\rangle\big)^k
 \big(\langle u,\psi\rangle+\langle v,\psi\rangle\big)^k
 \Big]\\
 &\quad=
 \sum_{j,\ell=0}^k\binom kj\binom k\ell
 \mathbb E\Big[
 f\,\langle\phi,u\rangle^{k-j}\langle u,\psi\rangle^{k-\ell}
 \mathbb E\big[
 \langle\phi,v\rangle^j\langle v,\psi\rangle^\ell
 \mid u\big]\Big]\\
 &\quad=
 \sum_{j=0}^k\binom kj^2j!\,
 \mathbb E\big[
 f\,\langle\phi,u\rangle^{k-j}\langle u,\psi\rangle^{k-j}
 \big]\langle\phi,\mathsf D\psi\rangle^j\\
 &\quad=
 \sum_{j=0}^k\binom kj^2j!\,
 \big\langle\phi^{\otimes k},
 \Pi_k[\gamma_f^{(k-j)}[u]\otimes\mathsf D^{\otimes j}]
 \Pi_k\psi^{\otimes k}\big\rangle.
\end{align*}
By polarization, the tensors $\phi^{\otimes k}$ with $\phi\in\gH^s$ span the test tensor space. Its density in $\bigotimes_{\mathrm{sym}}^k\gH$ and the boundedness of both sides give \eqref{eq:weighted-Gaussian-addition} in $\gS^2$.
\end{proof}

\medskip

\subsection{A priori bounds}

We return to the fields and weights on $(\gX\times\gX\times\mathfrak C,\widehat{\mathbb P}_\lambda)$ from Section~\ref{sec:proof-strategy}, with $\mathbb E=\int\dd\widehat{\mathbb P}_\lambda$ and $L^p=L^p(\widehat{\mathbb P}_\lambda)$. The Gaussian laws of $u_\lambda,u_\lambda+v_\lambda$, together with \eqref{eq:normalized-weights} and Lemma~\ref{lem:Gaussian-weighted-correlations}, yield the Hilbert--Schmidt realizations used below. We retain the convention $\gamma_{\nu_0}^{(0)}=1$.

\begin{proof}[Proof of Proposition~\ref{prop:moment-propagation}] Using $e^x\leq\mathrm e+e^x|x|^L$, $x\in\R$, we obtain
\begin{equation}
 f_\lambda
 =
 \frac{z_{\mathrm c}}{z_\lambda}
 e^{\Delta_{\lambda,\mathrm c}}f_{\mathrm c}
 \leq
 \mathrm e\,\frac{z_{\mathrm c}}{z_\lambda}f_{\mathrm c}
 +f_\lambda|\Delta_{\lambda,\mathrm c}|^L,
 \qquad
 f_{\mathrm q}
 =
 \frac{z_\lambda}{z_{\mathrm q}}
 e^{-\Delta_{\lambda,\mathrm q}}f_\lambda
 \leq
 \mathrm e\,\frac{z_\lambda}{z_{\mathrm q}}f_\lambda
 +f_{\mathrm q}|\Delta_{\lambda,\mathrm q}|^L.
 \label{eq:weight-domination}
\end{equation}
Choose $r=2+\mathfrak a_\lambda$ and $q=r/(r-1)$. By H\"older's inequality\kern0pt{}, \eqref{eq:normalized-weights}, and \eqref{eq:abstract-interaction-difference-moment}, we have
\begin{align*}
 &\big\|f_\lambda|\Delta_{\lambda,\mathrm c}|^L
 \big\|_{L^{\frac{2r}{2r-1}}}
 +\big\|f_{\mathrm q}|\Delta_{\lambda,\mathrm q}|^L
 \big\|_{L^{\frac{2r}{2r-1}}}
 \leq
 \|f_\lambda\|_{L^q}
 \|\Delta_{\lambda,\mathrm c}\|_{L^{2Lr}}^L
 +\|f_{\mathrm q}\|_{L^q}
 \|\Delta_{\lambda,\mathrm q}\|_{L^{2Lr}}^L\\
 &\quad\leq
 2\exp\Big(\frac{C_{d,c_0}\mathfrak a_\lambda}{r}\Big)
 \big(C_d(2Lr)^2\mathfrak r_\lambda\big)^L\lesssim_{d,c_0,L}
 (\mathfrak a_\lambda^2\mathfrak r_\lambda)^L,
\end{align*}
where $(2r-1)/(2r)=1/q+1/(2r)$, $\mathfrak a_\lambda/r\leq1$, and $r\leq3\mathfrak a_\lambda$. Moreover, \eqref{eq:entropy-rate} and the smallness condition $\mathfrak a_\lambda^2\mathfrak r_\lambda\ll1$ give
$$
 \max\Big\{
 \frac{z_{\mathrm c}}{z_\lambda},
 \frac{z_\lambda}{z_{\mathrm q}}\Big\}
 \leq
 \exp\big(C_{d,c_0}\mathfrak a_\lambda^2\mathfrak r_\lambda\big)
 \lesssim_{d,c_0}1.
$$

By\kern0pt{} \eqref{eq:weight-domination} and $ \gamma_{\nu_0}^{(k)} = \gamma_{f_{\mathrm c}}^{(k)}[u_\lambda+v_\lambda], $ we obtain
$$
 0\leq\gamma_{f_\lambda}^{(k)}[u_\lambda+v_\lambda]
 \leq
 \mathrm e\,\frac{z_{\mathrm c}}{z_\lambda}
 \gamma_{\nu_0}^{(k)}
 +\gamma_{f_\lambda|\Delta_{\lambda,\mathrm c}|^L}^{(k)}
 [u_\lambda+v_\lambda].
$$
Taking Hilbert--Schmidt norms and applying Lemma~\ref{lem:Gaussian-weighted-correlations} with exponent $2r/(2r-1)$ and conjugate exponent $2r$, we obtain
$$
 \big\|\gamma_{f_\lambda}^{(k)}[u_\lambda+v_\lambda]\big\|_{\gS^2}
 \leq
 \mathrm e\,\frac{z_{\mathrm c}}{z_\lambda}
 \big\|\gamma_{\nu_0}^{(k)}\big\|_{\gS^2}
 +
 (4kr)^k\|\mathsf C_0\|_{\gS^2}^k
 \big\|f_\lambda|\Delta_{\lambda,\mathrm c}|^L
 \big\|_{L^{\frac{2r}{2r-1}}}
 \lesssim_{d,c_0,k,L}
 \big\|\gamma_{\nu_0}^{(k)}\big\|_{\gS^2}
 +\mathfrak a_\lambda^k
 (\mathfrak a_\lambda^2\mathfrak r_\lambda)^L.
$$
Here $\|\mathsf C_0\|_{\gS^2}^2 =\sum_{j\in\Z^d}h(j)^{-2}\lesssim_d1$.

Since $f_\lambda$ is nonnegative and measurable with respect to $u_\lambda$, Lemma~\ref{lem:weighted-Gaussian-addition} gives $ 0\leq\gamma_{f_\lambda}^{(k)}[u_\lambda] \leq\gamma_{f_\lambda}^{(k)}[u_\lambda+v_\lambda]. $ Thus $\|\gamma_{f_\lambda}^{(k)}[u_\lambda]\|_{\gS^2} \leq \|\gamma_{f_\lambda}^{(k)}[u_\lambda+v_\lambda]\|_{\gS^2}$.

Finally, the second inequality in \eqref{eq:weight-domination} gives
$$
 0\leq\gamma_{f_{\mathrm q}}^{(k)}[u_\lambda]
 \leq
 \mathrm e\,\frac{z_\lambda}{z_{\mathrm q}}
 \gamma_{f_\lambda}^{(k)}[u_\lambda]
 +\gamma_{f_{\mathrm q}|\Delta_{\lambda,\mathrm q}|^L}^{(k)}
 [u_\lambda].
$$
Applying \eqref{eq:Gaussian-weighted-HS} with the same exponent, we obtain
$$
 \big\|\gamma_{f_{\mathrm q}}^{(k)}[u_\lambda]\big\|_{\gS^2}
 \leq
 \mathrm e\,\frac{z_\lambda}{z_{\mathrm q}}
 \big\|\gamma_{f_\lambda}^{(k)}[u_\lambda]\big\|_{\gS^2}
 +
 (4kr)^k\|\mathsf C_\lambda\|_{\gS^2}^k
 \big\|f_{\mathrm q}|\Delta_{\lambda,\mathrm q}|^L
 \big\|_{L^{\frac{2r}{2r-1}}}
 \lesssim_{d,c_0,k,L}
 \big\|\gamma_{\nu_0}^{(k)}\big\|_{\gS^2}
 +\mathfrak a_\lambda^k
 (\mathfrak a_\lambda^2\mathfrak r_\lambda)^L,
$$
where $\|\mathsf C_\lambda\|_{\gS^2}\leq\|\mathsf C_0\|_{\gS^2}$. Combining these estimates proves \eqref{eq:propagated-moments}.
\end{proof}

\bigskip

\section{Kernel domination and convergence of reduced density matrices}\label{sec:kernel-domination}

Proposition~\ref{prop:moment-propagation} provides a priori bounds for the weighted correlation operators. The following kernel comparison allows us to transfer these bounds to the quantum reduced density matrices.

\begin{proposition}\label{prop:quantum-kernel-bound} Suppose that $w\geq0$. For every $\lambda>0$ and $k\geq1$, the following kernel bound holds:
$$
 0\leq
 \frac{k!\lambda^k}{z_{\mathrm q}}
 (e^{-\bW_\lambda^\ren}\Gamma_{\rm free})^{(k)}(X;Y)
 \leq
 e^{kb_\lambda}\gamma_{f_{\mathrm q}}^{(k)}[u_\lambda](X;Y),
$$
where $b_\lambda$ is defined in \eqref{eq:comparison-b}. Moreover,
$$
 k!\lambda^k\big\|\Gamma_\lambda^{(k)}\big\|_{\gS^2}
 \leq
 e^{kb_\lambda}
 \frac{Z_{\rm free}z_{\mathrm q}}{Z_\lambda}
 \big\|\gamma_{f_{\mathrm q}}^{(k)}[u_\lambda]\big\|_{\gS^2}.
$$
\end{proposition}

Before proving Proposition~\ref{prop:quantum-kernel-bound}, we establish a trace inequality with a permutation inserted. Such traces arise when the squared Hilbert--Schmidt norm of a reduced density matrix is expressed.

\begin{lemma}\label{lem:permutation-Jensen} On $L^2((\bT^d)^n)$, let $H_{0,n}=\lambda\sum_{j=1}^nh_j$, let $V$ be multiplication by a real-valued bounded continuous function, and let $U_\pi$ permute the $n$ coordinates. Then
$$
 \tr(U_\pi e^{-H_{0,n}-V})
 \leq
 \tr(U_\pi e^{-V}e^{-H_{0,n}}).
$$
\end{lemma}

\begin{proof}
Write $K_t(X,Y)=e^{-tH_{0,n}}(X,Y)$. The unnormalized free Brownian bridge measure $\mathbb B_{X,Y}$ on $[0,1]$ has total mass $K_1(X,Y)$. The Feynman--Kac formula for bounded potentials gives
$$
 e^{-H_{0,n}-V}(X,Y)
 =
 \int e^{-\int_0^1V(\omega(t))\dd t}
 \dd\mathbb B_{X,Y}(\omega).
$$

Using Jensen's inequality and the time-$t$ marginal of the bridge, we obtain
\begin{align*}
 \tr(U_\pi e^{-H_{0,n}-V})
 &=
 \int\dd X\int
 e^{-\int_0^1V(\omega(t))\dd t}
 \dd\mathbb B_{X,\pi X}(\omega)\\
 &\leq
 \int_0^1\dd t\int\dd X\int
 e^{-V(\omega(t))}
 \dd\mathbb B_{X,\pi X}(\omega)\\
 &=
 \int_0^1\dd t\int e^{-V(Y)}
 \int K_t(X,Y)K_{1-t}(Y,\pi X)\dd X\,\dd Y\\
 &=
 \int e^{-V(Y)}K_1(Y,\pi Y)\dd Y
 =
 \tr(U_\pi e^{-V}e^{-H_{0,n}}).
\end{align*}
Indeed, permutation invariance and the semigroup identity give
$$
 \int K_t(X,Y)K_{1-t}(Y,\pi X)\dd X
 =
 \int K_t(Z,\pi Y)K_{1-t}(Y,Z)\dd Z
 =
 K_1(Y,\pi Y).
$$
Here we use the convention $(U_\pi F)(X)=F(\pi^{-1}X)$.
\end{proof}

\begin{proof}[Proof of Proposition~\ref{prop:quantum-kernel-bound}] For $X=(x_1,\ldots,x_k)$, expansion of \eqref{eq:comparison-quantum-interaction} gives
\begin{align}
 &\bW_\lambda^\ren(\mathbf x\sqcup X)
 -\bW_\lambda^\ren(\mathbf x)
 \nn\\
 &=
 \lambda^2\sum_{a=1}^k\sum_{j=1}^{|\mathbf x|}w(x_a-X_j)
 +\frac{\lambda^2}{2}\sum_{a,b=1}^kw(x_a-x_b)
 -k\lambda\big(\widehat w(0)\lambda\rho_0+\alpha_\lambda\big)
 \geq-kb_\lambda.
 \label{eq:particle-insertion-bound}
\end{align}
Here we used $\sum_n\widehat w(n)e_n(x)\overline{e_n(y)}=w(x-y)$ and $w\geq0$.

Write $\Pi_n$ for the orthogonal projection onto $\bigotimes_{\mathrm{sym}}^n\gH$, and let $U_{n,m}^{(k)}$ exchange the first $k$ coordinates of the $n$- and $m$-particle factors. On the full $n$-particle space, write
$$
 H_{\lambda,n}
 =
 \lambda\sum_{j=1}^nh_j+\bW_{\lambda,n}^\ren,
$$
where $\bW_{\lambda,n}^\ren$ is the restriction of $\bW_\lambda^\ren$ to that sector. For the operators used below, the partial trace satisfies
\begin{align*}
 \tr\big[U_{n,m}^{(k)}(A_n\otimes B_m)\big]
 &=
 \int A_n(X,Z;Y,Z)B_m(Y,T;X,T)
 \dd X\,\dd Y\,\dd Z\,\dd T
 \\
 &=
 \tr\big[
 (\tr_{k+1,\ldots,n}A_n)
 (\tr_{k+1,\ldots,m}B_m)
 \big].
\end{align*}
In the expansion of $\Pi_n\otimes \Pi_m$ as an average over permutations, each $U_{n,m}^{(k)}U_{\sigma\oplus\tau}$ permutes the $n+m$ coordinates. We apply Lemma~\ref{lem:permutation-Jensen} to $\bW_{\lambda,n}^\ren+\bW_{\lambda,m}^\ren$. Together with \eqref{eq:rdm-convention}, this yields
\begin{align}
 \|\Gamma_\lambda^{(k)}\|_{\gS^2}^2
 &=
 Z_\lambda^{-2}\sum_{n,m\geq k}\binom nk\binom mk
 \tr\big[
 U_{n,m}^{(k)}(\Pi_n\otimes \Pi_m)
 e^{-H_{\lambda,n}\otimes\1-\1\otimes H_{\lambda,m}}
 \big]
 \nn\\
 &\leq
 Z_\lambda^{-2}\sum_{n,m\geq k}\binom nk\binom mk
 \tr\Big[
 U_{n,m}^{(k)}(\Pi_n\otimes \Pi_m)
 (e^{-\bW_{\lambda,n}^\ren}\otimes
 e^{-\bW_{\lambda,m}^\ren})
 e^{-\lambda\sum_{j=1}^{n+m}h_j}
 \Big]
 \nn\\
 &=
 \Big(\frac{Z_{\rm free}}{Z_\lambda}\Big)^2
 \iint
 (e^{-\bW_\lambda^\ren}\Gamma_{\rm free})^{(k)}(X;Y)
 (e^{-\bW_\lambda^\ren}\Gamma_{\rm free})^{(k)}(Y;X)
 \dd X\,\dd Y.
 \label{eq:two-copy-kernel-comparison}
\end{align}

By \eqref{eq:free-coherent-representation} and \eqref{eq:abstract-coherent-kernel},
\begin{equation}
 k!\lambda^k
 (e^{-\bW_\lambda^\ren}\Gamma_{\rm free})^{(k)}(X;Y)
 =
 \int u^{\otimes k}(X)\overline{u^{\otimes k}(Y)}
 \int_{\mathfrak C}
 e^{-\bW_\lambda^\ren(\mathbf x\sqcup X)}
 \dd\rho_{\lambda,u}(\mathbf x)\dd\mu_\lambda(u).
 \label{eq:coherent-particle-insertion}
\end{equation}
At a fixed configuration, integration with respect to the Gaussian field gives
$$
 \int e^{-\|u\|^2/\lambda}\dd\mu_\lambda(u)
 =
 \prod_n(1+c_\lambda(n)/\lambda)^{-1}
 =Z_{\rm free}^{-1},
 \qquad
 (\mathsf C_\lambda^{-1}+\lambda^{-1})^{-1}
 =\lambda e^{-\lambda h}.
$$
Hence, 
\begin{equation}
 \int e^{-\|u\|^2/\lambda}
 \prod_{a=1}^{n+k}u(s_a)
 \prod_{b=1}^{n+k}\overline{u(t_b)}\dd\mu_\lambda(u)
 =
 \frac{\lambda^{n+k}}{Z_{\rm free}}
 \sum_{\sigma\in\mathfrak S_{n+k}}
 \prod_{a=1}^{n+k}e^{-\lambda h}(s_a,t_{\sigma(a)})
 \geq0.
 \label{eq:Gaussian-positive-kernel}
\end{equation}
The Gaussian field has a smooth version for $\lambda>0$: $\tr(h^a\mathsf C_\lambda)<\infty$ for every $a\geq0$. Thus the pointwise values and Gaussian moments in these formulas are well defined.

Applying \eqref{eq:particle-insertion-bound} after \eqref{eq:Gaussian-positive-kernel}, we obtain
\begin{equation}\label{eq:inserted-kernel-domination}
 0\leq
 \frac{k!\lambda^k}{z_{\mathrm q}}
 (e^{-\bW_\lambda^\ren}\Gamma_{\rm free})^{(k)}(X;Y)
 \leq
 e^{kb_\lambda}\gamma_{f_{\mathrm q}}^{(k)}[u_\lambda](X;Y).
\end{equation}
The same Gaussian formula shows that the last kernel is real-valued and nonnegative. Its rank-one integral representation gives its Hermitian symmetry. Substituting\kern0pt{} \eqref{eq:inserted-kernel-domination} into \eqref{eq:two-copy-kernel-comparison} proves the result.

The lower bound in \eqref{eq:interaction-lower-bounds} gives, at fixed parameters,
$$
 0\leq(\Pi_ne^{-H_{\lambda,n}})(X;Y)
 \leq
 e^{\alpha_\lambda^2/(4c_0)+\beta_\lambda}
 (\Pi_ne^{-\lambda\sum_jh_j})(X;Y).
$$
Moreover,
$$
 \tr_{\gF}e^{-\lambda\dG(h)+(\lambda/2)\cN}
 =
 \prod_n\big(1-e^{-\lambda(h(n)-1/2)}\big)^{-1}<\infty.
$$
Thus all particle-number-weighted traces in \eqref{eq:two-copy-kernel-comparison} are finite. For \eqref{eq:coherent-particle-insertion}, the sum of the absolute values of the integrands is bounded by
$$
 e^{\alpha_\lambda^2/(4c_0)+\beta_\lambda}
 \int |u^{\otimes k}(X)|\,|u^{\otimes k}(Y)|\dd\mu_\lambda(u)<\infty.
$$
These bounds justify the partial traces, sums, and changes in the order of integration for each fixed choice of parameters.
\end{proof}

We now complete the proof of Theorem~\ref{thm:abstract-comparison}.

\begin{proof}[Proof of Theorem~\ref{thm:abstract-comparison}] The relative free-energy estimate \eqref{eq:main-FE} was proved in Proposition~\ref{prop:estimate-E-lambda-diamond}.

Suppose now that $w\geq0$ and $\mathfrak a_\lambda^2\mathfrak r_\lambda\ll1$. We prove \eqref{eq:main-HS}. For $0<\lambda\leq1$, using $ k!\lambda^k\widetilde\Gamma_\lambda^{(k)} = \gamma_{f_\lambda}^{(k)}[u_\lambda] $ and $ \gamma_{\nu_0}^{(k)} = \gamma_{f_{\mathrm c}}^{(k)}[u_\lambda+v_\lambda], $ we have
\begin{align}
 k!\lambda^k\Gamma_\lambda^{(k)}-\gamma_{\nu_0}^{(k)}
 &=
 k!\lambda^k
 \big(\Gamma_\lambda^{(k)}-\widetilde\Gamma_\lambda^{(k)}\big)\nn\\
 &+
 \big(\gamma_{f_\lambda}^{(k)}[u_\lambda+v_\lambda]
 -
 \gamma_{f_{\mathrm c}}^{(k)}[u_\lambda+v_\lambda]\big)
 +
 \big(\gamma_{f_\lambda}^{(k)}[u_\lambda]
 -
 \gamma_{f_\lambda}^{(k)}[u_\lambda+v_\lambda]\big).
 \label{eq:proof-HS-three-comparisons}
\end{align}
\noindent\textbf{Step 1: Quantum reduced density matrices.} We estimate the first term in \eqref{eq:proof-HS-three-comparisons}. By Proposition~\ref{prop:quantum-kernel-bound} and Proposition~\ref{prop:moment-propagation},
\begin{equation}
 k!\lambda^k\|\Gamma_\lambda^{(k)}\|_{\gS^2}
 \leq
 e^{kb_\lambda}
 \frac{Z_{\rm free}z_{\mathrm q}}{Z_\lambda}
 \|\gamma_{f_{\mathrm q}}^{(k)}[u_\lambda]\|_{\gS^2}\lesssim_{d,c_0,k,L}
 e^{kb_\lambda}
 \Big[
 \|\gamma_{\nu_0}^{(k)}\|_{\gS^2}
 +\mathfrak a_\lambda^k
 (\mathfrak a_\lambda^2\mathfrak r_\lambda)^L
 \Big],\label{eq:bound-Gamma-gamma}
\end{equation}
which proves \eqref{eq:main-moment}.

The phase invariance of $\nu_\lambda$ and the identity $e^{\ii\theta\cN}\xi(u/\sqrt\lambda) =\xi(e^{\ii\theta}u/\sqrt\lambda)$ give $[\widetilde\Gamma_\lambda,\cN]=0$; we also have $[\Gamma_\lambda,\cN]=0$. We apply \cite[Lemma~11.4, (11.14)]{LNR21}, together with Pinsker's inequality $\|\Gamma_\lambda-\widetilde\Gamma_\lambda\|_{\gS^1} \leq\big[2\cH(\widetilde\Gamma_\lambda,\Gamma_\lambda)\big]^{1/2}$, to obtain
$$
 \big\|k!\lambda^k
 (\Gamma_\lambda^{(k)}-\widetilde\Gamma_\lambda^{(k)})
 \big\|_{\gS^2}^2
 \lesssim_k
 \cH(\widetilde\Gamma_\lambda,\Gamma_\lambda)^{1/2}
 \sum_{\ell=k}^{2k}\lambda^{2k-\ell}
 \Big[
 \lambda^\ell\|\Gamma_\lambda^{(\ell)}\|_{\gS^2}
 +\frac1{\ell!}\|\gamma_{f_\lambda}^{(\ell)}[u_\lambda]\|_{\gS^2}
 \Big].
$$
Propositions~\ref{prop:estimate-E-lambda-diamond} and~\ref{prop:quantum-entropy-interface} give $\cH(\widetilde\Gamma_\lambda,\Gamma_\lambda)^{1/2}\lesssim_{d,c_0}\mathfrak a_\lambda\mathfrak r_\lambda^{1/2}$. Combining this bound with \eqref{eq:bound-Gamma-gamma} and Proposition~\ref{prop:moment-propagation}, we obtain
\begin{align}
 \big\|k!\lambda^k
 (\Gamma_\lambda^{(k)}-\widetilde\Gamma_\lambda^{(k)})
 \big\|_{\gS^2}^2
 &\lesssim_{d,c_0,k,L}
 \mathfrak a_\lambda\mathfrak r_\lambda^{1/2}e^{2kb_\lambda}
 \sum_{\ell=k}^{2k}\lambda^{2k-\ell}
 \Big[
 \|\gamma_{\nu_0}^{(\ell)}\|_{\gS^2}
 +\mathfrak a_\lambda^\ell
 (\mathfrak a_\lambda^2\mathfrak r_\lambda)^L
 \Big]\nn\\
 &\lesssim_{d,c_0,k,L}
 \mathfrak a_\lambda\mathfrak r_\lambda^{1/2}e^{2kb_\lambda}
 \Big[1+
 \|\gamma_{\nu_0}^{(2k)}\|_{\gS^2}
 +\mathfrak a_\lambda^{2k}
 (\mathfrak a_\lambda^2\mathfrak r_\lambda)^L
 \Big],\label{eq:quantum-final-HS}
\end{align}
where the last line uses $\big\|\gamma_{\nu_0}^{(\ell)}\big\|_{\gS^2}^2\leq\big\|\gamma_{\nu_0}^{(2k)}\big\|_{\gS^2}^{\ell/k}$ for $1\leq\ell\leq 2k$.

\noindent\textbf{Step 2: Classical weights.} For the second term in \eqref{eq:proof-HS-three-comparisons}, Lemma~\ref{lem:classical-entropy-HS} gives
$$
 \big\|\gamma_{f_\lambda}^{(k)}[u_\lambda+v_\lambda]
 -\gamma_{f_{\mathrm c}}^{(k)}[u_\lambda+v_\lambda]\big\|_{\gS^2}^2
 \leq
 2\cH_{\cl}(f_{\lambda}\dd\widehat{\mathbb P}_\lambda,f_{\mathrm c}\dd\widehat{\mathbb P}_\lambda)
 \big\|\gamma_{f_{\lambda}}^{(2k)}[u_\lambda+v_\lambda]+\gamma_{f_{\mathrm c}}^{(2k)}[u_\lambda+v_\lambda]\big\|_{\gS^2}.
$$
Proposition~\ref{prop:estimate-E-lambda-diamond} gives $\cH_{\cl}(f_{\lambda}\dd\widehat{\mathbb P}_\lambda,f_{\mathrm c}\dd\widehat{\mathbb P}_\lambda)\lesssim_{d,c_0}\mathfrak a_\lambda^2\mathfrak r_\lambda$. Using Proposition~\ref{prop:moment-propagation} again, we have
\begin{equation}
 \big\|\gamma_{f_\lambda}^{(k)}[u_\lambda+v_\lambda]
 -\gamma_{f_{\mathrm c}}^{(k)}[u_\lambda+v_\lambda]\big\|_{\gS^2}^2
 \lesssim_{d,c_0,k,L}
 \mathfrak a_\lambda^2\mathfrak r_\lambda
 \Big[1+
 \|\gamma_{\nu_0}^{(2k)}\|_{\gS^2}
 +\mathfrak a_\lambda^{2k}
 (\mathfrak a_\lambda^2\mathfrak r_\lambda)^L
 \Big].
 \label{eq:classical-final-HS}
\end{equation}
\noindent\textbf{Step 3: Change of field.} For the third term in \eqref{eq:proof-HS-three-comparisons}, Lemma~\ref{lem:weighted-Gaussian-addition} yields
$$
 \big\|\gamma_{f_\lambda}^{(k)}[u_\lambda+v_\lambda]
 -\gamma_{f_\lambda}^{(k)}[u_\lambda]\big\|_{\gS^2}
 \leq
 \sum_{j=1}^k\binom kj^2j!\,
 \|\gamma_{f_\lambda}^{(k-j)}[u_\lambda]\|_{\gS^2}
 \|\mathsf C_0-\mathsf C_\lambda\|_{\gS^2}^{j}.
$$
Lemma~\ref{lem:covariance-estimates} gives $ \|\mathsf C_0-\mathsf C_\lambda\|_{\gS^2}^2 \lesssim_d\lambda^2N_{\rm free}, $ which, together with Proposition~\ref{prop:moment-propagation}, gives
\begin{equation}
 \big\|\gamma_{f_\lambda}^{(k)}[u_\lambda+v_\lambda]
 -\gamma_{f_\lambda}^{(k)}[u_\lambda]\big\|_{\gS^2}^2
 \lesssim_{d,c_0,k,L}
 \lambda^2N_{\rm free}
 \Big[
 1+\|\gamma_{\nu_0}^{(2k)}\|_{\gS^2}
 +\mathfrak a_\lambda^{2k}
 (\mathfrak a_\lambda^2\mathfrak r_\lambda)^L
 \Big].\label{eq:field-change-HS}
\end{equation}
Combining\kern0pt{} \eqref{eq:quantum-final-HS}, \eqref{eq:classical-final-HS}, and \eqref{eq:field-change-HS} with \eqref{eq:proof-HS-three-comparisons} proves \eqref{eq:main-HS}.
\end{proof}

\bigskip

\section{\texorpdfstring{Derivation of the local $\Phi^4_d$ measures}{Derivation of the local Phi4d measures}}
\label{sec:derivation-Phi4d}

We prove Corollaries~\ref{cor:Phi42-local} and \ref{cor:Phi43-local-HS} by combining the quantum-to-Hartree limits with the classical approximations. Throughout this section, $\gH=L^2(\bT^d;\CC)$ and $h=\1-\Delta$. For a probability measure $\rho$ on $\gH^{-1}$, we interpret $\gamma_\rho^{(k)}$ by Definition~\ref{def:weighted-correlations} with $s=1$, the identity field on $(\gH^{-1},\mathcal B(\gH^{-1}),\rho)$, and scalar weight $1$.

In dimension two, Wick ordering is taken with respect to $\mu_0$. The local interaction, measure, and partition function are
\begin{align*}
 \mathcal V[u]
 =\frac12\int_{\bT^2}:|u(x)|^4:\dd x,\qquad
 \dd\nu_{\Phi^4_2}(u)
 =z_{\Phi^4_2}^{-1}e^{-\mathcal V[u]}\dd\mu_0(u),
 \qquad
 z_{\Phi^4_2}=\int e^{-\mathcal V[u]}\dd\mu_0(u).
\end{align*}
By \cite[Lemma~4.6]{FKSS23}, in our normalization,\footnote{Under the change of variables $\phi(y)=\sqrt2\,u(2\pi y)$, the covariance becomes $(2\pi^2-\Delta_y/2)^{-1}$ on the unit torus, and the interaction range becomes $\varepsilon/(2\pi)$. The Hartree and local energies in \cite[(3.1), (2.2)]{FKSS23} correspond to $\pi^{-2}\cW_{0,\varepsilon}$ and $\pi^{-2}\mathcal V$, respectively. Their Lemma~4.6, H\"older's inequality, and
$$
 |x^{\pi^2}-y^{\pi^2}|
 \leq\pi^2|x-y|(x^{\pi^2-1}+y^{\pi^2-1}),
 \qquad x,y\geq0,
$$
give \eqref{eq:Phi42-classical-weight-convergence}.}
\begin{equation}\label{eq:Phi42-classical-weight-convergence}
 e^{-\cW_{0,\varepsilon}}
 \longrightarrow e^{-\mathcal V}
 \quad\text{in }L^q(\mu_0),
 \qquad 1\leq q<\infty,
 \qquad \varepsilon\to0^+.
\end{equation}

\begin{proof}[Proof of Corollary~\ref{cor:Phi42-local}]
By \eqref{eq:Phi42-classical-weight-convergence},
\begin{align*}
 |z_{0,\varepsilon}-z_{\Phi^4_2}|
 &\leq
 \|e^{-\cW_{0,\varepsilon}}-e^{-\mathcal V}\|_{L^1(\mu_0)}
 \longrightarrow0,
 \qquad z_{\Phi^4_2}>0,\\
 \Big\|
 \frac{e^{-\cW_{0,\varepsilon}}}{z_{0,\varepsilon}}
 -\frac{e^{-\mathcal V}}{z_{\Phi^4_2}}
 \Big\|_{L^2(\mu_0)}
 &\leq
 \frac{\|e^{-\cW_{0,\varepsilon}}-e^{-\mathcal V}\|_{L^2(\mu_0)}}
 {z_{0,\varepsilon}}
 +
 \frac{|z_{0,\varepsilon}-z_{\Phi^4_2}|}
 {z_{0,\varepsilon}z_{\Phi^4_2}}
 \|e^{-\mathcal V}\|_{L^2(\mu_0)}
 \longrightarrow0.
\end{align*}
Lemma~\ref{lem:Gaussian-weighted-correlations}, applied on $(\gH^{-1},\mathcal B(\gH^{-1}),\mu_0)$ with the identity Gaussian field and $q=2$, gives the Hilbert--Schmidt realizations of the correlations. Applying \eqref{eq:Gaussian-weighted-HS} to the difference of the normalized densities yields
$$\stepcounter{equation}
 \begin{gathered}
  \big\|\gamma_{\nu_\varepsilon}^{(k)}
  -\gamma_{\nu_{\Phi^4_2}}^{(k)}\big\|_{\gS^2}
  \leq
  (4k)^k\|h^{-1}\|_{\gS^2}^k
  \Big\|
  \frac{e^{-\cW_{0,\varepsilon}}}{z_{0,\varepsilon}}
  -\frac{e^{-\mathcal V}}{z_{\Phi^4_2}}
  \Big\|_{L^2(\mu_0)}
  \longrightarrow0.\quad
 \end{gathered}
$$
For $\varepsilon=\lambda^\eta$, $0<\eta<1/2$, Theorem~\ref{thm:Phi42} and these classical limits give
\begin{align*}
 \Big|\log\frac{Z_{\lambda,\lambda^\eta}}{Z_{\rm free}}
 -\log z_{\Phi^4_2}\Big|
 \leq
 \Big|\log\frac{Z_{\lambda,\lambda^\eta}}{Z_{\rm free}}
 -\log z_{0,\lambda^\eta}\Big|
 +
 |\log z_{0,\lambda^\eta}-\log z_{\Phi^4_2}|
 \longrightarrow0,
\end{align*}
and
\begin{align*}
 \big\|k!\lambda^k\Gamma_{\lambda,\lambda^\eta}^{(k)}
 -\gamma_{\nu_{\Phi^4_2}}^{(k)}\big\|_{\gS^2}
 \leq
 \big\|k!\lambda^k\Gamma_{\lambda,\lambda^\eta}^{(k)}
 -\gamma_{\nu_{\lambda^\eta}}^{(k)}\big\|_{\gS^2}
 +
 \big\|\gamma_{\nu_{\lambda^\eta}}^{(k)}
 -\gamma_{\nu_{\Phi^4_2}}^{(k)}\big\|_{\gS^2}
 \longrightarrow0.
\end{align*}
\end{proof}

\medskip

\begingroup\color{red}
We use the bounds \eqref{eq:Phi43-local-weighted-HS} in place of the preceding density comparison.

\begin{proof}[Proof of Corollary~\ref{cor:Phi43-local-HS}]
Fix $k\geq1$ and take $\delta=1/(16k)$. Proposition~\ref{prop:classical-Hartree-correlations} gives the unique positive Hilbert--Schmidt realizations of $\gamma_{\nu_\varepsilon}^{(k)}$ and $(h^{\delta/2})^{\otimes k}\gamma_{\nu_\varepsilon}^{(k)} (h^{\delta/2})^{\otimes k}$ for $0\leq\varepsilon\leq1$, where $\varepsilon=0$ refers to $\nu_{\Phi^4_3}$. For finite Fourier test tensors, \eqref{eq:intrinsic-correlation-form} gives
\begin{align*}
 \gamma_{\nu_\varepsilon}^{(k)}
 =
 (h^{-\delta/2})^{\otimes k}
 \big[(h^{\delta/2})^{\otimes k}
 \gamma_{\nu_\varepsilon}^{(k)}
 (h^{\delta/2})^{\otimes k}\big]
 (h^{-\delta/2})^{\otimes k}.
\end{align*}
Both sides are bounded, so the identity holds on $\bigotimes_{\mathrm{sym}}^k\gH$. In the Fourier basis,
\begin{align*}
 \big\|(\1^{\otimes k}-P_K^{\otimes k})
 (h^{-\delta/2})^{\otimes k}\big\|
 =
 \sup_{\substack{n_1,\ldots,n_k\in\Z^3\\
 \max_j h(n_j)>K}}
 \prod_{j=1}^k h(n_j)^{-\delta/2}
 \leq K^{-\delta/2},\qquad
 \big\|(h^{-\delta/2})^{\otimes k}\big\|\leq1.
\end{align*}
Thus the preceding factorization and \eqref{eq:Phi43-local-weighted-HS} yield
\begin{align}
 \big\|\gamma_{\nu_\varepsilon}^{(k)}
 -P_K^{\otimes k}\gamma_{\nu_\varepsilon}^{(k)}
 P_K^{\otimes k}\big\|_{\gS^2}
 &\leq
 \big\|(\1^{\otimes k}-P_K^{\otimes k})
 \gamma_{\nu_\varepsilon}^{(k)}\big\|_{\gS^2}
 +
 \big\|P_K^{\otimes k}\gamma_{\nu_\varepsilon}^{(k)}
 (\1^{\otimes k}-P_K^{\otimes k})\big\|_{\gS^2}\nn\\
 &\leq
 2K^{-\delta/2}
 \big\|(h^{\delta/2})^{\otimes k}
 \gamma_{\nu_\varepsilon}^{(k)}
 (h^{\delta/2})^{\otimes k}\big\|_{\gS^2}
 \leq C_{k,v,m}K^{-\delta/2}.
 \label{eq:Phi43-local-uniform-tail}
\end{align}
For fixed $K$, \eqref{eq:correlation-compression} and the distributional convergence of correlation functions in \cite[Theorem~2.6]{NZZ25} give
\begin{align*}
 \big\|P_K^{\otimes k}
 \big(\gamma_{\nu_\varepsilon}^{(k)}
 -\gamma_{\nu_{\Phi^4_3}}^{(k)}\big)
 P_K^{\otimes k}\big\|_{\gS^2}^2
 =
 \sum_{\Phi_i,\Phi_j\in\bigotimes_{\mathrm{sym}}^k(P_K\gH)}
 \big|\langle\Phi_i,
 \big(\gamma_{\nu_\varepsilon}^{(k)}
 -\gamma_{\nu_{\Phi^4_3}}^{(k)}\big)
 \Phi_j\rangle\big|^2
 \longrightarrow0.
\end{align*}
where the sum runs over the orthonormal symmetric Fourier basis of $\bigotimes_{\mathrm{sym}}^k(P_K\gH)$. Applying \eqref{eq:Phi43-local-uniform-tail} both to $\nu_\varepsilon$ and to $\nu_{\Phi^4_3}$, we obtain
\begin{align*}
 \big\|\gamma_{\nu_\varepsilon}^{(k)}
 -\gamma_{\nu_{\Phi^4_3}}^{(k)}\big\|_{\gS^2}
 &\leq
 \big\|\gamma_{\nu_\varepsilon}^{(k)}
 -P_K^{\otimes k}\gamma_{\nu_\varepsilon}^{(k)}
 P_K^{\otimes k}\big\|_{\gS^2}
 +
 \big\|P_K^{\otimes k}
 \big(\gamma_{\nu_\varepsilon}^{(k)}
 -\gamma_{\nu_{\Phi^4_3}}^{(k)}\big)
 P_K^{\otimes k}\big\|_{\gS^2}
 +
 \big\|\gamma_{\nu_{\Phi^4_3}}^{(k)}
 -P_K^{\otimes k}\gamma_{\nu_{\Phi^4_3}}^{(k)}
 P_K^{\otimes k}\big\|_{\gS^2}\\
 &\leq2C_{k,v,m}K^{-\delta/2}
 +
 \big\|P_K^{\otimes k}
 \big(\gamma_{\nu_\varepsilon}^{(k)}
 -\gamma_{\nu_{\Phi^4_3}}^{(k)}\big)
 P_K^{\otimes k}\big\|_{\gS^2}.
\end{align*}
Taking $\eps\downarrow 0$ first, and then taking $K\uparrow \infty$ proves \eqref{eq:Phi43-local-classical-HS}.
\end{proof}
\endgroup

\appendix

\bigskip

\section{Mass renormalization and covariance estimates}

\begin{lemma}[Mass renormalization] \label{lem:Gaussian-centered-quadratic-moments} Let $d\in\{2,3\}$, $\lambda\geq0$, and $P_K=\mathds1_{\{h\leq K\}}$. For every bounded operator $A$ on $\gH$, the limit
$$
 \cM_\lambda^A[u]
 =
 \lim_{K\to\infty}
 \Big(
 \langle P_Ku,AP_Ku\rangle
 -\tr_{\gH}(P_KAP_K\mathsf C_\lambda)
 \Big)
$$
exists in $L^p(\mu_\lambda)$ for every $1\leq p<\infty$. In particular, $\cM_{0,K}^A\to\cM_0^A$ in $L^p(\mu_0)$ for the cutoffs in \eqref{eq:classical-mass-cutoff}.

The map $A\mapsto\cM_\lambda^A$ is complex linear, $\int\cM_\lambda^A\dd\mu_\lambda=0$, and $\cM_\lambda^{A^*}=\overline{\cM_\lambda^A}$. Moreover,
\begin{equation}
 \big\|\cM_\lambda^A\big\|_{L^2(\mu_\lambda)}
 =
 \big\|\mathsf C_\lambda^{1/2}A
 \mathsf C_\lambda^{1/2}\big\|_{\gS^2},
 \qquad
 \big\|\cM_\lambda^A\big\|_{L^p(\mu_\lambda)}
 \lesssim
 p\big\|\mathsf C_\lambda^{1/2}A
 \mathsf C_\lambda^{1/2}\big\|_{\gS^2},
 \qquad 1\leq p<\infty.
 \label{eq:Gaussian-mass-limit-Lp}
\end{equation}
Complex linearity and these estimates uniquely specify the masses through the rule
$$
 \cM_\lambda^{|\phi\rangle\langle\psi|}[u]
 =
 \langle u,\phi\rangle\langle\psi,u\rangle
 -
 \langle\psi,\mathsf C_\lambda\phi\rangle,
 \qquad \phi,\psi\in\gH^1,
$$
where the pairings are understood by duality. For $\lambda>0$,
$$
 \cM_\lambda^A[u]
 =
 \langle u,Au\rangle-\tr_{\gH}(A\mathsf C_\lambda)
 \qquad\text{for }\mu_\lambda\text{-almost every }u.
$$

If $u,v$ are jointly Gaussian fields on a fixed probability space $(\gX,\mathcal A,\mathbb P)$, with laws $\mu_\lambda,\mu_{\lambda'}$, respectively, where $\lambda,\lambda'\geq0$, then, for bounded $A,B$,
\begin{equation}
 \big\|\cM_\lambda^A[u]-\cM_{\lambda'}^B[v]
 \big\|_{L^p(\mathbb P)}
 \lesssim
 p\big\|\cM_\lambda^A[u]-\cM_{\lambda'}^B[v]
 \big\|_{L^2(\mathbb P)},
 \qquad 1\leq p<\infty.
 \label{eq:Gaussian-joint-quadratic-moment}
\end{equation}
\end{lemma}

\begin{proof}
For finite-rank operators generated by vectors in $\gH^1$, use the rank-one formula and complex linearity. Applying \eqref{eq:Gaussian-field-correlations} at orders one and two gives
$$
 \int\cM_\lambda^A\dd\mu_\lambda
 =
 \tr_{\gH}(A\mathsf C_\lambda)
 -\tr_{\gH}(A\mathsf C_\lambda)=0,
$$
and
$$
 \int\cM_\lambda^A\overline{\cM_\lambda^B}\dd\mu_\lambda
 =
 2\tr_{\gH^{\otimes2}}
 \big[(A\otimes B^*)\Pi_2\mathsf C_\lambda^{\otimes2}\big]
 -
 \tr_{\gH}(A\mathsf C_\lambda)
 \tr_{\gH}(B^*\mathsf C_\lambda)
 =
 \tr_{\gH}(A\mathsf C_\lambda B^*\mathsf C_\lambda).
$$
Taking $B=A$ gives
$$
 \|\cM_\lambda^A\|_{L^2(\mu_\lambda)}^2
 =
 \|\mathsf C_\lambda^{1/2}A
 \mathsf C_\lambda^{1/2}\|_{\gS^2}^2,
$$
which also makes the definition independent of the rank-one decomposition.

For a finite Gaussian vector of zero mean on $(\gX,\mathcal A,\mathbb P)$, diagonalization represents each centered real quadratic form as $\sum_j a_j(X_j^2-1)$ in law, where the $X_j$ are independent standard real Gaussian variables and $a_j\in\R$. For $\sum_j a_j^2=1$, Gaussian integration gives
$$
 \log\int
 \exp\Big(\pm\frac14\sum_j a_j(X_j^2-1)\Big)\dd\mathbb P
 =
 \frac12\sum_j
 \Big[\mp\frac{a_j}{2}
 -\log\Big(1\mp\frac{a_j}{2}\Big)\Big]
 \leq\frac18\sum_j a_j^2=\frac18,
$$
where we used $-x-\log(1-x)\leq x^2$ for $|x|\leq1/2$. Using $x^p\leq(4p/\mathrm e)^p e^{x/4}$, $x\geq0$, we obtain
\begin{align*}
 \Big\|\sum_j a_j(X_j^2-1)\Big\|_{L^p(\mathbb P)}
 &\leq
 \frac{4p}{\mathrm e}
 \Big(
 \int\exp\Big(\frac14\Big|\sum_j a_j(X_j^2-1)\Big|\Big)
 \dd\mathbb P
 \Big)^{1/p}\\
 &\leq
 \frac{4p}{\mathrm e}(2e^{1/8})^{1/p}
 \lesssim
 p\Big\|\sum_j a_j(X_j^2-1)\Big\|_{L^2(\mathbb P)},
\end{align*}
since the squared $L^2$ norm is $2\sum_j a_j^2=2$. Homogeneity gives this estimate for all real coefficients. Applying it to the real and imaginary parts of the finite-rank mass yields
\begin{align*}
 \|\cM_\lambda^A\|_{L^p(\mu_\lambda)}
 &\leq
 \|\Re\cM_\lambda^A\|_{L^p(\mu_\lambda)}
 +\|\Im\cM_\lambda^A\|_{L^p(\mu_\lambda)}
 \lesssim
 p\big(
 \|\Re\cM_\lambda^A\|_{L^2(\mu_\lambda)}
 +\|\Im\cM_\lambda^A\|_{L^2(\mu_\lambda)}
 \big)\\
 &\lesssim
 p\|\cM_\lambda^A\|_{L^2(\mu_\lambda)}
 =
 p\|\mathsf C_\lambda^{1/2}A
 \mathsf C_\lambda^{1/2}\|_{\gS^2}.
\end{align*}

For bounded $A$, H\"older's inequality for Schatten norms gives
$$
 \sum_{m,n\in\Z^d}
 \frac{|\langle e_m,Ae_n\rangle|^2}{h(m)h(n)}
 =
 \|h^{-1/2}Ah^{-1/2}\|_{\gS^2}^2
 \leq
 \|A\|^2\tr_{\gH}h^{-2}<\infty.
$$
Since $P_K$ and $\mathsf C_\lambda$ are diagonal in the Fourier basis and $0\leq\mathsf C_\lambda\leq h^{-1}$,
$$
 \big\|\mathsf C_\lambda^{1/2}
 (A-P_KAP_K)\mathsf C_\lambda^{1/2}\big\|_{\gS^2}^2
 \leq
 \sum_{\substack{m,n\in\Z^d\\ \max\{h(m),h(n)\}>K}}
 \frac{|\langle e_m,Ae_n\rangle|^2}{h(m)h(n)}
 \longrightarrow0.
$$
Consequently, the finite-rank estimate gives
\begin{align*}
 &\big\|\cM_\lambda^{P_KAP_K}
 -\cM_\lambda^{P_JAP_J}\big\|_{L^p(\mu_\lambda)}
 \lesssim
 p\big\|\mathsf C_\lambda^{1/2}
 (P_KAP_K-P_JAP_J)\mathsf C_\lambda^{1/2}\big\|_{\gS^2}\\
 &\quad\leq
 p\big\|\mathsf C_\lambda^{1/2}
 (A-P_KAP_K)\mathsf C_\lambda^{1/2}\big\|_{\gS^2}
 +p\big\|\mathsf C_\lambda^{1/2}
 (A-P_JAP_J)\mathsf C_\lambda^{1/2}\big\|_{\gS^2}
 \longrightarrow0.
\end{align*}
The sequence converges in every finite $L^p(\mu_\lambda)$. These limits agree in $L^1(\mu_\lambda)$ and extend the finite-rank definition to $\cM_\lambda^A$. The rank-one formula gives
$$
 \cM_\lambda^{P_KAP_K}[u]
 =
 \langle P_Ku,AP_Ku\rangle
 -\tr_{\gH}(P_KAP_K\mathsf C_\lambda).
$$
For $\lambda=0$, this equals $\cM_{0,K}^A[u]$ in \eqref{eq:classical-mass-cutoff}. Moreover,
\begin{align*}
 \big\|\cM_\lambda^A
 -\cM_\lambda^{P_KAP_K}\big\|_{L^p(\mu_\lambda)}
 &\lesssim
 p\big\|\mathsf C_\lambda^{1/2}
 (A-P_KAP_K)\mathsf C_\lambda^{1/2}\big\|_{\gS^2}
 \longrightarrow0,\\
 \|\cM_\lambda^A\|_{L^2(\mu_\lambda)}
 &=
 \lim_{K\to\infty}
 \|\mathsf C_\lambda^{1/2}P_KAP_K
 \mathsf C_\lambda^{1/2}\|_{\gS^2}
 =
 \|\mathsf C_\lambda^{1/2}A
 \mathsf C_\lambda^{1/2}\|_{\gS^2}.
\end{align*}
The finite-rank $L^p$ bound gives \eqref{eq:Gaussian-mass-limit-Lp} by the same convergence. Zero mean, complex linearity and the adjoint relation pass to these $L^p$ limits. The displayed approximation estimate also proves uniqueness from the rank-one rule.

For $\lambda>0$, Gaussian pairing gives
$$
 \int\|u\|^4\dd\mu_\lambda(u)
 =
 (\tr_{\gH}\mathsf C_\lambda)^2
 +\tr_{\gH}\mathsf C_\lambda^2<\infty.
$$
Hence
\begin{align*}
 &\Big\|
 \cM_\lambda^{P_KAP_K}
 -\big(\langle u,Au\rangle
 -\tr_{\gH}(A\mathsf C_\lambda)\big)
 \Big\|_{L^2(\mu_\lambda)}\\
 &\quad\leq
 2\|A\|
 \Big(
 \int\|u\|^2\|(1-P_K)u\|^2\dd\mu_\lambda(u)
 \Big)^{1/2}
 +\|A\|\tr_{\gH}((1-P_K)\mathsf C_\lambda)
 \longrightarrow0,
\end{align*}
by dominated convergence and $\tr_{\gH}\mathsf C_\lambda<\infty$. This identifies the limit with the stated centered quadratic form.

Finally, for finite-rank $A,B$ generated by $\gH^1$, the real and imaginary parts of $\cM_\lambda^A[u]-\cM_{\lambda'}^B[v]$ are centered real quadratic forms in a finite jointly Gaussian vector. The quadratic-form estimate above gives \eqref{eq:Gaussian-joint-quadratic-moment} for these operators. For bounded $A,B$, the marginal laws and \eqref{eq:Gaussian-mass-limit-Lp} give
\begin{align*}
 &\Big\|
 \cM_\lambda^A[u]-\cM_{\lambda'}^B[v]
 -\cM_\lambda^{P_KAP_K}[u]
 +\cM_{\lambda'}^{P_KBP_K}[v]
 \Big\|_{L^p(\mathbb P)}
 \leq
 \|\cM_\lambda^{A-P_KAP_K}\|_{L^p(\mu_\lambda)}
 +\|\cM_{\lambda'}^{B-P_KBP_K}\|_{L^p(\mu_{\lambda'})}\\
 &\quad\lesssim
 p\big\|\mathsf C_\lambda^{1/2}
 (A-P_KAP_K)\mathsf C_\lambda^{1/2}\big\|_{\gS^2}
 +
 p\big\|\mathsf C_{\lambda'}^{1/2}
 (B-P_KBP_K)\mathsf C_{\lambda'}^{1/2}\big\|_{\gS^2}
 \longrightarrow0.
\end{align*}
Applying this convergence in $L^p(\mathbb P)$ and $L^2(\mathbb P)$ to the finite-rank inequality proves \eqref{eq:Gaussian-joint-quadratic-moment}.
\end{proof}

\begin{lemma}\label{lem:covariance-estimates} For $0<\lambda\leq1$, $d\in\{2,3\}$, and $c_0(n)=h(n)^{-1}$,
\begin{equation}\label{eq:shifted-covariance-defect}
 \sup_{\ell\in\Z^d}
 \sum_{n\in\Z^d}
 \big[c_0(n)c_0(n+\ell)-c_\lambda(n)c_\lambda(n+\ell)\big]
 \lesssim_d\lambda^2N_{\rm free}.
\end{equation}
Consequently,
\begin{equation}\label{eq:covariance-HS-bound}
 \|\mathsf C_0-\mathsf C_\lambda\|_{\gS^2}^2
 \lesssim_d\lambda^2N_{\rm free},
\end{equation}
and every bounded multiplication operator $A$ satisfies
\begin{equation}
 \|\mathsf C_0^{1/2}A\mathsf C_0^{1/2}\|_{\gS^2}^2
 -\|\mathsf C_\lambda^{1/2}A\mathsf C_\lambda^{1/2}\|_{\gS^2}^2
 \lesssim_d
 \lambda^2N_{\rm free}\|A\|_{L^\infty(\bT^d)}^2.
 \label{eq:covariance-multiplication-bound}
\end{equation}
\end{lemma}

\begin{proof}
The scalar inequality $0\leq t^{-1}-(e^t-1)^{-1}\leq\min\{1,t^{-1}\}$ gives
$$
 0\leq c_0(n)-c_\lambda(n)
 \leq\min\{\lambda,h(n)^{-1}\}
 \leq\frac2{\lambda^{-1}+h(n)}.
$$
Hence
\begin{align*}
 &\sum_n\big[c_0(n)c_0(n+\ell)-c_\lambda(n)c_\lambda(n+\ell)\big]
 \leq
 4\sum_n\frac1{(\lambda^{-1}+h(n))h(n+\ell)}\\
 &=
 4\int_0^1
 \sum_n\frac{\dd t}{
 [1+t\lambda^{-1}+|n+(1-t)\ell|^2+t(1-t)|\ell|^2]^2}\\
 &\lesssim_d
 \int_0^1(1+t\lambda^{-1})^{d/2-2}\dd t
 \lesssim_d
 \begin{cases}
  \lambda\log(2/\lambda),&d=2,\\
  \lambda^{1/2},&d=3.
 \end{cases}
\end{align*}
For the sum in the third line, integration over unit cubes gives
$$
 \sum_n(M+|n+y|^2)^{-2}
 \lesssim_d
 \int_{\R^d}(M+|x|^2)^{-2}\dd x
 \lesssim_d M^{d/2-2},
 \qquad M\geq1,\quad y\in\R^d.
$$
By \cite[Lemma~B.1]{LNR21},
$$
 \lambda^2N_{\rm free}\asymp_d
 \begin{cases}
 \lambda\log(2/\lambda),&d=2,\\
 \lambda^{1/2},&d=3.
 \end{cases}
$$
which proves \eqref{eq:shifted-covariance-defect}. The pointwise bound $(c_0(n)-c_\lambda(n))^2\leq c_0(n)^2-c_\lambda(n)^2$ proves \eqref{eq:covariance-HS-bound}.

For a multiplication operator, Fourier expansion and Parseval give
\begin{align*}
 \|\mathsf C_0^{1/2}A\mathsf C_0^{1/2}\|_{\gS^2}^2
 -\|\mathsf C_\lambda^{1/2}A\mathsf C_\lambda^{1/2}\|_{\gS^2}^2
 &=
 \sum_\ell|\langle e_\ell,Ae_0\rangle|^2
 \sum_n\big[c_0(n)c_0(n+\ell)-c_\lambda(n)c_\lambda(n+\ell)\big]\\
 &\lesssim_d
 \lambda^2N_{\rm free}\|Ae_0\|^2
 \leq\lambda^2N_{\rm free}\|A\|_\infty^2.
\end{align*}
All sums are absolutely convergent, since $\mathsf C_0\in\gS^2$ and $A$ is bounded.
\end{proof}

\bigskip

\section{Berezin--Lieb inequalities for the relative entropy}

The following result is from \cite[Lemma~B.1]{Cub26}.

\medskip

\begin{lemma}[Berezin--Lieb inequalities for the relative entropy]
\label{lem:relative-entropy-Berezin-Lieb}
Let $\mathfrak K$ be a separable Hilbert space, let $M$ be a separable metric space, and let
$$
 M\ni x\longmapsto n_x\in S\mathfrak K
 :=
 \{u\in\mathfrak K:\norm u_{\mathfrak K}=1\}
$$
be continuous.  For every Borel probability measure $\mu$ on $M$, let
$$
 A_\mu
 :=
 \int_M
 \ket{n_x}\bra{n_x}\,\dd\mu(x)
 \in\gS^1(\mathfrak K).
$$
Then, for all Borel probability measures $\mu,\mu'$ on $M$,
$$\stepcounter{equation}
 \cH(A_\mu,A_{\mu'})
 \leq
 \cH_{\cl}(\mu,\mu').
$$
Assume in addition that $\zeta$ is a positive Borel measure on $M$ satisfying
\begin{equation}
 \int_M
 \ket{n_x}\bra{n_x}\,\dd\zeta(x)
 =
 \1_{\mathfrak K}
 \label{eq:relative-entropy-resolution-identity}
\end{equation}
in the weak operator sense.  For every state $A$ on $\mathfrak K$, define its Husimi measure, or lower symbol, by $ \dd m_A(x) := \bangle{n_x,A n_x}\,\dd\zeta(x). $ Then, for all states $A,B$ on $\mathfrak K$,
$$\stepcounter{equation}
 \cH_{\cl}(m_A,m_B)
 \leq
 \cH(A,B).
$$
Consequently, under \eqref{eq:relative-entropy-resolution-identity},
$$\stepcounter{equation}
 \cH_{\cl}(m_{A_\mu},m_{A_{\mu'}})
 \leq
 \cH(A_\mu,A_{\mu'})
 \leq
 \cH_{\cl}(\mu,\mu').
$$
\end{lemma}

\bigskip

\begingroup\color{red}
\section{Classical Hartree correlations}
\label{app:classical-Hartree-correlations}

We use $\gH=L^2(\bT^d;\CC)$, $h=\1-\Delta$, and $P_K=\1_{\{h\leq K\}}$ as in Subsection~\ref{subsec:classical-field-theory}. The measures $\nu_\varepsilon$ are those of Theorems~\ref{thm:Phi42} and \ref{thm:Phi43}. All correlation forms are interpreted as in Definition~\ref{def:weighted-correlations}, with the identity field on $(\gH^{-1},\mathcal B(\gH^{-1}),\rho)$ and scalar weight $1$, where $\rho$ is the measure appearing in $\gamma_\rho^{(k)}$.

\begin{proposition}
Under Assumption~\ref{ass:Phi42}, there exists $\varepsilon_0=\varepsilon_0(v)\in(0,1]$ such that
$$\stepcounter{equation}
 \sup_{0<\varepsilon\leq\varepsilon_0}
 \|\gamma_{\nu_\varepsilon}^{(k)}\|_{\gS^2}
 \leq C_{k,v},
 \qquad k\geq1.
$$
\end{proposition}

\begin{proof}
By \cite[Lemma~4.6]{FKSS23}, $z_{0,\varepsilon}\to z_{\Phi^4_2}>0$ and $e^{-\cW_{0,\varepsilon}}$ is bounded in $L^2(\mu_0)$ for sufficiently small $\varepsilon$. Thus \eqref{eq:Gaussian-weighted-HS} with $q=2$ gives, for some $\varepsilon_0=\varepsilon_0(v)\in(0,1]$,
\begin{align*}
 \sup_{0<\varepsilon\leq\varepsilon_0}
 \|\gamma_{\nu_\varepsilon}^{(k)}\|_{\gS^2}
 \leq
 (4k)^k\|h^{-1}\|_{\gS^2}^k
 \sup_{0<\varepsilon\leq\varepsilon_0}
 \frac{\|e^{-\cW_{0,\varepsilon}}\|_{L^2(\mu_0)}}
 {z_{0,\varepsilon}}
 \leq C_{k,v}.
\end{align*}
\end{proof}

\begin{proposition}\label{prop:classical-Hartree-correlations}
For every $k\geq1$ and $0\leq\delta\leq1/(16k)$, the correlation forms have unique positive Hilbert--Schmidt realizations satisfying
\begin{equation}\label{eq:Phi43-local-weighted-HS}
 \sup_{0\leq\varepsilon\leq1}
 \big\|(h^{\delta/2})^{\otimes k}
 \gamma_{\nu_\varepsilon}^{(k)}
 (h^{\delta/2})^{\otimes k}\big\|_{\gS^2}
 \leq C_{k,v,m}.
\end{equation}
\end{proposition}

\begin{proof}
Let $(\Psi^\varepsilon,Z)$ be the jointly stationary process of \cite[Lemma~4.17]{NZZ25}, and set $\pi_\varepsilon=\Law(\Psi^\varepsilon(1),Z(1))$. Its marginals are $\nu_\varepsilon,\mu_0$. We write $\mathbb E$ for expectation under the process law and use $\Psi^\varepsilon-Z=\psi_l+\psi_h+Y$ from \cite[(3.34), (4.50)]{NZZ25}. By \cite[p.~53, proof of Theorem~4.1]{NZZ25} and the estimates for $Y$ in \cite[Lemma~3.3, Proposition~3.4, Tables~1--2, and (3.64)]{NZZ25}, for $p\geq2$,
\begin{align*}
 \int\|u-z\|_{\gH}^p\dd\pi_\varepsilon
 &=\mathbb E\|\psi_l(1)+\psi_h(1)+Y(1)\|_{\gH}^p\\
 &\lesssim_p
 \mathbb E\big(
 \|\psi_l(1)\|_{\gH}^p+\|\psi_h(1)\|_{\gH}^p
 +\|Y(1)\|_{\mathcal C^{1/4}}^p\big)
 \leq C_{p,v,m}.
\end{align*}
For sufficiently small $\kappa>0$, we have the embeddings $B_{4,\infty}^{1-2\kappa},\mathcal C^{1/4} \hookrightarrow\gH^{1/8}$. Joint stationarity, monotone convergence, and \cite[(4.52) and the subsequent estimate for $\psi_h$, p.~53]{NZZ25} therefore yield
\begin{align*}
 \int\|u-z\|_{\gH^{1/8}}^2\dd\pi_\varepsilon
 &=\lim_{N\to\infty}
 \mathbb E\|P_N(\Psi^\varepsilon(1)-Z(1))\|_{\gH^{1/8}}^2
 =\lim_{N\to\infty}\mathbb E\int_0^1
 \|P_N(\Psi^\varepsilon(t)-Z(t))\|_{\gH^{1/8}}^2\dd t\\
 &\lesssim\mathbb E\int_0^1
 \big(\|\psi_l(t)\|_{\gH^1}^2+
 \|\psi_h(t)\|_{B_{4,\infty}^{1-2\kappa}}^2+
 \|Y(t)\|_{\mathcal C^{1/4}}^2\big)\dd t
 \leq C_{v,m}.
\end{align*}
Fix $k\geq1$ and $0\leq\delta\leq1/(16k)$. Sobolev interpolation and Cauchy--Schwarz give
\begin{align}
 \int\|u-z\|_{\gH^\delta}^{2k}\dd\pi_\varepsilon
 &\leq
 \int\|u-z\|_{\gH^{1/8}}\|u-z\|_{\gH}^{2k-1}
 \dd\pi_\varepsilon\nn\\
 &\leq
 \left(\int\|u-z\|_{\gH^{1/8}}^2\dd\pi_\varepsilon\right)^{1/2}
 \left(\int\|u-z\|_{\gH}^{4k-2}\dd\pi_\varepsilon\right)^{1/2}
 \leq C_{k,v,m}.
 \label{eq:Phi43-local-regular-remainder}
\end{align}
For $\phi\in\gH^{1+\delta}$, Gaussian integration therefore gives
\begin{align*}
 \int|\langle\phi,h^{\delta/2}u\rangle|^{2k}\dd\nu_\varepsilon(u)
 \lesssim_k
 \langle\phi,h^{\delta-1}\phi\rangle^k+
 \|\phi\|_{\gH}^{2k}
 \int\|u-z\|_{\gH^\delta}^{2k}\dd\pi_\varepsilon
 \leq C_{k,v,m}\|\phi\|_{\gH}^{2k}.
\end{align*}
This verifies \eqref{eq:correlation-integrability} for $h^{\delta/2}u$ with $s=1+\delta$.

Take two independent samples $(u,z),(u',z')$ of $\pi_\varepsilon$, and write $\mathbb E=\int\dd(\pi_\varepsilon\otimes\pi_\varepsilon)$ in the following estimates. Since $z$ is independent of $(u',z')$ and $h^{\delta-1}\leq1$,
\begin{align*}
 {}\mathbb E|\langle P_Kh^{\delta/2}z,
 P_Kh^{\delta/2}(u'-z')\rangle|^{2k}
 =
 k!\mathbb E\langle P_Kh^{\delta/2}(u'-z'),
 h^{\delta-1}P_Kh^{\delta/2}(u'-z')\rangle^k
 \leq k!\int\|u-z\|_{\gH^\delta}^{2k}\dd\pi_\varepsilon.
\end{align*}
The other mixed term has the same bound. For the purely Gaussian term we use \eqref{eq:gamma-1-k-identity}, and for the remainder term we use Cauchy--Schwarz and independence. The rank-one identity and the $L^{2k}$ triangle inequality give
\begin{align*}
 \big\|(P_Kh^{\delta/2})^{\otimes k}
 \gamma_{\nu_\varepsilon}^{(k)}
 (h^{\delta/2}P_K)^{\otimes k}\big\|_{\gS^2}^{1/k}
 &=
 \left(\mathbb E|\langle P_Kh^{\delta/2}u,
 P_Kh^{\delta/2}u'\rangle|^{2k}
 \right)^{1/(2k)}\\
 &\lesssim_k
 \|h^{\delta-1}\|_{\gS^2}
 +
 \left(\int\|u-z\|_{\gH^\delta}^{2k}\dd\pi_\varepsilon
 \right)^{1/(2k)}
 +
 \left(\int\|u-z\|_{\gH^\delta}^{2k}\dd\pi_\varepsilon
 \right)^{1/k}\\
 &\leq C_{k,v,m},
\end{align*}
where \eqref{eq:Phi43-local-regular-remainder} was used.

The finite-dimensional operators here are the correlation integrals of $P_Kh^{\delta/2}u$ against $\nu_\varepsilon$. Lemma~\ref{lem:correlation-realization-approximation} gives their unique positive Hilbert--Schmidt realization and
$$
 \sup_{0<\varepsilon\leq1}
 \big\|(h^{\delta/2})^{\otimes k}
 \gamma_{\nu_\varepsilon}^{(k)}
 (h^{\delta/2})^{\otimes k}\big\|_{\gS^2}
 \leq C_{k,v,m}.
$$

By \cite[Theorem~2.6]{NZZ25}, $\nu_\varepsilon\rightharpoonup\nu_{\Phi^4_3}$ on $\gH^{-1}$, and the product measures converge weakly as well. For $\phi\in\gH^{1+\delta}$ and fixed $K$, lower semicontinuity gives
\begin{align*}
 {}\int|\langle\phi,h^{\delta/2}u\rangle|^{2k}
 \dd\nu_{\Phi^4_3}(u)
 \leq
 \liminf_{\varepsilon\to0}
 \int|\langle\phi,h^{\delta/2}u\rangle|^{2k}\dd\nu_\varepsilon(u)
 \leq C_{k,v,m}\|\phi\|_{\gH}^{2k},
\end{align*}
which verifies \eqref{eq:correlation-integrability} under $\nu_{\Phi^4_3}$ with $s=1+\delta$. Moreover,
\begin{align*}
 \big\|(P_Kh^{\delta/2})^{\otimes k}
 \gamma_{\nu_{\Phi^4_3}}^{(k)}
 (h^{\delta/2}P_K)^{\otimes k}\big\|_{\gS^2}^{2}
 &=
 \iint|\langle P_Kh^{\delta/2}u,
 P_Kh^{\delta/2}u'\rangle|^{2k}
 \dd\nu_{\Phi^4_3}(u)\dd\nu_{\Phi^4_3}(u')\\
 &\quad\leq
 \liminf_{\varepsilon\to0}
 \big\|(P_Kh^{\delta/2})^{\otimes k}
 \gamma_{\nu_\varepsilon}^{(k)}
 (h^{\delta/2}P_K)^{\otimes k}\big\|_{\gS^2}^{2}
 \leq C_{k,v,m}.
\end{align*}
Lemma~\ref{lem:correlation-realization-approximation} gives the same Hilbert--Schmidt bound at $\varepsilon=0$, proving \eqref{eq:Phi43-local-weighted-HS}.
\end{proof}

\endgroup


\end{document}